\documentclass[a4paper,11pt]{article}
\usepackage{archive}
\title{A general model-checking procedure for \\ semiparametric accelerated failure time models}

\author{
    Dongrak Choi\footnotemark[1] \\
    Department of Biostatistics and Bioinformatics \\
    Duke University \\
    Durham, NC, USA \\
    \texttt{dongrak.choi@duke.edu}
    \And
    Woojung Bae\thanks{Contribute equally to this study} \\
    Department of Statistics\\
    University of Florida\\
    Gainesville, FL, USA \\
    \texttt{woojung.bae@ufl.edu} 
    \And
    Jun Yan \\
    Department of Statistics\\
    University of Connecticut\\
    Storrs, CT, USA \\
    \texttt{jun.yan@uconn.edu}
    \And
    Sangwook Kang\thanks{Corresponding author} \\
    Department of Applied Statistics and \\
    Department of Statistics and Data Science \\
    Yonsei University\\
    Seoul, Republic of Korea \\
    \texttt{kanggi1@yonsei.ac.kr}
}

\begin{document}
\maketitle

\begin{abstract}
    We propose a set of goodness-of-fit tests for the semiparametric accelerated failure time (AFT) model, including an omnibus test, a link function test, and a functional form test. This set of tests is derived from a multi-parameter cumulative sum process shown to follow asymptotically a zero-mean Gaussian process. Its evaluation is based on the asymptotically equivalent perturbed version, which enables both graphical and numerical evaluations of the assumed AFT model. Empirical p-values are obtained using the Kolmogorov-type supremum test, which provides a reliable approach for estimating the significance of both proposed un-standardized and standardized test statistics. The proposed procedure is illustrated using the induced smoothed rank-based estimator but is directly applicable to other popular estimators such as non-smooth rank-based estimator or least-squares estimator.Our proposed methods are rigorously evaluated using extensive simulation experiments that demonstrate their effectiveness in maintaining a Type \Romannum{1} error rate and detecting departures from the assumed AFT model in practical sample sizes and censoring rates. Furthermore, the proposed approach is applied to the analysis of the Primary Biliary Cirrhosis data, a widely studied dataset in survival analysis, providing further evidence of the practical usefulness of the proposed methods in real-world scenarios. To make the proposed methods more accessible to researchers, we have implemented them in the \textbf{R} package \textbf{afttest}, which is publicly available on the Comprehensive R Archieve Network.
    
\end{abstract}

\section{Introduction} \label{sec:background}
    An accelerated failure time (AFT) model provides a valuable tool for assessing the effects of covariates on failure times. It links the natural logarithm of failure time to a set of predictors (including risk factors and covariates) in a linear fashion with an addition of a random error term. When the distribution of the random error is left unspecified, the corresponding AFT model is referred to as a semiparametric AFT model. This semiparametric AFT model has several notable features: First, covariates act directly on the failure time which makes interpreting the effects straightforward compared to the other popular models that act on the hazard function such as Cox \citep{wei1992accelerated}. Second, it maintains flexibility in the shape of the failure time distribution since no distributional assumption is made on the error term. With these desirable features, semiparametric AFT models have served as a useful alternative to the popular Cox model in the regression analysis of failure time data subject to censoring. Estimation of regression parameters, the rank-based estimation procedure \citep{prentice1978linear, tsiatis1990estimating, jin2003rank, chiou2014fitting} and least-squares estimation procedure \citep{buckley1979linear, jin2006least, chiou2014marginal} has been popular among others. 
    
    Model-checking procedures are essential steps that follow inference procedures to validate the model-fitting results. The popular Cox model has offered well-developed procedures for estimating regression parameters under a wide range of settings \citep{cox1972regression}. Moreover, it provides an array of model-checking methods, including omnibus tests, link function tests, and functional form tests. For example, as a technique for assessing the validity of the marginal Cox model, procedures using a cumulative partial-sum process based on martingale-type residuals have been predominantly popular \citep{lin1993checking, spiekerman1996checking,huang2011model,lu2014testing,li2015checking,lee2019model}. These methodological and theoretical developments are implemented in popular computer software such as \textbf{R} \citep{r2022}. For example, an \textbf{R} package for assessing goodness-of-fit of a Cox proportional hazards regression model and Fine-Gray subdistribution model, \textbf{goftte} \citep{sfumato2019goftte}, has been developed. 

    Unlike these developments of model-checking procedures for the Cox model, most of the literature on model-checking procedures for AFT models has focused on parametric AFT models \citep{bagdonavivcius2013chi, balakrishnan2013testing,lin1996model,cockeran2021goodness}; only limited studies are available for checking semiparametric AFT models \citep{novak2013goodness, novak2015regression, chiou2014marginal}. Based on the rank-based estimator, \cite{novak2013goodness, novak2015regression} proposed an omnibus test checking the overall departure from the assumed model, which shares the similar spirit of \citet{lin1993checking}. \citet{lin1996model} developed similar procedures but with parametric AFT models. Despite these developments, several important aspects are either missing or need to be improved. First, checking the assumptions regarding the linearity of the effect of each covariate on the transformed $T$ or the exponential link function between a set of covariates and the failure time are often of specific interest, and violations of such specific assumptions could lead to invalid statistical inferences. Therefore, checking these assumptions and assessing model adequacy are crucial after fitting an AFT model but such procedures are not fully developed for semiparametric AFT models. For least-squares estimators, another popular class of estimators, checking the functional form of covariates in marginal semiparametric AFT models have been considered \citep{chiou2014marginal} but the derivation of the properties of the proposed processes requires a formal investigation. Second, the omnibus test procedures proposed in \citet{novak2013goodness,novak2015regression} are restricted to the rank-based estimator of \citet{prentice1978linear,jin2003rank}, non-smooth in model parameters. Procedures for checking specific model components as well as the overall departure from the assumed AFT model under other popular estimators such as the induced smoothed estimator \citep{Brow:Wang:indu:2007,johnson2009induced}, computationally more efficient version of the non-smooth rank-based estimator or the aforementioned least-squares estimators have not yet been developed. Third, the proposed process in \citet{novak2013goodness} is a process of time ($t$) evaluated at  fixed values of covariates ($z$). Meanwhile, the cumulative-sum processes considered in \citet{lin1993checking, spiekerman1996checking} are multi-parameter processes in $t$ as well as $z$.  
    Therefore, model diagnostic procedures that properly account for the aforementioned three aspects in checking the semiparametric AFT model are still remain to be developed.  
    
    Our primary objective can be summarized as follows. We propose to develop suitable test statistics that can effectively assess the adequacy of  semiparametric AFT models. We propose a series of test statistics based on a multi-parameter stochastic process of $t$ and $z$ and demonstrate that they follow Gaussian processes under the null hypothesis of the assumed AFT model. We consider various forms of functions based on the proposed test statistics to enable omnibus tests, link function tests, and functional form tests, providing a comprehensive assessment of semiparametric AFT models. We provide non-standardized and standardized test statistics based on the popular estimators for the semiparametric AFT model; namely, the non-smoothed rank-based estimator, its induced-smoothed version  and least-squares estimator. To ensure objectivity, we employ the Kolmogorov-type supremum test for calculating p-values, providing a rigorous measure of the model's goodness-of-fit. Furthermore, to facilitate the practical implementation of our proposed methods, we have also developed an \textbf{R} package, \textbf{afttest}, \citep{bae2022afttest}. This package provides user-friendly tools for conducting model adequacy assessments based on our proposed test statistics, allowing for easy application to medical data.
    
    The structure of the article is organized as follows. Section \ref{sec:modest} introduces the semiparametric AFT model and its parameter estimation procedure. Section \ref{sec:modchk} presents the proposed model-checking methods for assessing the adequacy of the AFT model. In Section \ref{sec:sim}, we describe the simulation studies conducted to evaluate the performance of the proposed model-checking methods. Section \ref{sec:realdata} demonstrates the application of the proposed model-checking procedures to the well-known Primary Biliary Cirrhosis (PBC) study data. Lastly, Section \ref{sec:conclusion} provides a discussion and concluding remarks.
    
\section{Semiparametric AFT model and estimation} \label{sec:modest}
    Suppose that the sample is composed of $n$ independent subjects. Let $T_{i}$ and $C_{i}$ denote the potential failure time and censoring time for the $i^{th}$ subject, respectively ($i = 1, \cdots, n$). We assume that $T_{i}$ is subject to right censoring, so the observed time is $X_{i} = \min \left( T_{i}, C_{i} \right)$. The corresponding failure indicator is $\Delta_{i} = I \left(T_{i} \leq C_{i} \right)$. For each subject, we observe a $p$-dimensional bounded vector of covariates $\bZ_{i} = \left( Z_{i1}, \cdots, Z_{ip} \right)^{\top}$. Throughout this article, we consider time-invariant $\bZ$. For the $i^{th}$ subject, the observed data are then $\left( X_{i}, \Delta_{i}, \bZ_{i} \right)$. $T_{i}$ and $C_{i}$ are assumed to be independent given $\bZ_{i}$ for all $i$. For a given $\bZ_{i}$, $T_{i}$ is assumed to follow a semiparametric AFT model. Specifically, for the $i^{th}$ subject, 
    \begin{align} \label{eq:sec:modest:1}
        \log T_{i} = - \bZ_{i}^{\top} \bb_{0} + \epsilon_{i},
    \end{align}
    where $\bb_{0} = \left(\beta_{01} , \cdots , \beta_{0p} \right)^{\top}$ is a true vector of regression parameters and $\epsilon_{i}$ is an unspecified random error term.

    Let $e_{i} \left( \bb \right) = \log X_{i} + \bZ_{i}^{\top} \bb$ denote the $i^{th}$ observed residual based on the assumed model in \eqref{eq:sec:modest:1} where $\bb= \left( \beta_{1},\cdots,\beta_{p} \right)^{\top}$ is an unknown vector of regression parameters. Let $N_{i} \left( t, \bb \right) = I \left( e_{i} \left( \bb \right) \leq { t }, \Delta_{i} = 1 \right)$ and $Y_{i} \left( t, \bb \right) = I \left( e_{i} \left( \bb \right) \geq { t } \right)$ be the corresponding counting process and at-risk process, respectively. 
    
    For estimation of the regression coefficients $\bb$ in \eqref{eq:sec:modest:1}, the rank-based approach \citep{prentice1978linear} and the least squares approach (references) have been popular. 
    In the rank-based approach, estimating functions $\score \left(t, \bb \right)$ at $t$ based on the ranks of $e_{i}$'s may be used where
    \begin{align} \label{eq:sec:modest:2}
        \score \left(t, \bb \right) = \sum_{i=1}^{n} \int_{- \infty}^{t} \psi \left( s, \bb \right) \left( \bZ_{i} - E \left( s, \bb \right) \right) dN_{i} \left(s, \bb \right),
    \end{align}    
    $\psi\left(s, \bb \right)$ is a possibly data-dependent weight function, $S_{d} \left( t, \bb \right) = \sum_{i=1}^{n} \bZ_{i}^{\otimes d} Y_{i} \left( t, \bb \right), d=0, 1$ and $E \left( t, \bb \right) = S_{1} \left( t, \bb \right) / S_{0} \left( t, \bb \right)$. Throughout this article, we assume the Gehan-type weight for $\psi\left(s, \bb \right)$, i.e.,  $\psi\left(s, \bb \right)= n^{-1} S_{0} \left( t, \bb \right)$. Then, with $t = \infty$, \eqref{eq:sec:modest:2} can be reduced as 
    \begin{align} \label{eq:sec:modest:3}
        \scoreN \left( \bb \right)
        & = n^{-1} \sum_{i=1}^{n} \sum_{j=1}^{n} \Delta_{i} \left( \bZ_{i} - \bZ_{j} \right)  I \left( e_{j} \left( \bb \right) \geq e_{i} \left( \bb \right) \right) 
    \end{align}
    Note that \eqref{eq:sec:modest:3} is non-smooth in $\bb$. An induced smoothed version of \eqref{eq:sec:modest:3}, $\scoreI \left( \bb \right)$, may also be considered for computational efficiency, especially in variance estimation. Specifically,     
    \begin{align} \label{eq:sec:modest:4}
        \scoreI \left( \bb  \right)
        = n^{-1} \sum_{i=1}^{n} \sum_{j=1}^{n} \Delta_{i} \left(\bZ_{i} - \bZ_j \right) \Phi \left( \frac{ e_{j} \left( \bb \right) - e_{i} \left( \bb \right) }{r_{ij}} \right),
    \end{align}
    where $r_{ij}^{2} = n^{-1} \left( \bZ_{i} - \bZ_{j} \right)^{\top} \left( \bZ_{i} - \bZ_{j} \right)$, and $\Phi$ is the standard normal cumulative distribution function (CDF). Estimators $\bbhn$ and $\bbhi$ for $\bb_0$ are defined as the solutions to $\scoreN \left( \bb \right) = 0$ and $\scoreI \left( \bb \right) = 0$, respectively \citep{chiou2014fitting, chiou2015rank}. Under some regularity conditions, $\bbhn$ is shown to be consistent for $\bb_0$ and $n^{\frac{1}{2}} \left( \bbhn - \bb_0 \right)$ is shown to converge to a zero-mean normal random variable with a finite covariance matrix \citep{lin1998accelerated}. $\bbhi$ is shown to share the same asymptotic properties of $\bbhn$ \citep{johnson2009induced, chiou2014fitting, chiou2015rank}.

    In the least squares approach, the usual least squares normal equations are used with replacing the right-censored failure time by its estimated conditional expectation $\widehat{\mbox{E}} \left[ \log T_{i} | X_{i}, \Delta_{i}, \bZ_{i} \right]$. Specifically,
    \begin{align*}
        \widehat{T}_{i}^{\ast} \left( \bb \right) 
        = \widehat{\mbox{E}}\left[ \log T_{i} | X_{i}, \Delta_{i}, \bZ_{i} \right] 
        = \Delta_{i} \log X_{i} + \left( 1 - \Delta_{i} \right) \left[ \frac{\int_{e_{i} \left( \bb \right)} ud \widehat{F}_{\bb} \left( u \right) }{1 - \widehat{F}_{\bb} \left\{ e_{i} \left( \bb \right)\right\} } - \bZ_{i}^{\top}\bb \right]
    \end{align*}
    where $\widehat{F}_{\bb}(\cdot)$ is the Kaplan-Meier estimator of the CDF of $e_{i} \left( \bb \right)$s. Then, the corresponding normal equations are written as
    \begin{align*}
        \scoreL \left( \bb, b \right) = \sum_{i=1}^n (\bZ_{i} - \bar{\bZ})\left\{ \widehat{T}_{i}^{\ast} \left( b \right) + \bZ_{i}^{\top} \bb \right\},
    \end{align*}
    where $\bar{\bZ} = n^{-1}\sum_{i=1}^n \bZ_{i}$. An estimator for $\bb$ can be obtained by iteratively evaluating $\bbhl^{(m)} = L(\bbhl^{(m-1)})$ for $m \geq 1$ where
    \begin{align*}
        L \left( b \right) = \left\{ \sum_{i=1}^n (\bZ_{i} - \bar{\bZ})^{\otimes 2}\right\}^{-1}\left[ \sum_{i=1}^n (\bZ_{i} - \bar{\bZ}) \left\{ \bar{T}^{\ast} \left( b \right) - \widehat{T}_{i}^{\ast} \left( b \right)\right\}\right] \mbox{ and}
    \end{align*}
    $\bar{T}^{\ast} \left( b \right) = n^{-1} \sum_{i=1}^n \widehat{T}_{i}^{\ast} \left( b \right)$ \citep{jin2006least, chiou2014marginal}. 
    This approach is computationally more efficient and stable than the Buckley-James estimator \citep{buckley1979linear} that solves $\scoreL(\bb, \bb) = 0$ since $\scoreL(\bb, \bb)$ is neither smooth nor monotone in $\bb$. For each $m$, $\bbhl^{(m)}$ is shown to be consistent and asymptotically normal if the initial estimator, $\bbhl^{(0)}$ is consistent and asymptotically normal \citep{jin2006least}. 
    
\section{Model-checking procedures} \label{sec:modchk}

\subsection{Multi-parameter stochastic process} \label{subsec:multiproc}
    In this section, we will introduce a class of multi-parameter stochastic processes, which is the cumulative sum of martingale residuals for the AFT model, and discuss its properties. We will begin by defining a Nelson-Aalen type estimator of the cumulative hazard function of residuals, denoted by $\Lambda(t, \bb)$, where $\bb$ represents the vector of regression coefficients such as
    \begin{align*}
        \widehat{\Lambda} \left( t, \bb \right) = \int_{0}^{t} \frac{J \left( s \right)}{S_{0} \left( s, \bb \right)} dN_{\bullet} \left( s, \bb \right) 
    \end{align*}   
    where $N_{\bullet} \left( s, \bb \right)=\sum_{i=1}^{n} N_{i} \left( t, \bb \right)$ and $J\left( t \right) = I \left( S_{0} \left( t, \bb \right) > 0 \right)$. Then, a stochastic process $M_{i} \left( t, \bb \right)$, $i=1, \cdots, n$, defined by $M_{i} \left( t, \bb \right) = N_{i} \left( t, \bb \right) - \int_{0}^{t} {Y_{i} \left( s, \bb \right) d \Lambda \left( s \right)}$ is a zero-mean martingale with respect to a suitable filtration \citep{lin1998accelerated, novak2015regression}. The corresponding martingale residual, an estimated version of $M_{i} \left( t, \bb \right)$, is denoted as $\widehat{M}_{i} \left( t, \bbhi \right) = N_{i} \left( t, \bbhi \right) - \int_{0}^{t} {Y_{i} \left( s, \bbhi \right) d \widehat{\Lambda} \left( s, \bbhi \right)}$. If the assumptions on the AFT model are satisfactory, the martingale residuals $\widehat{M}_{i} \left( t, \bbhi \right)$ are expected to fluctuate around zero \citep{novak2013goodness, novak2015regression}. This offers an essential building block for detecting model misspecifications. 
    
    For an objective diagnostic test, we consider a class of multi-parameter stochastic processes, the cumulative sum of martingale residuals for the AFT model, whose form is given as
    \begin{align*}
        \w \left( t, \bz  \right) = n^{- \frac{1}{2}} \sum_{i=1}^{n} \pi_{i} \left( \bz \right) \widehat{M}_{i} \left( t, \bbhi \right),
    \end{align*}
    where $\pi_{i} \left( \bz \right) = I \left( \bZ_{i} \leq \bz \right) = \prod_{q=1}^{p} I \left( \bZ_{iq} \leq \bz_{q} \right)$ is a weight function of $\bz$. One can use $\pi_{i} \left( \bz \right) = \ell \left( \bz; \bZ_{i} \right) I \left( \bZ_{i} \leq \bz \right)$ where $\ell \left( \cdot \right)$ is a bounded function of $\bz$. Similar ideas have been popularly used in the survival analysis literature \citep[see, $e.g.$, ][]{barlow1988residuals, therneau2015package, lin1993checking, novak2013goodness, novak2015regression}. For checking the semiparametric AFT model, \citet{novak2013goodness, novak2015regression} also considered the same process only with $\bbhn$, the non-smoothed estimator for $\bb$ and at a fixed $\bz$. 

    Under some regularity conditions, the proposed stochastic process can be shown to converge to a zero-mean Gaussian process. We use this process to detect model misspecifications in several directions. The asymptotic property is summarized in the following theorem. 
    \begin{restatable}{theorem}{theoremone} \label{theorem1}
        $\w \left( t, \bz \right)$ follows zero-mean Gaussian process asymptotically for varying $t$ and $\bz$, and it is consistent against any misspecifications of the semiparametric AFT model.
    \end{restatable}
    \begin{proof}
	A sketch of the proof is provided in Appendix. 
    \end{proof}

    Here, we define $\w \left( t, \bz \right)$ and derive its asymptotic distribution based on $\bbhi$, induced smoothed rank-based estimator for $\bb$ but the result is more general in the sense that the other two estimators $\bbhn$ and $\bbhl$ can also be incorporated. $n^{\frac{1}{2}} \left( \bbhn -  \bb_0 \right)$ and $n^{\frac{1}{2}} \left( \bbhi -  \bb_0 \right)$ are shown to converge to the same limiting distribution \citet{johnson2009induced}, so the result of Theorem 1 remains unchanged for $\bbhn$. For $\bbhl$, the asymptotic representation of $n^{\frac{1}{2}} \left( \bbhl -  \bb_0 \right)$ is different from that for $n^{\frac{1}{2}} \left( \bbhn -  \bb_0 \right)$ (and, consequently, for $n^{\frac{1}{2}} \left( \bbhi -  \bb_0 \right)$) but is shown to be asymptotically normal (Jin et al. 2006). So, the main result still remains the same with a different form of the asymptotic covariance function.

    A direct assessment of the significance level of the proposed tests is difficult due to the complicated nature of the covariance function of the limiting distribution of $\w \left( t, \bz \right)$. So, we propose a perturbed version of the original process, which has the same limiting distribution under the null hypothesis. A similar approach has also been considered in \citet{novak2013goodness}. Define
    \begin{align*}
        S_{\pi} \left(t, \bz, \bb \right) = \sum_{i=1}^{n} \pi_{i} \left( \bz \right) Y_{i} \left( t, \bb \right) 
        &\hspace{1.0cm} 
        E_{\pi} \left(t, \bz, \bb \right) = \frac{S_{\pi} \left(t, \bz, \bb \right) }{S_{0} \left( t, \bb \right)} \\
        f_{\pi} \left( t,\bz \right) = n^{-1} \sum_{i=1}^{n} \Delta_{i} \pi_{i} \left( \bz \right) f_{0} \left( t \right) t \bZ_{i} 
        &\hspace{1.0cm}
        g_{\pi} \left( t,\bz \right) = n^{-1} \sum_{i=1}^{n} \pi_{i} \left( \bz \right) g_{0} \left( t \right) t \bZ_{i} ,
    \end{align*}
    where $f_{0} \left( t \right)$ and $g_{0} \left( t \right)$ are the baseline densities of $T_{i} \exp \left( {\bZ_{i}^{\top} \bb_{0}} \right)$ and $X_{i} \exp \left( {\bZ_{i}^{\top} \bb_{0}} \right)$, respectively. One can obtain estimators, $\widehat{f}_{0} \left( t \right)$ and $\widehat{g}_{0} \left( t \right)$ through the kernel smoothed techniques \citep{diehl1988kernel, novak2013goodness, novak2015regression, silverman2018density}. We assume each Kernel estimate of $\widehat{f}$ and $\widehat{g}$ have bounded variations and converge in probability uniformly to $f_{0}$ and $g_{0}$, respectively. Indeed, we can obtain $\widehat{f}_{\pi} \left( t, \bz \right)$ and $\widehat{g}_{\pi} \left( t, \bz \right)$ by plugging in $\widehat{f}_{0} \left( t \right)$ and $\widehat{g}_{0} \left( t \right)$, respectively. We further assume that $f_{\pi} \left( t, \bz \right)$ and $g_{\pi} \left( t, \bz \right)$ have bounded variances and converges a.s. to $f_{\pi 0 } \left( t, \bz \right)$ and $g_{\pi 0} \left( t, \bz \right)$, respectively.
        
    We also define the following perturbed processes by generating $n$ independent and identically distributed (i.i.d.) positive random variables, $(\phi_{1}, \cdots, \phi_{n})$, with $\mE{\phi_{i}} = \mvar{\phi_{i}} = 1$, $i=1\, \ldots, n$:
    \begin{align*}
        \score_{\pi}^{\phi} \left(t, \bz, \bb \right) &=  \sum_{i=1}^{n} \int_{- \infty}^{t} \psi \left( s, \bb         \right) \left\{ \pi_{i} \left( \bz \right) - E_{\pi} \left( s, \bz, \bb \right) \right\} d\widehat{M}_{i} \left(s, \bb \right) (\phi_{i} - 1) \ \mbox{ and } \\
        \scoreI^{\phi} \left( \bb \right) 
        &= n^{-1} \sum_{i=1}^{n} \sum_{j=1}^{n} \Delta_{i} \left(\bZ_{i} - \bZ_j \right) \Phi \left( \frac{ e_{j} \left( \bb \right) - e_{i} \left( \bb \right) }{r_{ij}} \right) \phi_{i}. 
    \end{align*}

    The following theorem summarizes the asymptotic equivalence between the original process and its perturbed version for the induced smoothed rank-based estimator.
    \begin{restatable}{theorem}{theoremtwo} \label{theorem2}
        Under the assumptions in Appendix and given the observed data, $\left( N_{i} \left( t, \bb \right), Y_{i} \left( t, \bb \right), \bZ_{i} \right)$, $i=1, \ldots, n$, $\w \left( t, \bz \right)$ and $\wh \left(t, \bz \right)$ has the same limiting distribution where 
        \begin{align*}
            \wh \left( t, \bz \right) 
            &=n^{-\frac{1}{2}} \score_{\pi}^{\phi} \left( t, \bz,  \bbhi \right) -n^{\frac{1}{2}} \left( \widehat{f}_{\pi} \left( t, \bz \right) + \int_{0}^{t} \widehat{g}_{\pi} \left( s, \bz\right) d \widehat{\Lambda} \left(s, \bbhi \right) \right)^{\top} \left( \bbhi - \bbhid \right) \\
             -& n^{-\frac{1}{2}} \int_{0}^{t} S_{\pi} \left( s, \bz, \bbhi \right) d \left( \widehat{\Lambda} \left(s, \bbhi \right) - \widehat{\Lambda} \left(s, \bbhid \right) \right) \ \mbox{ and } \
            \bbhid \mbox{ is the solution to } 
            \scoreI^{\phi} \left( \bb \right) = 0.
        \end{align*}
    \end{restatable}
    \begin{proof}
    A sketch of the proof is provided in Appendix.
    \end{proof}

    Note that Theorem~\ref{theorem2} is directly applicable to the other estimators mentioned in Section 2. Specifically, for the non-smooth rank-based estimator, we replace $\bbhi$ and $\bbhid$ with $\bbhn$ and $\bbhnd$, respectively, where $\bbhnd$ is the solution to $\scoreN^{\phi} \left( \bb \right) = 0$ and 
    \begin{align*}
        \scoreN^{\phi} \left( \bb \right) 
        = n^{-1} \sum_{i=1}^{n} \sum_{j=1}^{n} \Delta_{i} \left(\bZ_{i} - \bZ_j \right) I \left( e_{j} \left( \bb \right) \geq  e_{i}\left( \bb \right) \right) \phi_{i}.
    \end{align*}
    Likewise, for the least-squares estimator, we replace $\bbhi$ and $\bbhid$ with $\bbhl$ and $\bbhld$, respectively, where $\bbhld$ is the solution to $\scoreL^{\phi} \left( \bb \right) = 0$,  
    \begin{align*}
        \scoreL^{\phi} \left( \bb, b \right) = \sum_{i=1}^n (\bZ_{i} - \bar{\bZ})\left\{\widehat{T}^{\phi \ast}_{i} \left( b \right) + \bZ_{i}^{\top}\bb \right\} \phi_{i}, \ 
        \widehat{T}^{\phi \ast}_{i} \left( \bb \right) = \Delta_{i} \log X_{i} + \left( 1 - \Delta_{i} \right) \left[ \frac{\int_{e_{i} \left( \bb \right)} ud\widehat{F}^{\ast}_{\bb} \left( u \right) }{1 - \widehat{F}^{\ast}_{\bb}\left\{e_{i} \left( \bb \right)\right\}} - \bZ_{i}^{\top}\bb \right]  \mbox{ and } 
    \end{align*}
    $\widehat{F}^{\ast}_{\bb}(\cdot)$ is a perturbed version of $\widehat{F}_{\bb}(\cdot)$.
    
\subsection{Specific test forms} \label{subsec:testform}

\subsubsection{Functional form test} \label{subsubsec:formtest}
    We first consider checking the functional form of a covariate. 
    A general representation of the AFT model with a functional form for the $q^{\text{th}}$ covariate as $g_{q} \left( Z_{iq} \right)$ is as follows:
    \begin{align*}
        \log T_{i} &= g^{\top} \left( \bZ_{i} \right) \bb + \epsilon_{i}
    \end{align*}
    where $g^{\top} \left( \bZ_{i} \right) = \left\{ g_{1}(Z_{i1}), \cdots, g_{p}(Z_{ip}) \right\}$. If provided data satisfy the AFT model assumptions, each function $g_{i}(\cdot)$ for $\bZ_{i}$ is an identity function. Using the analogous arguments in \citet{lin1993checking}, the proposed test statistic for checking the functional form can be represented in the following way:
    \begin{align*}
        \sup_{\bz} \abs{ \wh_{q} \left( \bz \right) }
        \quad \text{where} \quad
        \wh_{q} \left( \bz \right) = \sum_{i=1}^{n}  I \left( Z_{iq} \leq \bz_{q} \right) \widehat{M}_{i} \left( \infty , \bbhi \right).
    \end{align*}
    Note that it is a special case of $\wh \left( t, \bz \right)$ with $t = \infty$ and $\bz_{r} = \infty$ for all $r \neq q (r, q = 1, \ldots, p)$. With this choice of $\bz$, $\pi_{i} \left( \bz \right)$ can be reexpressed as $\pi_{i} \left( \bz \right) =  I \left( \bZ_{i} \leq \bz \right) = \prod_{q=1}^{p} I \left( Z_{iq} \leq \bz \right) = I \left( Z_{iq} \leq \bz_{q} \right)$. Since it checks the functional form of the $q^{\text{th}}$ covariate, functional forms of other covariates are dominated by $\bZ_{r} = \infty$.
    
\subsubsection{Link function test} \label{subsubsec:linktest}
    The AFT model with a general link function can be written as 
    \begin{align*}
        \log T_{i} &= g \left( \bZ_{i}^{\top} \bb \right) + \epsilon_{i} 
    \end{align*}
    where $g \left( \cdot \right)$ is a link function. If the AFT model assumptions are satisfied, the link function $g \left( \cdot \right)$ is the identity function. Using the analogous arguments in \citet{lin1993checking}, the proposed test statistic for checking the link function is given as follows: 
    \begin{align*} 
        \sup_{\bz} \abs{\wh \left( \bz \right)} \quad \text{where} \quad \wh \left( \bz \right) = \sum_{i=1}^{n} I \left( \bZ_{i} \leq \bz \right) \widehat{M}_{i} \left( \infty, \bbhi \right),
    \end{align*}
    It is a special case of $\wh \left( t, \bz \right)$ with $t = \infty$ for $p>1$. When $p=1$, the test statistic for checking the link function reduces to that for checking the functional form. 

\subsubsection{Omnibus test}  \label{subsubsec:omnibustest}
    An omnibus test to detect an overall departure from the assumed AFT model can be conducted by using the following test statistic: 
    \begin{align*}
        \sup_{t,z} \abs{\wh \left( t, \bz \right)} \quad \text{where} \quad \wh \left( t, \bz \right)
        = \sum_{i=1}^{n} I \left( \bZ_{i} \leq \bz \right) \widehat{M}_{i} \left( t, \bbhi \right), 
    \end{align*}
    This is also a special case of $\wh \left( t, \bz \right)$ where $t$ and $\bz$ are considered simultaneously in finding the supremum. A similar test statistic with a single parameter stochastic process is considered in \citet{novak2013goodness, novak2015regression} in which the procedure is based on the non-smoothed estimator that solves \eqref{eq:sec:modest:3}. 
    
\subsubsection{Test procedures}  \label{subsubsec:testprocedures}
    By considering specific forms of $\w \left( t, \bz \right)$ in subsections \ref{subsubsec:formtest} - \ref{subsubsec:omnibustest}, several aspects of model misspecifications including the overall departure from the assumed AFT model, link function and functional form can be evaluated. Here, we consider two types of Kolmogorov-type supremum tests, so-called unstandardized and standardized tests, where 
    \begin{align*}
        \sup_{t, z} \abs{\w \left( t, \bz \right)} \qquad \text{and} \qquad \sup_{t, z} \abs{\frac{ \w \left( t, \bz \right) }{\sqrt{ \widehat{\text{V}} \left( \wh \left( t, \bz \right) \right) } }},
    \end{align*}
    where $\widehat{\text{V}} \left( \wh \left( t, \bz \right) \right)$ can be obtained from the sample variance of $\wh \left( t, \bz \right)$. The specific test procedure based on the unstandardized test statistics is as follows: 
    \begin{enumerate}
        \item Calculate the observed statistics $\w \left( t, \bz \right)$ and an absolute supremum of $\w \left( t, \bz \right)$, $W_{\text{obs}} = \underset{t, z} {\sup} \abs{\w \left( t, \bz \right)} $
        
        \item Generate $K$ approximated paths, $\left\{ \wh^{(l)} \left( t, \bz \right) \right\}_{l=1}^{K}$ and let $\wh_{\text{app}}^{(l)} = \underset{t, z}{\sup} \abs{\wh \left( t, \bz \right)}$ for $l = 1, \cdots, K$
        
        \item Let $l_{0}$ be the number of approximated paths satisfying $\left\{ W_{\text{obs}} \leq W_{\text{app}}^{(l)} \right\}_{l=1}^{K}$, $i.e.$, $l_{0} = \sum_{l=1}^{K} I \left( W_{\text{obs}} \leq W_{\text{app}}^{(l)} \right) $
        
        \item Consider the $p$-value as the ratio of how the observed statistic is embedded in the distribution of approximated paths.  $i.e.$ $\text{$p$-value} = \frac{l_{0}}{K}$
        
    \end{enumerate}
    
    The testing procedure based on the standardized test statistics can be performed in the same way except that $W_{\text{obs}}$ and $\wh_{\text{app}}^{(l)}$ in Steps 1 and 2 are replaced by their standardized counterparts. In the subsequent simulation studies and real data analysis sections, results from using both unstandardized and standardized test procedures are provided.

\section{Simulation Study} \label{sec:sim}

\subsection{Simulation Scenario 1} \label{subsec:sim1}
    To assess the finite sample performance of the proposed model-checking procedures, we conducted extensive simulation experiments under different sets of covariates, censoring rates, and sample sizes. In Simulation \hypertarget{scenario1}{Scenario 1}, we generated $T$ from an AFT model with linear and quadratic terms for the covariate $Z$.
    \begin{align} \tag{\ref{subsec:sim1}.1}
        \log T_{i} = - \beta_{0} - \beta_{1} Z_{i} - \gamma Z_{i}^{2} + \epsilon_{i}, \label{model:sim1}
    \end{align}
    where $Z_{i}$ is a normal random variable with mean $2$ and standard error $1$. We set $\beta_{0} = - 4$ and $\beta_{1} = 1$. $\gamma$ is set to increase by 0.1 from 0 to 0.5. The error term is randomly generated from the standard normal distribution. We consider two censoring rates (20\% and 40\%) with three different sample sizes ($N = 100, 300$, and $500$). The censoring time is independently generated from a log-normal distribution with mean $\tau$ and standard error $1$ where $\tau$ is determined to achieve the desired censoring rates around $20\%$ or $40\%$. For the test statistic, we consider both unstandardized and standardized versions. We considered two model misspecification aspects: the link function and functional form of covariates, as well as omnibus tests. We also included the test statistic proposed by \citet{novak2013goodness}, which is based on the non-smoothed estimator for regression coefficients. The null model assumed was a standard AFT model with a linear term for $Z$. When $\gamma = 0$, we assessed the type \Romannum{1} error rate. When $\gamma \neq 0$, the null model was misspecified by not accounting for the quadratic term, and we evaluated the power of the proposed test statistics. For each configuration, we generated 500 approximated sample paths from the null model to calculate $p$-values. We replicated the simulations 1,000 times. The induced smoothed and non-smoothed estimators for the regression coefficients were obtained using the \textbf{R} package \textbf{aftgee} \citep{chiou2014fitting}. The aforementioned model-checking procedures were implemented in the \textbf{R} package \textbf{afttest}, which is available in CRAN \citep{bae2022afttest}.
    
    In Simulation \hyperlink{scenario1}{Scenario 1}, we first assess the type \Romannum{1} error rate when $\gamma = 0$. When $\gamma \neq 0$, the assumed AFT model is misspecified by not accounting for the quadratic term and we assess the power of the proposed statistics in detecting this misspecification.
    \begin{table}[htp]
    \fontsize{8}{8}\selectfont
    \caption{\label{tab:sim1:result} Results for the omnibus (omni), link function (link), and functional form (form) tests under Simulation scenario 1. Monotone non-smooth method (mns) and monotone induced smoothed method (mis) are considered under the sample sizes of $n = 100, 300$, and $500$ with two different censoring rates of 20\% and 40\%. Each number in a cell represents the rejection ratio based on 1,000 replications for the standardized statistic (bold) or unstandardized statistic.} 
    \centering
    \begin{tabular}[htp]{cccccccccccccc}
        \toprule
        \multicolumn{2}{c}{censoring} & \multicolumn{6}{c}{20\%} & \multicolumn{6}{c}{40\%} \\
        \cmidrule(l{3pt}r{3pt}){1-2} \cmidrule(l{3pt}r{3pt}){3-8} \cmidrule(l{3pt}r{3pt}){9-14}
        \multicolumn{2}{c}{n} & \multicolumn{2}{c}{100} & \multicolumn{2}{c}{300} & \multicolumn{2}{c}{500} & \multicolumn{2}{c}{100} & \multicolumn{2}{c}{300} & \multicolumn{2}{c}{500} \\
        \cmidrule(l{3pt}r{3pt}){1-2} \cmidrule(l{3pt}r{3pt}){3-4} \cmidrule(l{3pt}r{3pt}){5-6} \cmidrule(l{3pt}r{3pt}){7-8} \cmidrule(l{3pt}r{3pt}){9-10} \cmidrule(l{3pt}r{3pt}){11-12} \cmidrule(l{3pt}r{3pt}){13-14}
        $\gamma$ & test & mns & mis & mns & mis & mns & mis & mns & mis & mns & mis & mns & mis\\
        \midrule
        \textbf{} & \textbf{} & \textbf{0.011} & \textbf{0.009} & \textbf{0.016} & \textbf{0.016} & \textbf{0.015} & \textbf{0.015} & \textbf{0.009} & \textbf{0.010} & \textbf{0.015} & \textbf{0.015} & \textbf{0.017} & \textbf{0.020}\\
        
         & \multirow{-2}{*}{omni} & 0.004 & 0.005 & 0.006 & 0.008 & 0.006 & 0.006 & 0.006 & 0.006 & 0.009 & 0.011 & 0.006 & 0.006\\
        \cmidrule(l{3pt}r{3pt}){2-2} \cmidrule(l{3pt}r{3pt}){3-4} \cmidrule(l{3pt}r{3pt}){5-6} \cmidrule(l{3pt}r{3pt}){3-4} \cmidrule(l{3pt}r{3pt}){5-6} \cmidrule(l{3pt}r{3pt}){7-8} \cmidrule(l{3pt}r{3pt}){9-10} \cmidrule(l{3pt}r{3pt}){11-12} \cmidrule(l{3pt}r{3pt}){13-14}
        
        \textbf{} & \textbf{} & \textbf{0.026} & \textbf{0.023} & \textbf{0.028} & \textbf{0.027} & \textbf{0.024} & \textbf{0.024} & \textbf{0.020} & \textbf{0.022} & \textbf{0.025} & \textbf{0.029} & \textbf{0.026} & \textbf{0.028}\\
        
         & \multirow{-2}{*}{link} & 0.011 & 0.012 & 0.014 & 0.010 & 0.008 & 0.010 & 0.011 & 0.011 & 0.014 & 0.019 & 0.011 & 0.010\\ 
        \cmidrule(l{3pt}r{3pt}){2-2} \cmidrule(l{3pt}r{3pt}){3-4} \cmidrule(l{3pt}r{3pt}){5-6} \cmidrule(l{3pt}r{3pt}){3-4} \cmidrule(l{3pt}r{3pt}){5-6} \cmidrule(l{3pt}r{3pt}){7-8} \cmidrule(l{3pt}r{3pt}){9-10} \cmidrule(l{3pt}r{3pt}){11-12} \cmidrule(l{3pt}r{3pt}){13-14}
        
        \textbf{} & \textbf{} & \textbf{0.024} & \textbf{0.025} & \textbf{0.027} & \textbf{0.028} & \textbf{0.025} & \textbf{0.026} & \textbf{0.023} & \textbf{0.020} & \textbf{0.029} & \textbf{0.027} & \textbf{0.029} & \textbf{0.028}\\ 
        
        \multirow{-6}{*}{0} & \multirow{-2}{*}{form} & 0.011 & 0.011 & 0.012 & 0.010 & 0.011 & 0.010 & 0.012 & 0.012 & 0.017 & 0.018 & 0.010 & 0.011\\
        \cmidrule(l{3pt}r{3pt}){1-14}
        
        \textbf{} & \textbf{} & \textbf{0.035} & \textbf{0.033} & \textbf{0.098} & \textbf{0.097} & \textbf{0.186} & \textbf{0.184} & \textbf{0.016} & \textbf{0.018} & \textbf{0.072} & \textbf{0.077} & \textbf{0.127} & \textbf{0.131}\\
        
         & \multirow{-2}{*}{omni} & 0.006 & 0.006 & 0.015 & 0.013 & 0.028 & 0.026 & 0.007 & 0.009 & 0.025 & 0.022 & 0.032 & 0.035\\
        \cmidrule(l{3pt}r{3pt}){2-2} \cmidrule(l{3pt}r{3pt}){3-4} \cmidrule(l{3pt}r{3pt}){5-6} \cmidrule(l{3pt}r{3pt}){3-4} \cmidrule(l{3pt}r{3pt}){5-6} \cmidrule(l{3pt}r{3pt}){7-8} \cmidrule(l{3pt}r{3pt}){9-10} \cmidrule(l{3pt}r{3pt}){11-12} \cmidrule(l{3pt}r{3pt}){13-14}
        
        \textbf{} & \textbf{} & \textbf{0.101} & \textbf{0.100} & \textbf{0.208} & \textbf{0.208} & \textbf{0.319} & \textbf{0.320} & \textbf{0.051} & \textbf{0.057} & \textbf{0.177} & \textbf{0.181} & \textbf{0.256} & \textbf{0.260}\\
        
         & \multirow{-2}{*}{link} & 0.013 & 0.014 & 0.022 & 0.022 & 0.040 & 0.038 & 0.015 & 0.018 & 0.029 & 0.030 & 0.040 & 0.038\\
        \cmidrule(l{3pt}r{3pt}){2-2} \cmidrule(l{3pt}r{3pt}){3-4} \cmidrule(l{3pt}r{3pt}){5-6} \cmidrule(l{3pt}r{3pt}){3-4} \cmidrule(l{3pt}r{3pt}){5-6} \cmidrule(l{3pt}r{3pt}){7-8} \cmidrule(l{3pt}r{3pt}){9-10} \cmidrule(l{3pt}r{3pt}){11-12} \cmidrule(l{3pt}r{3pt}){13-14}
        
        \textbf{} & \textbf{} & \textbf{0.098} & \textbf{0.101} & \textbf{0.214} & \textbf{0.208} & \textbf{0.324} & \textbf{0.320} & \textbf{0.054} & \textbf{0.055} & \textbf{0.173} & \textbf{0.182} & \textbf{0.258} & \textbf{0.256}\\
        
        \multirow{-6}{*}{0.1} & \multirow{-2}{*}{form} & 0.011 & 0.012 & 0.024 & 0.022 & 0.038 & 0.034 & 0.016 & 0.016 & 0.028 & 0.028 & 0.036 & 0.038\\
        \cmidrule(l{3pt}r{3pt}){1-14}
        
        \textbf{} & \textbf{} & \textbf{0.090} & \textbf{0.094} & \textbf{0.378} & \textbf{0.387} & \textbf{0.698} & \textbf{0.697} & \textbf{0.038} & \textbf{0.034} & \textbf{0.222} & \textbf{0.226} & \textbf{0.479} & \textbf{0.478}\\
        
         & \multirow{-2}{*}{omni} & 0.010 & 0.008 & 0.054 & 0.054 & 0.149 & 0.132 & 0.011 & 0.012 & 0.049 & 0.048 & 0.101 & 0.096\\
        \cmidrule(l{3pt}r{3pt}){2-2} \cmidrule(l{3pt}r{3pt}){3-4} \cmidrule(l{3pt}r{3pt}){5-6} \cmidrule(l{3pt}r{3pt}){3-4} \cmidrule(l{3pt}r{3pt}){5-6} \cmidrule(l{3pt}r{3pt}){7-8} \cmidrule(l{3pt}r{3pt}){9-10} \cmidrule(l{3pt}r{3pt}){11-12} \cmidrule(l{3pt}r{3pt}){13-14}
        
        \textbf{} & \textbf{} & \textbf{0.224} & \textbf{0.223} & \textbf{0.591} & \textbf{0.595} & \textbf{0.830} & \textbf{0.833} & \textbf{0.123} & \textbf{0.121} & \textbf{0.428} & \textbf{0.428} & \textbf{0.653} & \textbf{0.651}\\
        
         & \multirow{-2}{*}{link} & 0.020 & 0.025 & 0.085 & 0.088 & 0.197 & 0.188 & 0.019 & 0.024 & 0.070 & 0.064 & 0.128 & 0.122\\
        \cmidrule(l{3pt}r{3pt}){2-2} \cmidrule(l{3pt}r{3pt}){3-4} \cmidrule(l{3pt}r{3pt}){5-6} \cmidrule(l{3pt}r{3pt}){3-4} \cmidrule(l{3pt}r{3pt}){5-6} \cmidrule(l{3pt}r{3pt}){7-8} \cmidrule(l{3pt}r{3pt}){9-10} \cmidrule(l{3pt}r{3pt}){11-12} \cmidrule(l{3pt}r{3pt}){13-14}
        
        \textbf{} & \textbf{} & \textbf{0.228} & \textbf{0.233} & \textbf{0.589} & \textbf{0.598} & \textbf{0.835} & \textbf{0.835} & \textbf{0.117} & \textbf{0.118} & \textbf{0.428} & \textbf{0.429} & \textbf{0.650} & \textbf{0.651}\\
        
        \multirow{-6}{*}{0.2} & \multirow{-2}{*}{form} & 0.021 & 0.022 & 0.086 & 0.090 & 0.193 & 0.197 & 0.023 & 0.024 & 0.071 & 0.068 & 0.124 & 0.122\\
        \cmidrule(l{3pt}r{3pt}){1-14}
        
        \textbf{} & \textbf{} & \textbf{0.174} & \textbf{0.182} & \textbf{0.701} & \textbf{0.733} & \textbf{0.954} & \textbf{0.947} & \textbf{0.051} & \textbf{0.052} & \textbf{0.401} & \textbf{0.396} & \textbf{0.739} & \textbf{0.729}\\
        
         & \multirow{-2}{*}{omni} & 0.016 & 0.016 & 0.143 & 0.169 & 0.368 & 0.366 & 0.016 & 0.016 & 0.086 & 0.078 & 0.180 & 0.178\\
        \cmidrule(l{3pt}r{3pt}){2-2} \cmidrule(l{3pt}r{3pt}){3-4} \cmidrule(l{3pt}r{3pt}){5-6} \cmidrule(l{3pt}r{3pt}){3-4} \cmidrule(l{3pt}r{3pt}){5-6} \cmidrule(l{3pt}r{3pt}){7-8} \cmidrule(l{3pt}r{3pt}){9-10} \cmidrule(l{3pt}r{3pt}){11-12} \cmidrule(l{3pt}r{3pt}){13-14}
        
        \textbf{} & \textbf{} & \textbf{0.372} & \textbf{0.364} & \textbf{0.844} & \textbf{0.861} & \textbf{0.978} & \textbf{0.978} & \textbf{0.161} & \textbf{0.158} & \textbf{0.628} & \textbf{0.631} & \textbf{0.863} & \textbf{0.868}\\
        
         & \multirow{-2}{*}{link} & 0.035 & 0.035 & 0.202 & 0.222 & 0.460 & 0.454 & 0.031 & 0.034 & 0.109 & 0.113 & 0.228 & 0.229\\
        \cmidrule(l{3pt}r{3pt}){2-2} \cmidrule(l{3pt}r{3pt}){3-4} \cmidrule(l{3pt}r{3pt}){5-6} \cmidrule(l{3pt}r{3pt}){3-4} \cmidrule(l{3pt}r{3pt}){5-6} \cmidrule(l{3pt}r{3pt}){7-8} \cmidrule(l{3pt}r{3pt}){9-10} \cmidrule(l{3pt}r{3pt}){11-12} \cmidrule(l{3pt}r{3pt}){13-14}
        
        \textbf{} & \textbf{} & \textbf{0.366} & \textbf{0.361} & \textbf{0.840} & \textbf{0.855} & \textbf{0.978} & \textbf{0.981} & \textbf{0.167} & \textbf{0.159} & \textbf{0.637} & \textbf{0.633} & \textbf{0.867} & \textbf{0.875}\\
        
        \multirow{-6}{*}{0.3} & \multirow{-2}{*}{form} & 0.036 & 0.034 & 0.200 & 0.228 & 0.464 & 0.450 & 0.036 & 0.032 & 0.110 & 0.112 & 0.233 & 0.235\\
        \cmidrule(l{3pt}r{3pt}){1-14}
        
        \textbf{} & \textbf{} & \textbf{0.272} & \textbf{0.275} & \textbf{0.894} & \textbf{0.904} & \textbf{0.995} & \textbf{0.996} & \textbf{0.078} & \textbf{0.070} & \textbf{0.560} & \textbf{0.572} & \textbf{0.889} & \textbf{0.899}\\
        
         & \multirow{-2}{*}{omni} & 0.026 & 0.027 & 0.266 & 0.288 & 0.664 & 0.676 & 0.018 & 0.018 & 0.116 & 0.119 & 0.294 & 0.292\\
        \cmidrule(l{3pt}r{3pt}){2-2} \cmidrule(l{3pt}r{3pt}){3-4} \cmidrule(l{3pt}r{3pt}){5-6} \cmidrule(l{3pt}r{3pt}){3-4} \cmidrule(l{3pt}r{3pt}){5-6} \cmidrule(l{3pt}r{3pt}){7-8} \cmidrule(l{3pt}r{3pt}){9-10} \cmidrule(l{3pt}r{3pt}){11-12} \cmidrule(l{3pt}r{3pt}){13-14}
        
        \textbf{} & \textbf{} & \textbf{0.508} & \textbf{0.521} & \textbf{0.956} & \textbf{0.960} & \textbf{0.998} & \textbf{0.999} & \textbf{0.223} & \textbf{0.214} & \textbf{0.784} & \textbf{0.785} & \textbf{0.958} & \textbf{0.957}\\
        
         & \multirow{-2}{*}{link} & 0.050 & 0.052 & 0.352 & 0.378 & 0.746 & 0.756 & 0.040 & 0.036 & 0.166 & 0.168 & 0.358 & 0.357\\
        \cmidrule(l{3pt}r{3pt}){2-2} \cmidrule(l{3pt}r{3pt}){3-4} \cmidrule(l{3pt}r{3pt}){5-6} \cmidrule(l{3pt}r{3pt}){3-4} \cmidrule(l{3pt}r{3pt}){5-6} \cmidrule(l{3pt}r{3pt}){7-8} \cmidrule(l{3pt}r{3pt}){9-10} \cmidrule(l{3pt}r{3pt}){11-12} \cmidrule(l{3pt}r{3pt}){13-14}
        
        \textbf{} & \textbf{} & \textbf{0.507} & \textbf{0.511} & \textbf{0.955} & \textbf{0.956} & \textbf{0.999} & \textbf{0.999} & \textbf{0.222} & \textbf{0.221} & \textbf{0.784} & \textbf{0.785} & \textbf{0.959} & \textbf{0.957}\\
        
        \multirow{-6}{*}{0.4} & \multirow{-2}{*}{form} & 0.053 & 0.054 & 0.357 & 0.372 & 0.744 & 0.757 & 0.040 & 0.038 & 0.170 & 0.168 & 0.357 & 0.360\\
        \cmidrule(l{3pt}r{3pt}){1-14}
        
        \textbf{} & \textbf{} & \textbf{0.387} & \textbf{0.394} & \textbf{0.975} & \textbf{0.975} & \textbf{1.000} & \textbf{1.000} & \textbf{0.096} & \textbf{0.092} & \textbf{0.698} & \textbf{0.704} & \textbf{0.960} & \textbf{0.960}\\
        
         & \multirow{-2}{*}{omni} & 0.044 & 0.043 & 0.461 & 0.476 & 0.874 & 0.897 & 0.020 & 0.019 & 0.168 & 0.168 & 0.418 & 0.422\\
        \cmidrule(l{3pt}r{3pt}){2-2} \cmidrule(l{3pt}r{3pt}){3-4} \cmidrule(l{3pt}r{3pt}){5-6} \cmidrule(l{3pt}r{3pt}){3-4} \cmidrule(l{3pt}r{3pt}){5-6} \cmidrule(l{3pt}r{3pt}){7-8} \cmidrule(l{3pt}r{3pt}){9-10} \cmidrule(l{3pt}r{3pt}){11-12} \cmidrule(l{3pt}r{3pt}){13-14}
        
        \textbf{} & \textbf{} & \textbf{0.627} & \textbf{0.637} & \textbf{0.989} & \textbf{0.988} & \textbf{1.000} & \textbf{1.000} & \textbf{0.277} & \textbf{0.282} & \textbf{0.879} & \textbf{0.875} & \textbf{0.981} & \textbf{0.985}\\
        
         & \multirow{-2}{*}{link} & 0.076 & 0.079 & 0.562 & 0.564 & 0.922 & 0.935 & 0.044 & 0.046 & 0.232 & 0.227 & 0.503 & 0.514\\
        \cmidrule(l{3pt}r{3pt}){2-2} \cmidrule(l{3pt}r{3pt}){3-4} \cmidrule(l{3pt}r{3pt}){5-6} \cmidrule(l{3pt}r{3pt}){3-4} \cmidrule(l{3pt}r{3pt}){5-6} \cmidrule(l{3pt}r{3pt}){7-8} \cmidrule(l{3pt}r{3pt}){9-10} \cmidrule(l{3pt}r{3pt}){11-12} \cmidrule(l{3pt}r{3pt}){13-14}
        
        \textbf{} & \textbf{} & \textbf{0.627} & \textbf{0.631} & \textbf{0.990} & \textbf{0.990} & \textbf{1.000} & \textbf{1.000} & \textbf{0.266} & \textbf{0.278} & \textbf{0.874} & \textbf{0.880} & \textbf{0.986} & \textbf{0.985}\\
        
        \multirow{-6}{*}{0.5} & \multirow{-2}{*}{form} & 0.073 & 0.076 & 0.558 & 0.572 & 0.915 & 0.928 & 0.043 & 0.046 & 0.236 & 0.230 & 0.513 & 0.514\\
        \bottomrule
    \end{tabular}
\end{table}
    Table \ref{tab:sim1:result} presents the simulation results for the omnibus, link function, and functional form tests under Simulation \hyperlink{scenario1}{Scenario 1}. When the censoring rate is 20\% and $\gamma = 0$, the type \Romannum{1} error rates are reasonably well controlled for all the test statistics considered. When $\gamma > 0$ and fixed, the powers tend to increase as the sample sizes increase. For increasing values of $\gamma$, the corresponding rejection ratios also increase. In general, standardized tests produce higher power compared to their unstandardized counterparts. This implies that a departure from the assumed AFT model can be more easily detected by the standardized test procedure than its unstandardized version. For example, when $\gamma = 0.2$ with the 20\% censoring and sample size 100, checking the functional form using the based on and induced smoothed method, the rejection ratio for the unstandardized and standardized tests are 0.022 and 0.233, respectively. For the increased censoring rate of 40\%, overall findings remain similar with smaller powers compared to those with the same sample sizes. The results for the non-smooth test are, in general, similar to those for the proposed test. We also observe that the results of the link function test and functional form test are very similar. As discussed in \ref{subsubsec:linktest}, this is expected because we consider only one continuous covariate.

    \begin{figure}[htp]
    \centering
    \begin{subfigure}[b]{0.45\textwidth}
        \centering
        \captionsetup{justification=raggedright,singlelinecheck = false}
	    \includegraphics[scale=0.39]{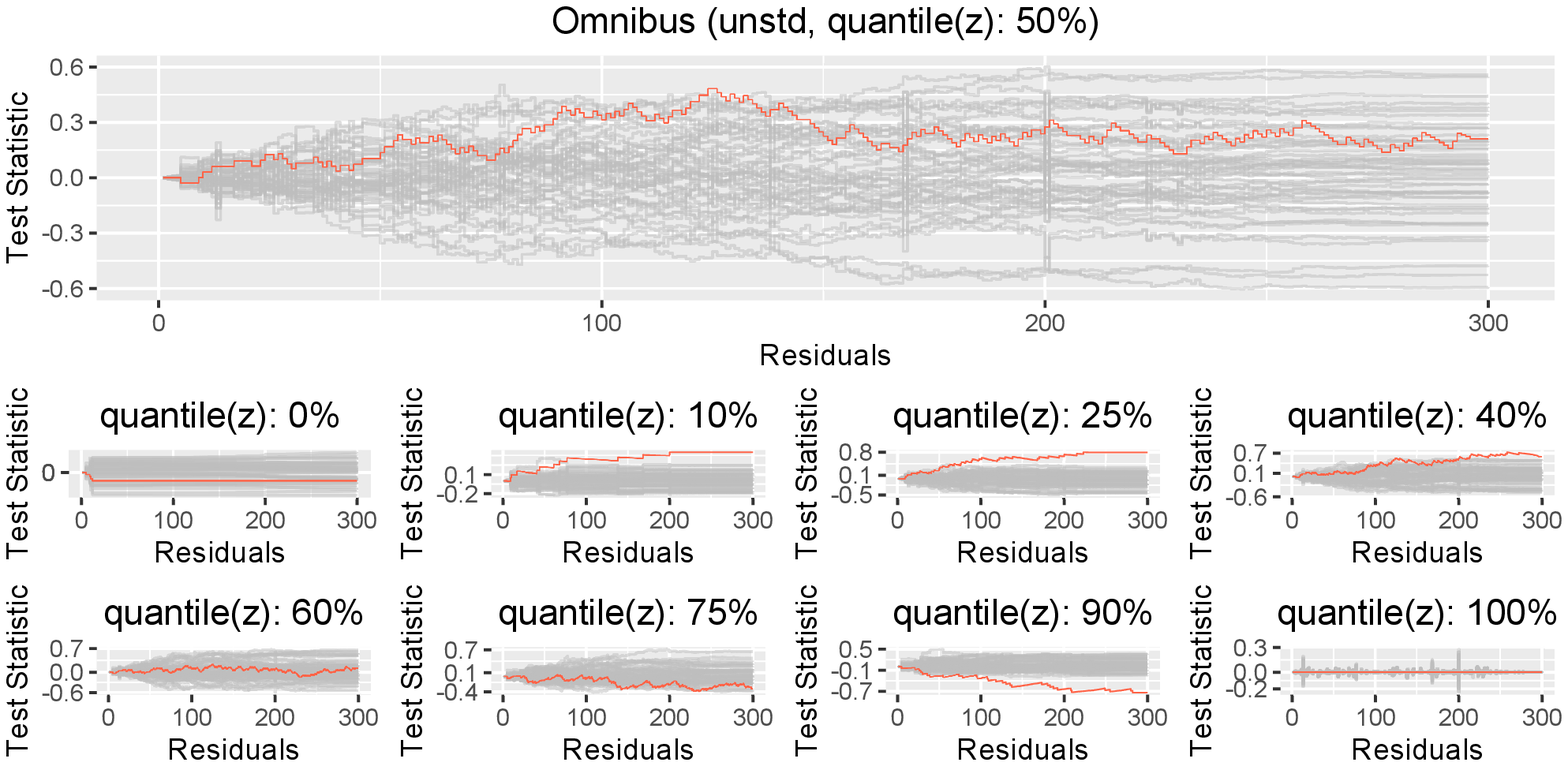} \\
	    \caption{afttestplot: unstandardized omnibus test with \texttt{mis} for model \eqref{model:sim1}} \label{fig:sim1_omni_mis_unstd} \vspace{0em}
    \end{subfigure}
    \hfill
    \begin{subfigure}[b]{0.45\textwidth}
        \centering
        \captionsetup{justification=raggedright,singlelinecheck = false}
	    \includegraphics[scale=0.39]{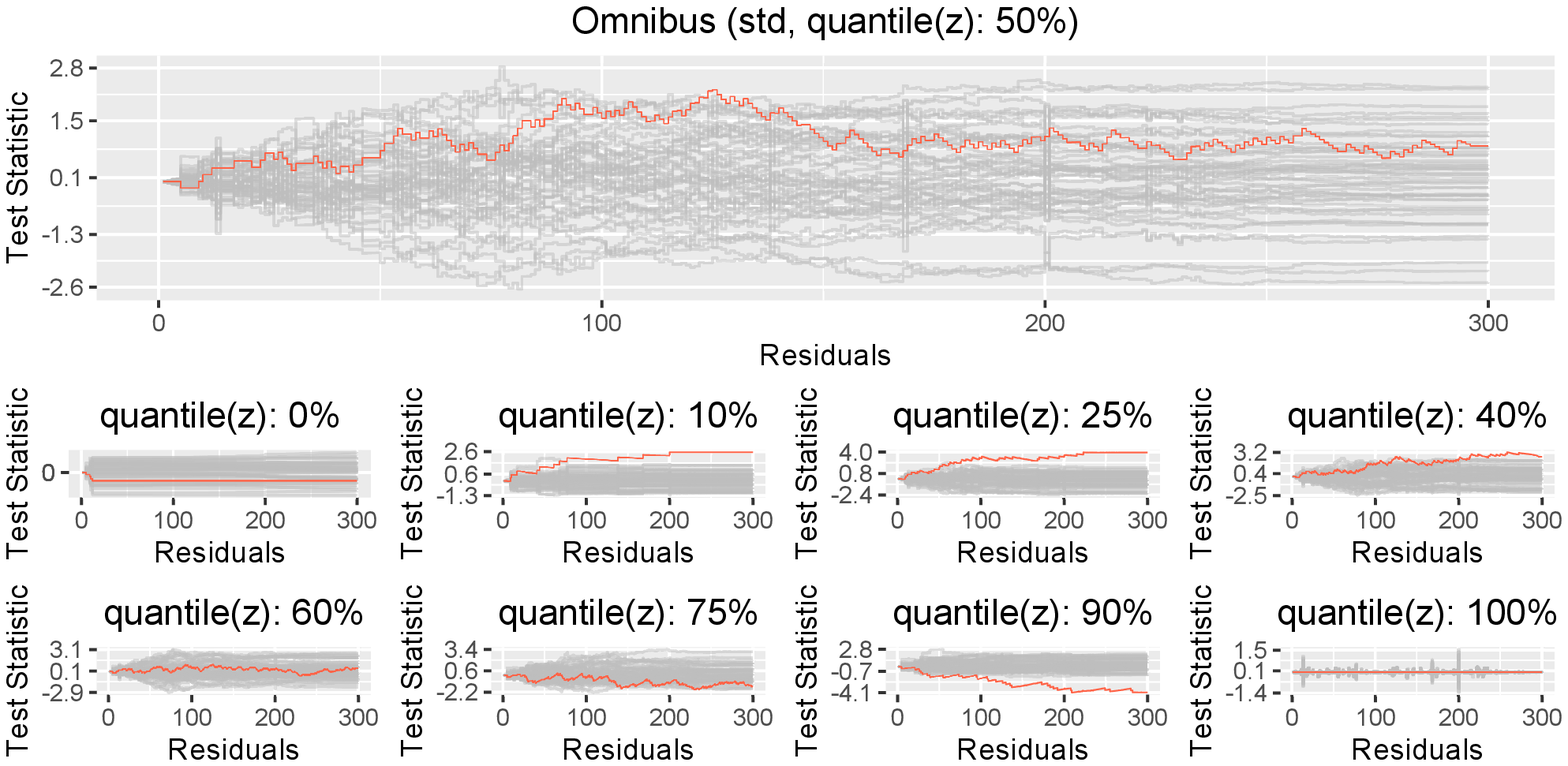} \\
	    \caption{afttestplot: standardized omnibus test with \texttt{mis} for model \eqref{model:sim1}} \label{fig:sim1_omni_mis_std} \vspace{0em}
    \end{subfigure}
    \\
    \begin{subfigure}[b]{0.45\textwidth}
        \centering
        \captionsetup{justification=raggedright,singlelinecheck = false}
	    \includegraphics[scale=0.39]{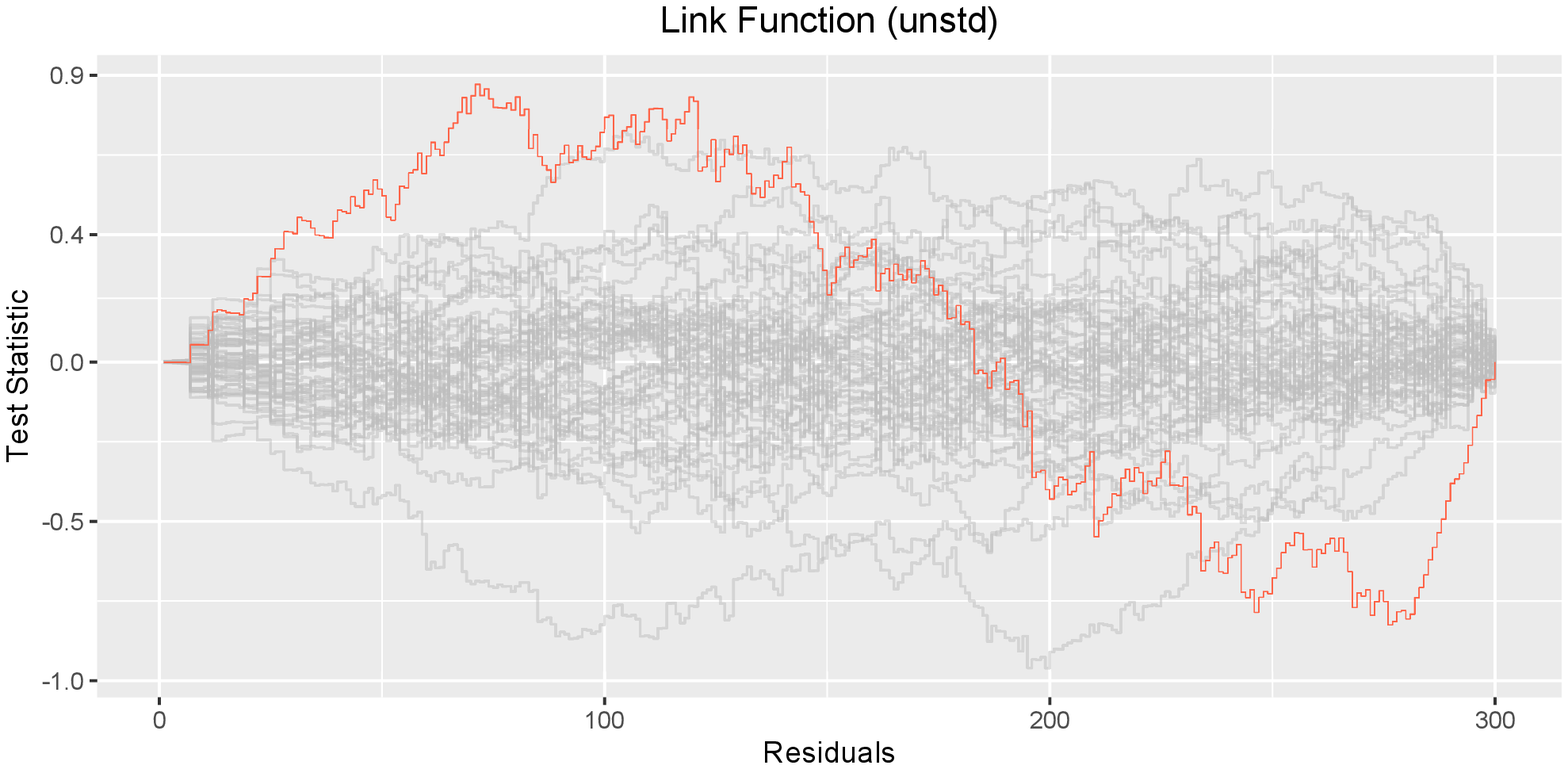} \\
	    \caption{afttestplot: unstandardized link function test with \texttt{mis} for model \eqref{model:sim1}} \label{fig:sim1_link_mis_std} \vspace{0em}
    \end{subfigure}
    \hfill
    \begin{subfigure}[b]{0.45\textwidth}
        \centering
        \captionsetup{justification=raggedright,singlelinecheck = false}
	    \includegraphics[scale=0.39]{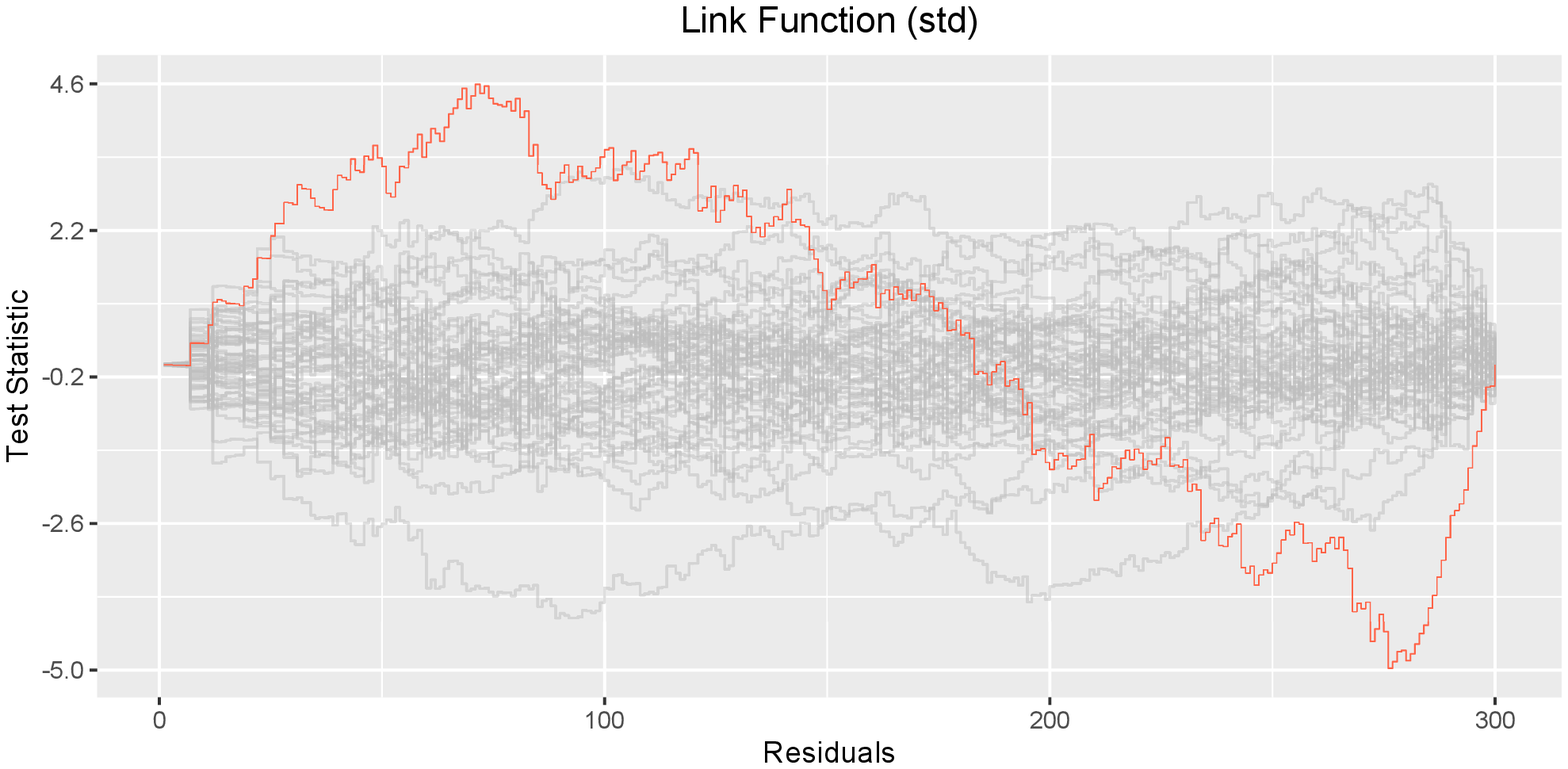} \\
	    \caption{afttestplot: standardized link function test with \texttt{mis} for model \eqref{model:sim1}} \label{fig:sim1_link_mis_std} \vspace{0em}
    \end{subfigure}
    \\
    \begin{subfigure}[b]{0.45\textwidth}
        \centering
        \captionsetup{justification=raggedright,singlelinecheck = false}
	    \includegraphics[scale=0.39]{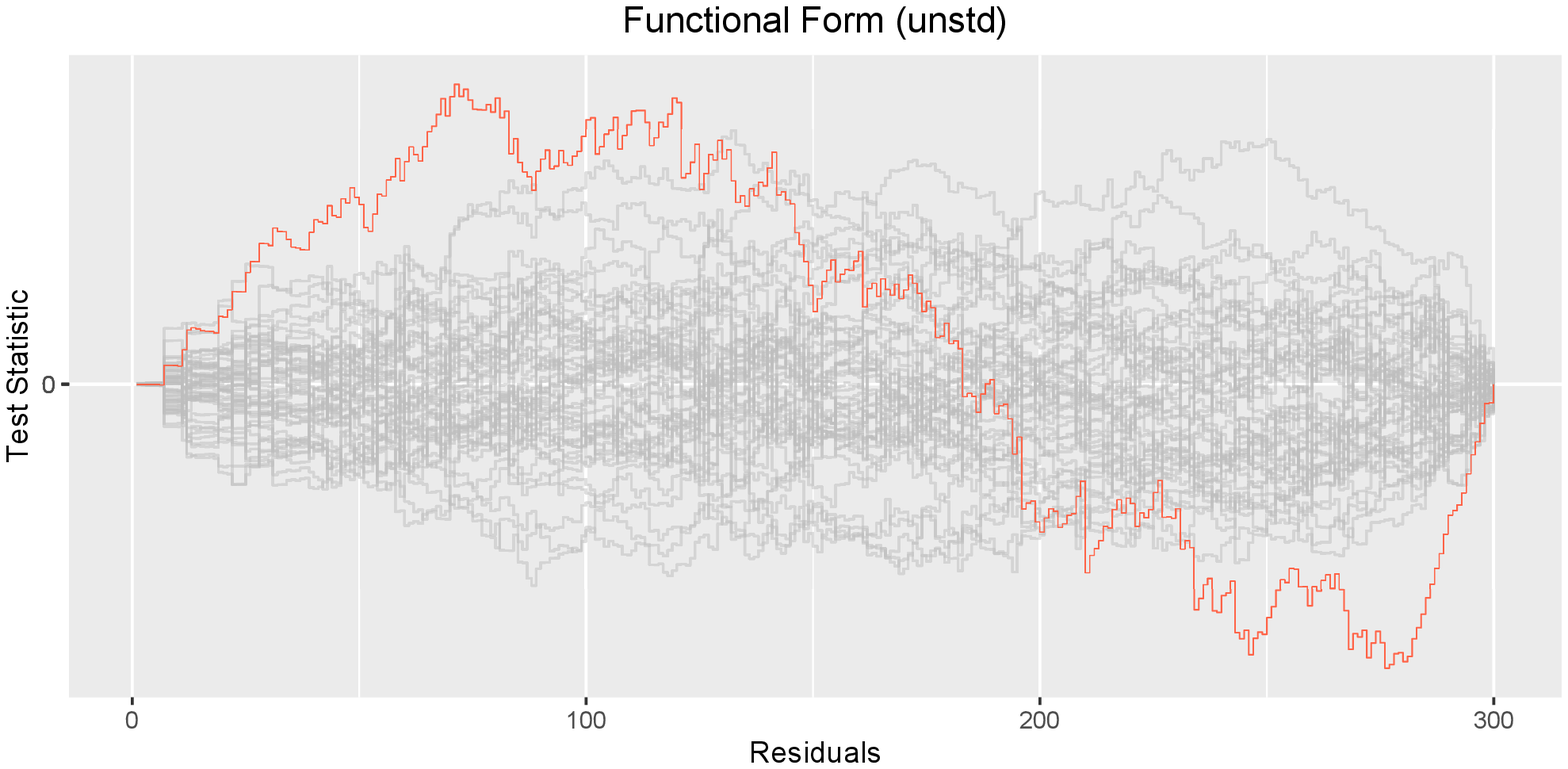} \\
	    \caption{afttestplot: unstandardized functional form test with \texttt{mis} for model \eqref{model:sim1}} \label{fig:sim1_form_mis_unstd} \vspace{0em}
    \end{subfigure}
    \hfill
    \begin{subfigure}[b]{0.45\textwidth}
        \centering
        \captionsetup{justification=raggedright,singlelinecheck = false}
	    \includegraphics[scale=0.39]{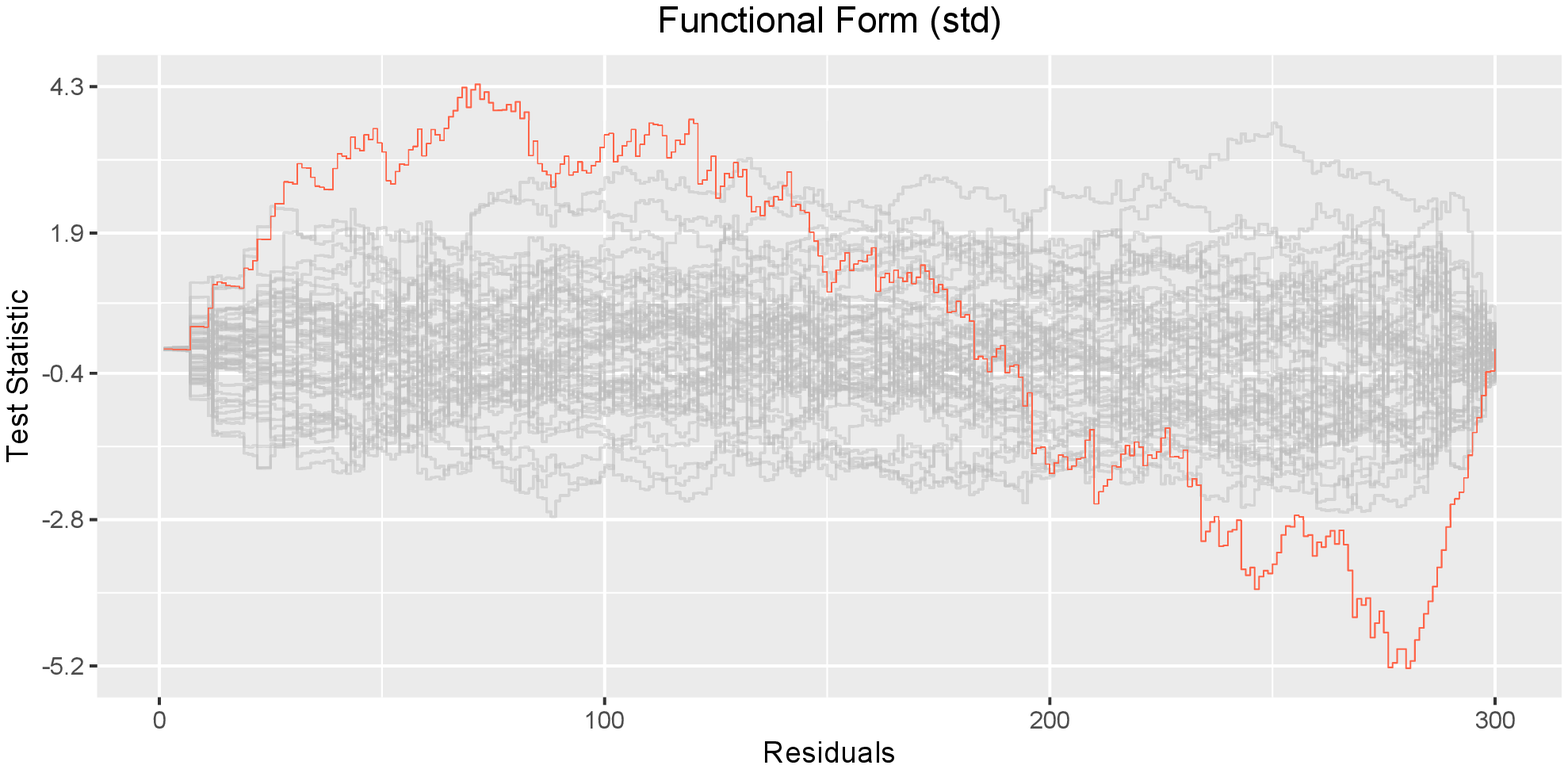} \\
	    \caption{afttestplot: standardized functional form test with \texttt{mis} for model \eqref{model:sim1}} \label{fig:sim1_form_mis_std} \vspace{0em}
    \end{subfigure}
    \caption{\label{fig:sim1:mis} Plots of sample paths based on the proposed unstandardized and standardized test procedures for Simulation scenario 1 with $\gamma = 0.3$, 20\% censoring rate and the sample size of 300 (induced-smoothed). $x$-axis indicates the rank of the log transformed residuals. The sample path for the test statistic (red) overlays by 50 approximated sample paths under the null (grey) are displayed. Top panel is the plot for the omnibus test under the different quantiles of covariates. Bottom two panels are the plots for testing the link function (left) and functional form (right)}
\end{figure}
    Figure \ref{fig:sim1:mis} present plots of sample paths for the unstandardized and standardized test procedures, respectively, for Simulation \hyperlink{scenario1}{Scenario 1}. Both test procedures are based on a 20\% censoring rate, a sample size of 300, and a $\gamma = 0.3$. The plots show the test statistic's sample path based on the observed data in red, with 50 simulated sample paths under the null hypothesis in grey. The plots are divided into three panels, showing the omnibus, link function, and functional form tests. We observe some departure from the null at the lower and upper ranks of log-transformed residuals when testing the link function and functional form using the unstandardized and standardized test procedure in Figure \ref{fig:sim1:mis} ((a), (c), and (e)). These departures become more obvious in Figure \ref{fig:sim1:mis}  ((b), (d), and (f)), where the standardized test procedure is considered.
    
\subsection{Simulation scenario 2} \label{subsec:sim2}
    In Simulation \hypertarget{scenario2}{Scenario 2}, we consider a setting similar to that in Simulation \hyperlink{scenario1}{Scenario 1} but  with the addition of an extra covariate. Specifically, 
    \begin{align} \tag{\ref{subsec:sim2}.1}
        \log T_{i} = - \beta_{0} - \beta_{1} Z_{1i} - \beta_{2} Z_{2i} - \gamma Z^{2}_{2i} + \epsilon_{i}, \label{model:sim2}
    \end{align}
    where $Z_{1i} \in \left\{ 0, 1 \right\}$ is Bernoulli with probability $0.5$ and $Z_{2i}$ is the normal random variable with mean $2$ and standard error $1$. Each parameter is set to be $\beta_{0} = - 4$, $\beta_{1} = 1$, $\beta_{2} = 1$ and $\gamma$ varies from 0 to 0.5 with 0.1 increment as in Simulation \hyperlink{scenario1}{Scenario 1}. 
    \begin{table}[htp]
    \fontsize{8}{8}\selectfont
    \caption{\label{tab:sim2:result} Results for the omnibus (omni), link function (link), and functional form (form) tests under Simulation scenario 2. Monotone non-smooth method (mns) and monotone induced smoothed method (mis) are considered under the sample sizes of $n = 100, 300$, and $500$ with two different censoring rates of 20\% and 40\%. Each number in a cell represents the rejection ratio based on 1,000 replications for the standardized statistic (bold) or unstandardized statistic.}
    \centering
    \begin{tabular}[htp]{cccccccccccccc}
        \toprule
        \multicolumn{2}{c}{censoring} & \multicolumn{6}{c}{20\%} & \multicolumn{6}{c}{40\%} \\
        \cmidrule(l{3pt}r{3pt}){1-2} \cmidrule(l{3pt}r{3pt}){3-8} \cmidrule(l{3pt}r{3pt}){9-14}
        \multicolumn{2}{c}{n} & \multicolumn{2}{c}{100} & \multicolumn{2}{c}{300} & \multicolumn{2}{c}{500} & \multicolumn{2}{c}{100} & \multicolumn{2}{c}{300} & \multicolumn{2}{c}{500} \\
        \cmidrule(l{3pt}r{3pt}){1-2} \cmidrule(l{3pt}r{3pt}){3-4} \cmidrule(l{3pt}r{3pt}){5-6} \cmidrule(l{3pt}r{3pt}){7-8} \cmidrule(l{3pt}r{3pt}){9-10} \cmidrule(l{3pt}r{3pt}){11-12} \cmidrule(l{3pt}r{3pt}){13-14}
        $\gamma$ & test & mns & mis & mns & mis & mns & mis & mns & mis & mns & mis & mns & mis\\
        \midrule
        \textbf{} & \textbf{} & \textbf{0.005} & \textbf{0.007} & \textbf{0.012} & \textbf{0.011} & \textbf{0.009} & \textbf{0.009} & \textbf{0.002} & \textbf{0.002} & \textbf{0.006} & \textbf{0.006} & \textbf{0.013} & \textbf{0.014}\\
        
         & \multirow{-2}{*}{omni} & 0.002 & 0.002 & 0.004 & 0.004 & 0.004 & 0.004 & 0.002 & 0.002 & 0.004 & 0.006 & 0.007 & 0.006\\
        \cmidrule(l{3pt}r{3pt}){2-2} \cmidrule(l{3pt}r{3pt}){3-4} \cmidrule(l{3pt}r{3pt}){5-6} \cmidrule(l{3pt}r{3pt}){3-4} \cmidrule(l{3pt}r{3pt}){5-6} \cmidrule(l{3pt}r{3pt}){7-8} \cmidrule(l{3pt}r{3pt}){9-10} \cmidrule(l{3pt}r{3pt}){11-12} \cmidrule(l{3pt}r{3pt}){13-14}
        
        \textbf{} & \textbf{} & \textbf{0.020} & \textbf{0.018} & \textbf{0.032} & \textbf{0.032} & \textbf{0.034} & \textbf{0.035} & \textbf{0.011} & \textbf{0.010} & \textbf{0.032} & \textbf{0.035} & \textbf{0.035} & \textbf{0.034}\\
        
         & \multirow{-2}{*}{link} & 0.006 & 0.008 & 0.006 & 0.006 & 0.009 & 0.010 & 0.008 & 0.006 & 0.007 & 0.008 & 0.011 & 0.013\\
        \cmidrule(l{3pt}r{3pt}){2-2} \cmidrule(l{3pt}r{3pt}){3-4} \cmidrule(l{3pt}r{3pt}){5-6} \cmidrule(l{3pt}r{3pt}){3-4} \cmidrule(l{3pt}r{3pt}){5-6} \cmidrule(l{3pt}r{3pt}){7-8} \cmidrule(l{3pt}r{3pt}){9-10} \cmidrule(l{3pt}r{3pt}){11-12} \cmidrule(l{3pt}r{3pt}){13-14}
        
        \textbf{} & \textbf{} & \textbf{0.024} & \textbf{0.023} & \textbf{0.024} & \textbf{0.024} & \textbf{0.022} & \textbf{0.022} & \textbf{0.014} & \textbf{0.015} & \textbf{0.024} & \textbf{0.021} & \textbf{0.021} & \textbf{0.020}\\
        
        \multirow{-6}{*}{0} & \multirow{-2}{*}{form} & 0.005 & 0.005 & 0.002 & 0.004 & 0.006 & 0.006 & 0.007 & 0.005 & 0.007 & 0.007 & 0.009 & 0.007\\
        \cmidrule(l{3pt}r{3pt}){1-14}
        
        \textbf{} & \textbf{} & \textbf{0.018} & \textbf{0.018} & \textbf{0.084} & \textbf{0.086} & \textbf{0.188} & \textbf{0.186} & \textbf{0.002} & \textbf{0.002} & \textbf{0.029} & \textbf{0.031} & \textbf{0.091} & \textbf{0.092}\\
        
         & \multirow{-2}{*}{omni} & 0.003 & 0.002 & 0.010 & 0.012 & 0.021 & 0.018 & 0.005 & 0.004 & 0.009 & 0.010 & 0.023 & 0.020\\
        \cmidrule(l{3pt}r{3pt}){2-2} \cmidrule(l{3pt}r{3pt}){3-4} \cmidrule(l{3pt}r{3pt}){5-6} \cmidrule(l{3pt}r{3pt}){3-4} \cmidrule(l{3pt}r{3pt}){5-6} \cmidrule(l{3pt}r{3pt}){7-8} \cmidrule(l{3pt}r{3pt}){9-10} \cmidrule(l{3pt}r{3pt}){11-12} \cmidrule(l{3pt}r{3pt}){13-14}
        
        \textbf{} & \textbf{} & \textbf{0.070} & \textbf{0.068} & \textbf{0.212} & \textbf{0.208} & \textbf{0.328} & \textbf{0.321} & \textbf{0.034} & \textbf{0.032} & \textbf{0.153} & \textbf{0.154} & \textbf{0.252} & \textbf{0.245}\\
        
         & \multirow{-2}{*}{link} & 0.007 & 0.007 & 0.014 & 0.015 & 0.037 & 0.033 & 0.008 & 0.007 & 0.016 & 0.015 & 0.035 & 0.035\\
        \cmidrule(l{3pt}r{3pt}){2-2} \cmidrule(l{3pt}r{3pt}){3-4} \cmidrule(l{3pt}r{3pt}){5-6} \cmidrule(l{3pt}r{3pt}){3-4} \cmidrule(l{3pt}r{3pt}){5-6} \cmidrule(l{3pt}r{3pt}){7-8} \cmidrule(l{3pt}r{3pt}){9-10} \cmidrule(l{3pt}r{3pt}){11-12} \cmidrule(l{3pt}r{3pt}){13-14}
        
        \textbf{} & \textbf{} & \textbf{0.088} & \textbf{0.093} & \textbf{0.180} & \textbf{0.176} & \textbf{0.302} & \textbf{0.300} & \textbf{0.034} & \textbf{0.038} & \textbf{0.152} & \textbf{0.146} & \textbf{0.234} & \textbf{0.235}\\
        
        \multirow{-6}{*}{0.1} & \multirow{-2}{*}{form} & 0.007 & 0.006 & 0.012 & 0.010 & 0.028 & 0.025 & 0.006 & 0.009 & 0.015 & 0.015 & 0.032 & 0.032\\
        \cmidrule(l{3pt}r{3pt}){1-14}
        
        \textbf{} & \textbf{} & \textbf{0.050} & \textbf{0.048} & \textbf{0.346} & \textbf{0.336} & \textbf{0.648} & \textbf{0.648} & \textbf{0.007} & \textbf{0.004} & \textbf{0.125} & \textbf{0.122} & \textbf{0.324} & \textbf{0.328}\\
        
         & \multirow{-2}{*}{omni} & 0.005 & 0.007 & 0.033 & 0.034 & 0.116 & 0.112 & 0.006 & 0.006 & 0.025 & 0.024 & 0.072 & 0.070\\
        \cmidrule(l{3pt}r{3pt}){2-2} \cmidrule(l{3pt}r{3pt}){3-4} \cmidrule(l{3pt}r{3pt}){5-6} \cmidrule(l{3pt}r{3pt}){3-4} \cmidrule(l{3pt}r{3pt}){5-6} \cmidrule(l{3pt}r{3pt}){7-8} \cmidrule(l{3pt}r{3pt}){9-10} \cmidrule(l{3pt}r{3pt}){11-12} \cmidrule(l{3pt}r{3pt}){13-14}
        
        \textbf{} & \textbf{} & \textbf{0.151} & \textbf{0.146} & \textbf{0.543} & \textbf{0.550} & \textbf{0.760} & \textbf{0.752} & \textbf{0.068} & \textbf{0.064} & \textbf{0.330} & \textbf{0.330} & \textbf{0.569} & \textbf{0.565}\\
        
         & \multirow{-2}{*}{link} & 0.015 & 0.016 & 0.053 & 0.053 & 0.150 & 0.150 & 0.012 & 0.011 & 0.038 & 0.040 & 0.094 & 0.090\\
        \cmidrule(l{3pt}r{3pt}){2-2} \cmidrule(l{3pt}r{3pt}){3-4} \cmidrule(l{3pt}r{3pt}){5-6} \cmidrule(l{3pt}r{3pt}){3-4} \cmidrule(l{3pt}r{3pt}){5-6} \cmidrule(l{3pt}r{3pt}){7-8} \cmidrule(l{3pt}r{3pt}){9-10} \cmidrule(l{3pt}r{3pt}){11-12} \cmidrule(l{3pt}r{3pt}){13-14}
        
        \textbf{} & \textbf{} & \textbf{0.198} & \textbf{0.202} & \textbf{0.565} & \textbf{0.558} & \textbf{0.784} & \textbf{0.779} & \textbf{0.080} & \textbf{0.087} & \textbf{0.390} & \textbf{0.384} & \textbf{0.625} & \textbf{0.620}\\
        
        \multirow{-6}{*}{0.2} & \multirow{-2}{*}{form} & 0.014 & 0.013 & 0.048 & 0.050 & 0.164 & 0.161 & 0.013 & 0.011 & 0.041 & 0.035 & 0.097 & 0.095\\
        \cmidrule(l{3pt}r{3pt}){1-14}
        
        \textbf{} & \textbf{} & \textbf{0.104} & \textbf{0.109} & \textbf{0.694} & \textbf{0.707} & \textbf{0.941} & \textbf{0.942} & \textbf{0.011} & \textbf{0.011} & \textbf{0.222} & \textbf{0.206} & \textbf{0.574} & \textbf{0.574}\\
        
         & \multirow{-2}{*}{omni} & 0.009 & 0.010 & 0.088 & 0.096 & 0.297 & 0.288 & 0.010 & 0.007 & 0.040 & 0.040 & 0.140 & 0.132\\
        \cmidrule(l{3pt}r{3pt}){2-2} \cmidrule(l{3pt}r{3pt}){3-4} \cmidrule(l{3pt}r{3pt}){5-6} \cmidrule(l{3pt}r{3pt}){3-4} \cmidrule(l{3pt}r{3pt}){5-6} \cmidrule(l{3pt}r{3pt}){7-8} \cmidrule(l{3pt}r{3pt}){9-10} \cmidrule(l{3pt}r{3pt}){11-12} \cmidrule(l{3pt}r{3pt}){13-14}
        
        \textbf{} & \textbf{} & \textbf{0.273} & \textbf{0.258} & \textbf{0.807} & \textbf{0.814} & \textbf{0.955} & \textbf{0.955} & \textbf{0.092} & \textbf{0.089} & \textbf{0.486} & \textbf{0.489} & \textbf{0.763} & \textbf{0.757}\\
        
         & \multirow{-2}{*}{link} & 0.024 & 0.022 & 0.137 & 0.143 & 0.370 & 0.364 & 0.018 & 0.018 & 0.062 & 0.058 & 0.180 & 0.168\\
        \cmidrule(l{3pt}r{3pt}){2-2} \cmidrule(l{3pt}r{3pt}){3-4} \cmidrule(l{3pt}r{3pt}){5-6} \cmidrule(l{3pt}r{3pt}){3-4} \cmidrule(l{3pt}r{3pt}){5-6} \cmidrule(l{3pt}r{3pt}){7-8} \cmidrule(l{3pt}r{3pt}){9-10} \cmidrule(l{3pt}r{3pt}){11-12} \cmidrule(l{3pt}r{3pt}){13-14}
        
        \textbf{} & \textbf{} & \textbf{0.342} & \textbf{0.345} & \textbf{0.857} & \textbf{0.853} & \textbf{0.971} & \textbf{0.973} & \textbf{0.115} & \textbf{0.110} & \textbf{0.580} & \textbf{0.582} & \textbf{0.849} & \textbf{0.842}\\
        
        \multirow{-6}{*}{0.3} & \multirow{-2}{*}{form} & 0.023 & 0.023 & 0.158 & 0.169 & 0.453 & 0.444 & 0.020 & 0.021 & 0.064 & 0.062 & 0.188 & 0.184\\
        \cmidrule(l{3pt}r{3pt}){1-14}
        
        \textbf{} & \textbf{} & \textbf{0.176} & \textbf{0.172} & \textbf{0.876} & \textbf{0.873} & \textbf{0.993} & \textbf{0.993} & \textbf{0.012} & \textbf{0.014} & \textbf{0.355} & \textbf{0.353} & \textbf{0.765} & \textbf{0.765}\\
        
         & \multirow{-2}{*}{omni} & 0.015 & 0.013 & 0.160 & 0.168 & 0.513 & 0.518 & 0.007 & 0.005 & 0.062 & 0.060 & 0.218 & 0.212\\
        \cmidrule(l{3pt}r{3pt}){2-2} \cmidrule(l{3pt}r{3pt}){3-4} \cmidrule(l{3pt}r{3pt}){5-6} \cmidrule(l{3pt}r{3pt}){3-4} \cmidrule(l{3pt}r{3pt}){5-6} \cmidrule(l{3pt}r{3pt}){7-8} \cmidrule(l{3pt}r{3pt}){9-10} \cmidrule(l{3pt}r{3pt}){11-12} \cmidrule(l{3pt}r{3pt}){13-14}
        
        \textbf{} & \textbf{} & \textbf{0.355} & \textbf{0.358} & \textbf{0.922} & \textbf{0.920} & \textbf{0.994} & \textbf{0.997} & \textbf{0.121} & \textbf{0.124} & \textbf{0.628} & \textbf{0.636} & \textbf{0.878} & \textbf{0.876}\\
        
         & \multirow{-2}{*}{link} & 0.031 & 0.035 & 0.245 & 0.253 & 0.605 & 0.605 & 0.023 & 0.023 & 0.100 & 0.101 & 0.268 & 0.272\\
        \cmidrule(l{3pt}r{3pt}){2-2} \cmidrule(l{3pt}r{3pt}){3-4} \cmidrule(l{3pt}r{3pt}){5-6} \cmidrule(l{3pt}r{3pt}){3-4} \cmidrule(l{3pt}r{3pt}){5-6} \cmidrule(l{3pt}r{3pt}){7-8} \cmidrule(l{3pt}r{3pt}){9-10} \cmidrule(l{3pt}r{3pt}){11-12} \cmidrule(l{3pt}r{3pt}){13-14}
        
        \textbf{} & \textbf{} & \textbf{0.462} & \textbf{0.462} & \textbf{0.962} & \textbf{0.965} & \textbf{0.999} & \textbf{0.998} & \textbf{0.158} & \textbf{0.156} & \textbf{0.748} & \textbf{0.738} & \textbf{0.943} & \textbf{0.948}\\
        
        \multirow{-6}{*}{0.4} & \multirow{-2}{*}{form} & 0.040 & 0.040 & 0.318 & 0.342 & 0.734 & 0.736 & 0.024 & 0.021 & 0.111 & 0.109 & 0.306 & 0.310\\
        \cmidrule(l{3pt}r{3pt}){1-14}
        
        \textbf{} & \textbf{} & \textbf{0.242} & \textbf{0.239} & \textbf{0.958} & \textbf{0.957} & \textbf{0.999} & \textbf{1.000} & \textbf{0.016} & \textbf{0.015} & \textbf{0.460} & \textbf{0.452} & \textbf{0.878} & \textbf{0.869}\\
        
         & \multirow{-2}{*}{omni} & 0.020 & 0.018 & 0.270 & 0.276 & 0.730 & 0.751 & 0.010 & 0.010 & 0.086 & 0.090 & 0.320 & 0.314\\
        \cmidrule(l{3pt}r{3pt}){2-2} \cmidrule(l{3pt}r{3pt}){3-4} \cmidrule(l{3pt}r{3pt}){5-6} \cmidrule(l{3pt}r{3pt}){3-4} \cmidrule(l{3pt}r{3pt}){5-6} \cmidrule(l{3pt}r{3pt}){7-8} \cmidrule(l{3pt}r{3pt}){9-10} \cmidrule(l{3pt}r{3pt}){11-12} \cmidrule(l{3pt}r{3pt}){13-14}
        
        \textbf{} & \textbf{} & \textbf{0.438} & \textbf{0.440} & \textbf{0.970} & \textbf{0.973} & \textbf{0.999} & \textbf{0.999} & \textbf{0.151} & \textbf{0.136} & \textbf{0.730} & \textbf{0.730} & \textbf{0.940} & \textbf{0.945}\\
        
         & \multirow{-2}{*}{link} & 0.047 & 0.043 & 0.358 & 0.370 & 0.797 & 0.810 & 0.024 & 0.022 & 0.142 & 0.137 & 0.388 & 0.392\\
        \cmidrule(l{3pt}r{3pt}){2-2} \cmidrule(l{3pt}r{3pt}){3-4} \cmidrule(l{3pt}r{3pt}){5-6} \cmidrule(l{3pt}r{3pt}){3-4} \cmidrule(l{3pt}r{3pt}){5-6} \cmidrule(l{3pt}r{3pt}){7-8} \cmidrule(l{3pt}r{3pt}){9-10} \cmidrule(l{3pt}r{3pt}){11-12} \cmidrule(l{3pt}r{3pt}){13-14}
        
        \textbf{} & \textbf{} & \textbf{0.573} & \textbf{0.569} & \textbf{0.992} & \textbf{0.994} & \textbf{1.000} & \textbf{1.000} & \textbf{0.194} & \textbf{0.208} & \textbf{0.834} & \textbf{0.836} & \textbf{0.984} & \textbf{0.982}\\
        
        \multirow{-6}{*}{0.5} & \multirow{-2}{*}{form} & 0.065 & 0.058 & 0.521 & 0.528 & 0.908 & 0.921 & 0.024 & 0.023 & 0.156 & 0.166 & 0.458 & 0.458\\
        \bottomrule
    \end{tabular}
\end{table} 
    Table \ref{tab:sim2:result} presents the results for the omnibus, link function, and functional form tests under Simulation \hyperlink{scenario2}{Scenario 2}. Overall, the findings are similar to those under Simulation \hyperlink{scenario1}{Scenario 1}, with the type \Romannum{1} error rates all under the nominal level $\alpha = 0.05$. Powers increase as the sample size or $\gamma$ increase. In most cases, the standardized version produces higher power than its unstandardized counterpart. The powers under Simulation \hyperlink{scenario2}{Scenario 2} are, in general, slightly smaller than those under Simulation \hyperlink{scenario1}{Scenario 1}. In addition to the results in Simulation \hyperlink{scenario1}{Scenario 1}, we observe that the powers of the link function and functional form tests are different in the presence of an added binary covariate, with the latter tending to be higher than the former.

    \begin{figure}[htp]
    \centering
    \begin{subfigure}[b]{0.45\textwidth}
        \centering
        \captionsetup{justification=raggedright,singlelinecheck = false}
	    \includegraphics[scale=0.39]{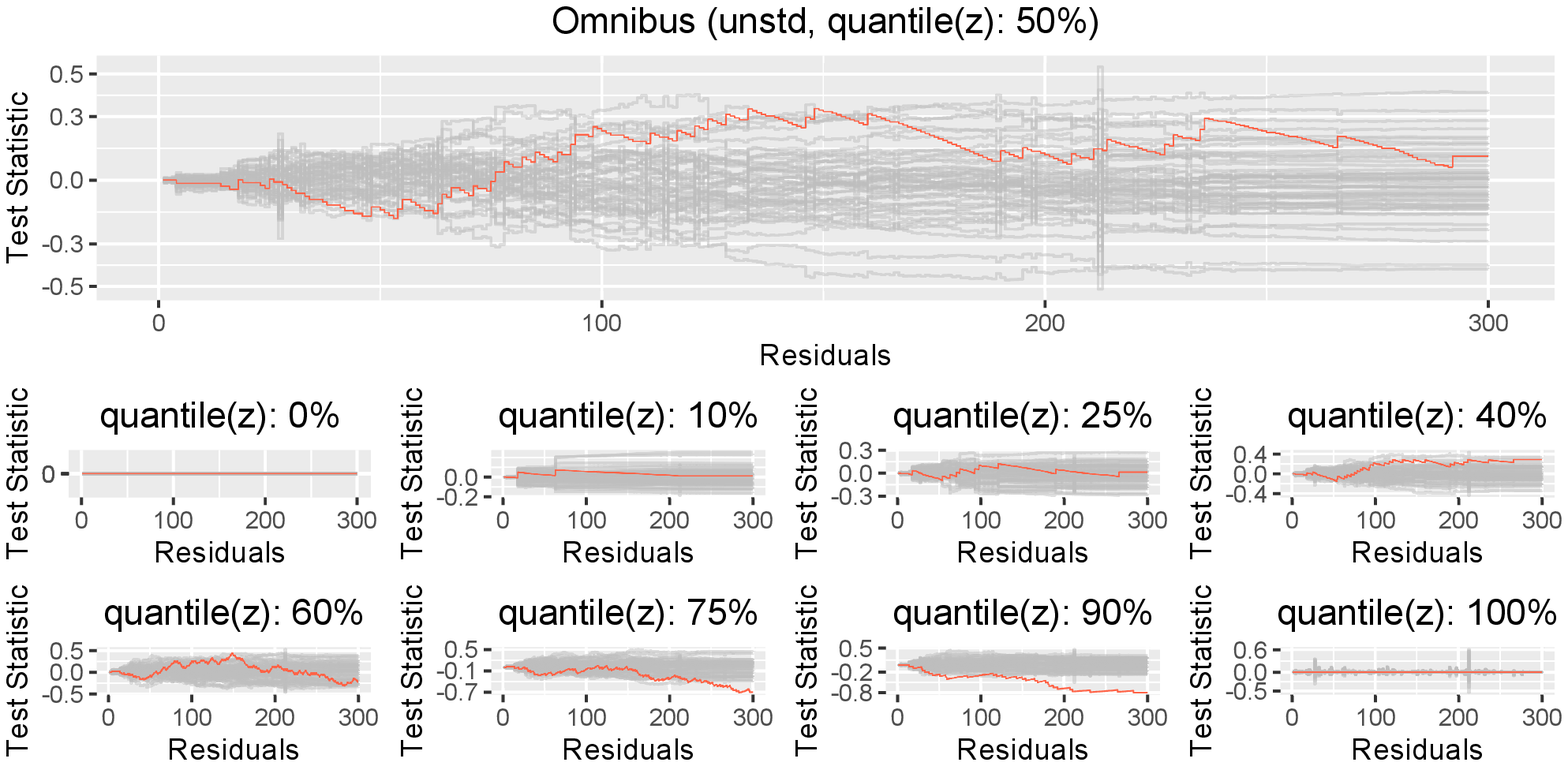} \\
	    \caption{afttestplot: unstandardized omnibus test with \texttt{mis} for model \eqref{model:sim2}} \label{fig:sim2_omni_mis_unstd} \vspace{0em}
    \end{subfigure}
    \hfill
    \begin{subfigure}[b]{0.45\textwidth}
        \centering
        \captionsetup{justification=raggedright,singlelinecheck = false}
	    \includegraphics[scale=0.39]{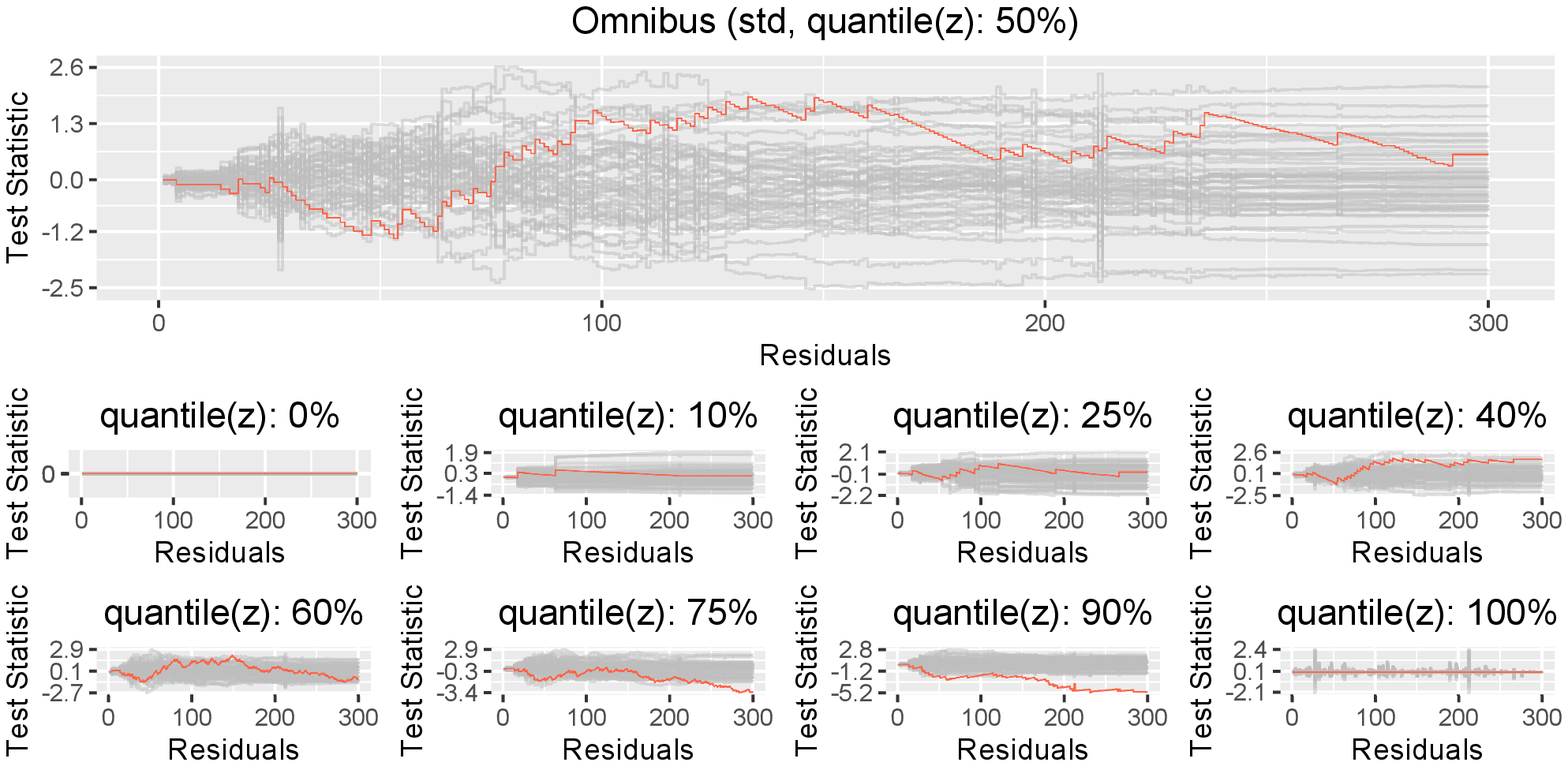} \\
	    \caption{afttestplot: standardized omnibus test with \texttt{mis} for model \eqref{model:sim2}} \label{fig:sim2_omni_mis_std} \vspace{0em}
    \end{subfigure}
    \\
    \begin{subfigure}[b]{0.45\textwidth}
        \centering
        \captionsetup{justification=raggedright,singlelinecheck = false}
	    \includegraphics[scale=0.39]{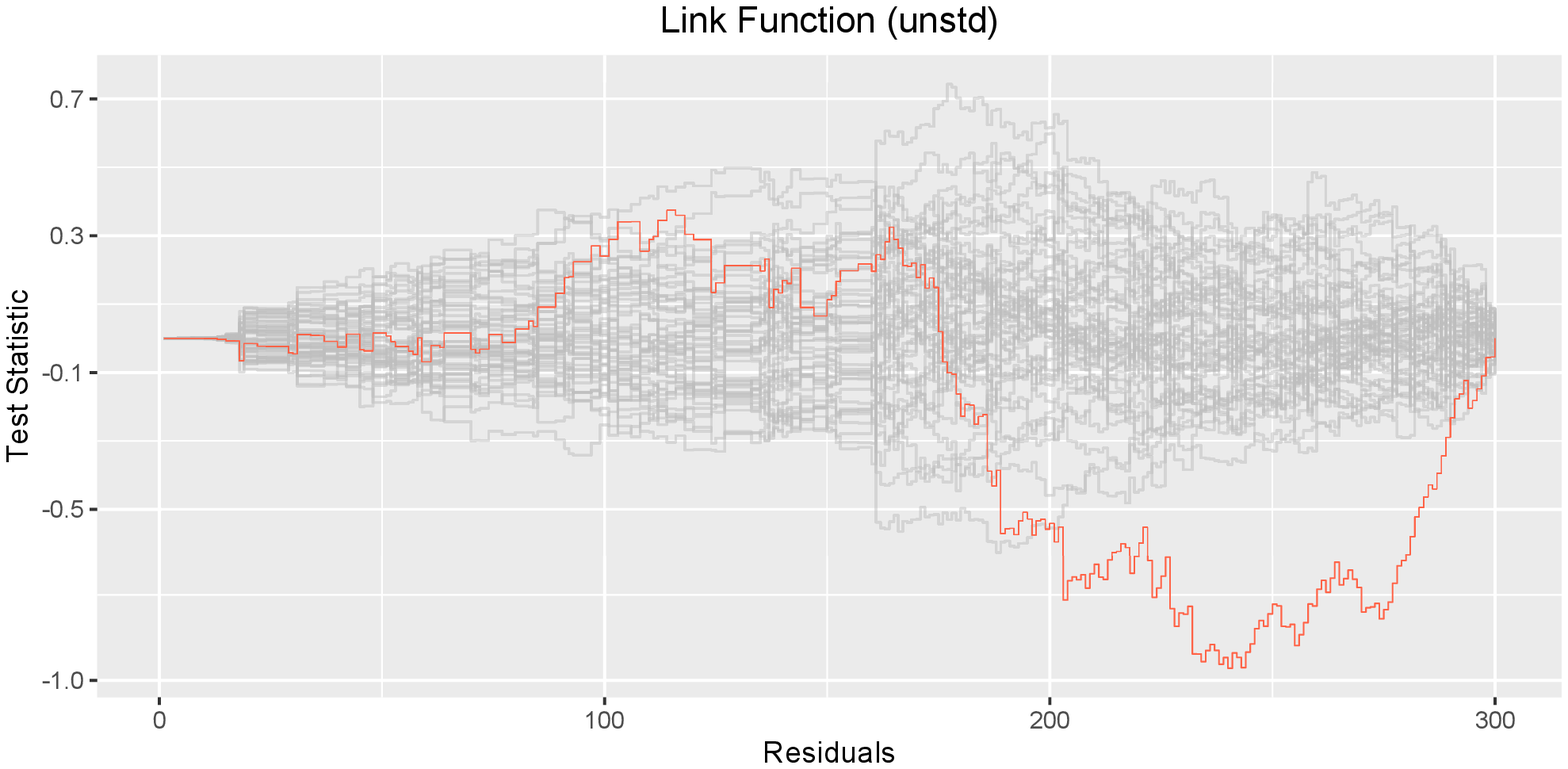} \\
	    \caption{afttestplot: unstandardized link function test with \texttt{mis} for model \eqref{model:sim2}} \label{fig:sim2_link_mis_std} \vspace{0em}
    \end{subfigure}
    \hfill
    \begin{subfigure}[b]{0.45\textwidth}
        \centering
        \captionsetup{justification=raggedright,singlelinecheck = false}
	    \includegraphics[scale=0.39]{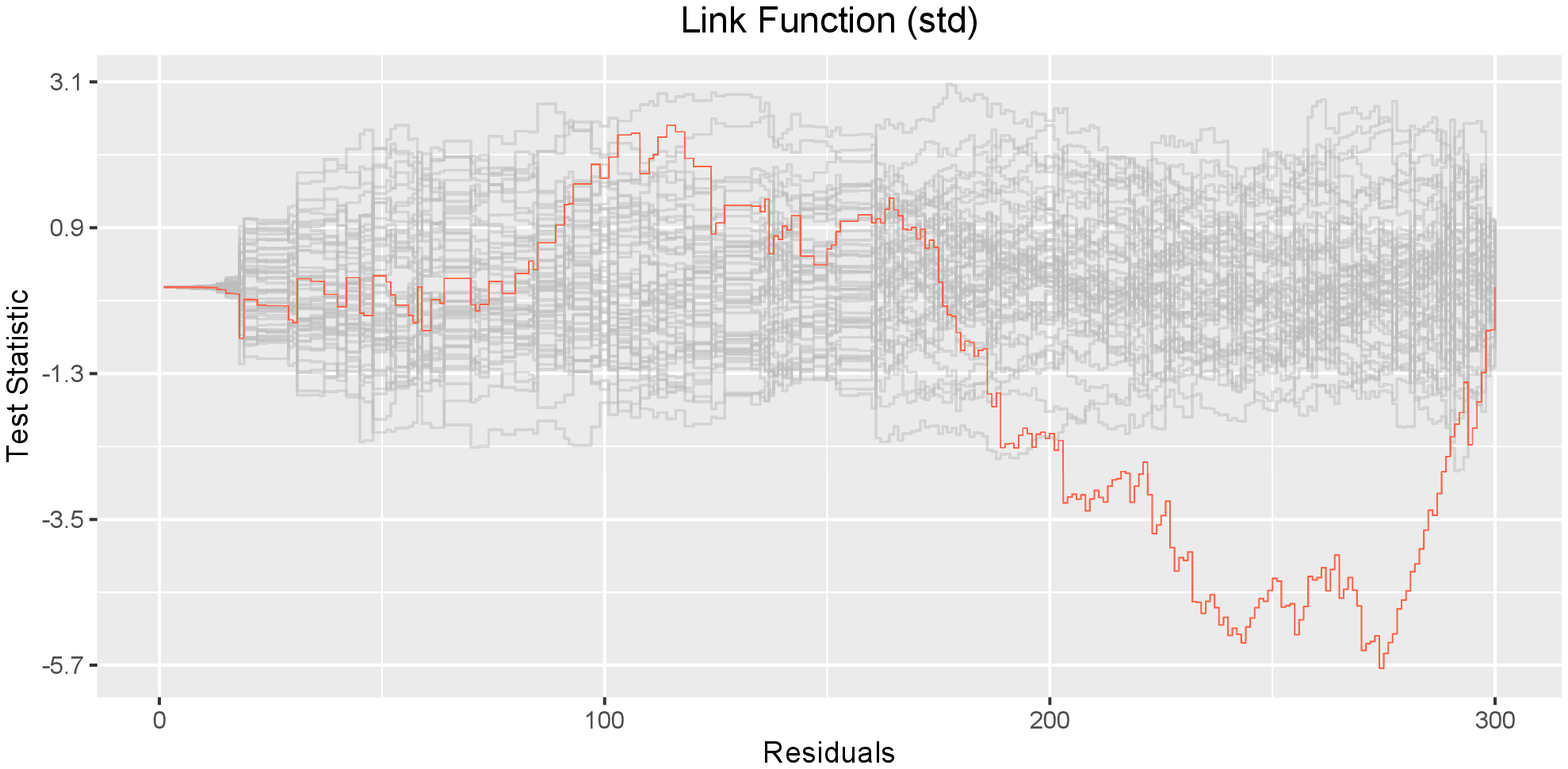} \\
	    \caption{afttestplot: standardized link function test with \texttt{mis} for model \eqref{model:sim2}} \label{fig:sim2_link_mis_std} \vspace{0em}
    \end{subfigure}
    \\
    \begin{subfigure}[b]{0.45\textwidth}
        \centering
        \captionsetup{justification=raggedright,singlelinecheck = false}
	    \includegraphics[scale=0.39]{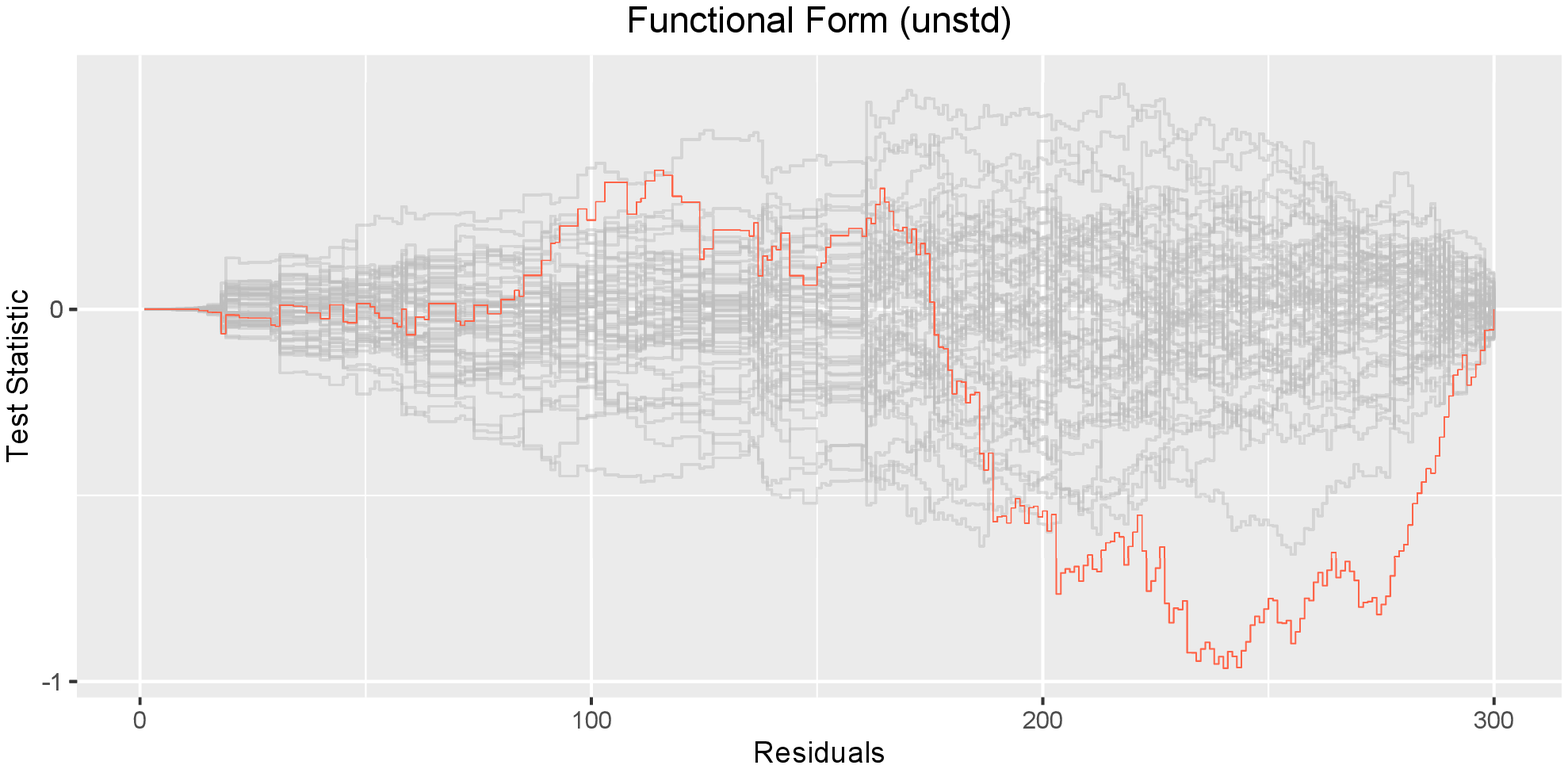} \\
	    \caption{afttestplot: unstandardized functional form test with \texttt{mis} for model \eqref{model:sim2}} \label{fig:sim2_form_mis_unstd} \vspace{0em}
    \end{subfigure}
    \hfill
    \begin{subfigure}[b]{0.45\textwidth}
        \centering
        \captionsetup{justification=raggedright,singlelinecheck = false}
	    \includegraphics[scale=0.39]{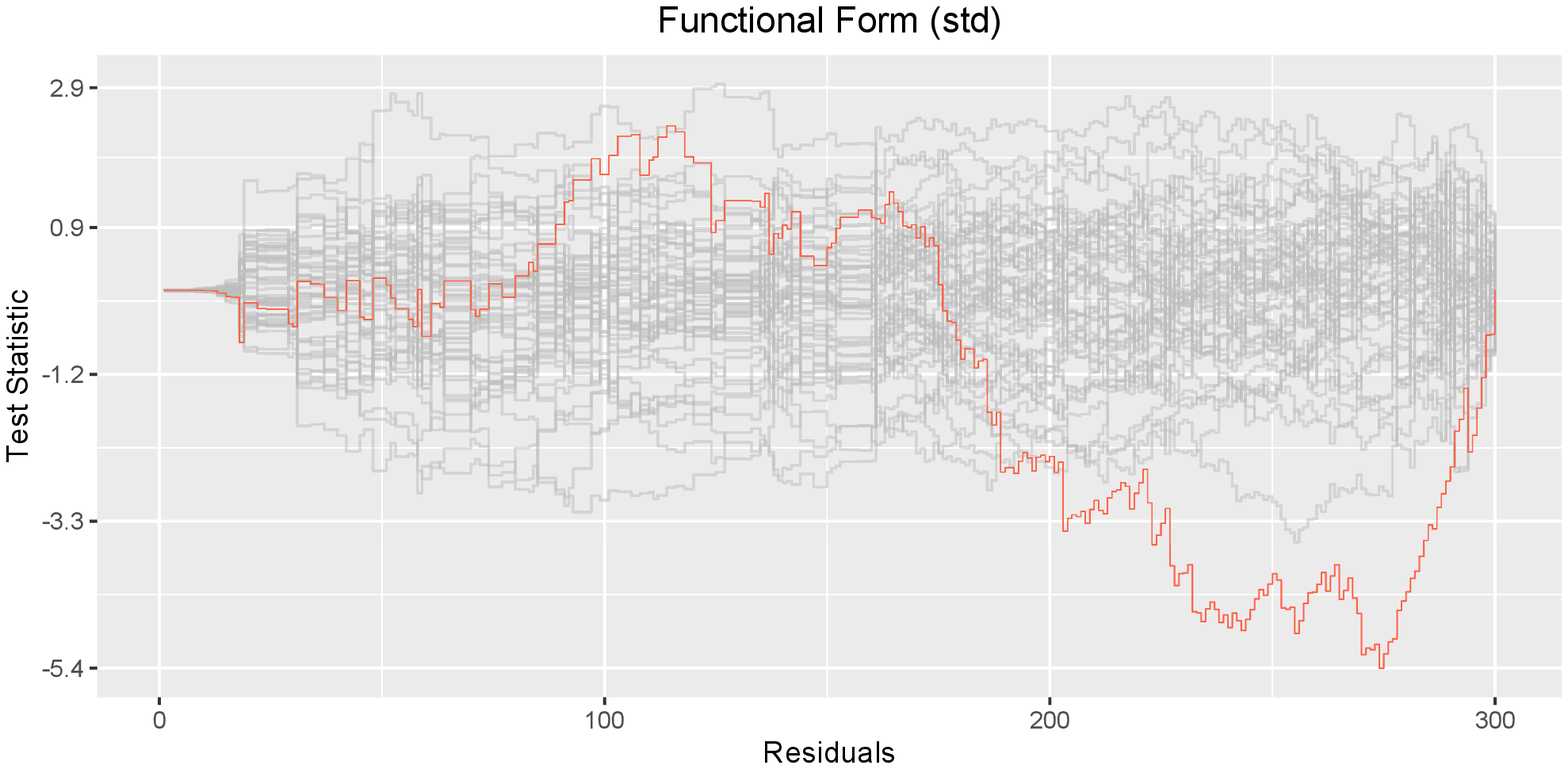} \\
	    \caption{afttestplot: standardized functional form test with \texttt{mis} for model \eqref{model:sim2}} \label{fig:sim2_form_mis_std} \vspace{0em}
    \end{subfigure}
    \caption{\label{fig:sim2:mis} Plots of sample paths based on the proposed unstandardized and standardized test procedures for Simulation scenario 2 with $\gamma = 0.3$, 20\% censoring rate and the sample size of 300 (induced-smoothed). $x$-axis indicates the rank of the log transformed residuals. The sample path for the test statistic (red) overlays by 50 approximated sample paths under the null (grey) are displayed. Top panel is the plot for the omnibus test under the different quantiles of covariates. Bottom two panels are the plots for testing the link function (left) and functional form (right)}
\end{figure}
    Figures \ref{fig:sim2:mis} present plots of sample paths for the unstandardized and standardized test procedures for Simulation \hyperlink{scenario2}{Scenario 2}, respectively, using a 20\% censoring rate, a sample size of 300, and a $\gamma = 0.3$. The plots show the test statistic's sample path based on the observed data in red, with 50 simulated sample paths under the null hypothesis in grey. The plots for the link function and functional form tests in both test procedure exhibit departures at higher ranks of log-transformed residuals, which is more obvious in the standardized version.
    
\section{Analysis of primary biliary cirrhosis (PBC) data} \label{sec:realdata}
    We illustrate our proposed model-checking procedures using the well-known primary biliary cirrhosis (PBC) study data \citep{fleming2011counting}. The PBC dataset consists of data on 418 patients diagnosed with PBC, a chronic liver disease that specifically affects the bile ducts. The dataset includes comprehensive information on each patient, including demographic variables such as age and sex, laboratory measurements such as serum bilirubin and albumin levels, and survival data such as the duration from diagnosis to liver transplant or death. The primary objective of the study was to investigate the natural course of PBC and to identify the prognostic factors that influence patient survival. For a more detailed description of the data and variables, please refer to \citet{fleming2011counting}. The AFT model has also been considered by several authors \citep{jin2003rank, ding2015estimating}. For example, \citet{jin2003rank} considered five covariates: age (\texttt{age}), edema (\texttt{edem}), bilirubin (\texttt{bili}), protime (\texttt{prot}) and albumin (\texttt{albu}). We also consider the same set of covariates for the AFT model and fit the model using the induced smoothed rank-based estimating functions with the Gehan-type weight. We then investigate several model-checking aspects using the proposed test procedures. Specifically, we conduct an omnibus test, a link function test, and a functional form test for each covariate using the standardized tests with 1,000 sample paths using both induced smoothing and non-smoothed methods.
    
    We first fit the model without transforming covariates using the standardized test procedure for \texttt{age}, \texttt{edem}, \texttt{bili}, \texttt{protime}, and \texttt{albu} (referred to as "\hypertarget{model:pbc01}{\hyperlink{model:pbc01}{Model1}}: Na"ive model"). The first row in Table 3 shows the results from \hyperlink{model:pbc01}{Model1}. \hyperlink{model:pbc01}{Model1} is not a viable model because the $p$-values for the standardized omnibus tests, standardized link function tests, and standardized functional form tests for \texttt{bilirubin} are below the pre-specified $\alpha = 0.05$. The plot for testing the functional form of \texttt{bili} in Figure \ref{fig:pbc01:mis:std} clearly shows a departure from the null.
    
    As the next model, we considered a model with log-transformed \texttt{bilirubin} (\texttt{logbili}), \texttt{prot}, \texttt{albu}, \texttt{age}, and \texttt{edem} (referred to as "\hypertarget{model:pbc02}{\hyperlink{model:pbc02}{Model 2}}: log-transformed \texttt{bilirubin} model"). The second row in Table 3 shows the test results for \hyperlink{model:pbc02}{Model 2}. The $p$-values for the omnibus, link function, and functional form tests indicate that the model is valid. In Figure \ref{fig:pbc02:mis:std}, sample paths based on the observed data in each plot (red), including that for \texttt{logbili}, are all covered by the 50 simulated sample paths generated under the null distribution. Therefore, we selected \hyperlink{model:pbc02}{Model 2} as our final model.
    
    We fitted Models \hyperlink{model:pbc01}{1} and \hyperlink{model:pbc02}{2} using the \textbf{R} package \textbf{aftgee} \citep{chiou2014fitting}. To ensure numerical stability in parameter estimation, we used the standardized covariates as considered in \citet{jin2003rank}. The parameter estimates of the five covariates, along with the corresponding 95\% point-wise confidence intervals in parentheses, based on the induced smoothing method, are as follows: -0.574 (-0.730, -0.417), -0.247 (-0.424, -0.069), 0.194 (0.046, 0.342), -0.269 (-0.407, -0.131), and -0.940 (-1.483, -0.397) for each of covariate, respectively.
    
    \begin{table}
    \caption{\label{tab:real:pbc} The results of omnibus test, link function test, and functional forms of covariates from Model 1 and Model 2. Both monotone non-induced smoothed and monotone induced smoothed methods with standardized test statistics are considered.}
    \centering
    \begin{tabular}[t]{ccccccccccccc}
        \toprule
        \multicolumn{1}{c}{ } & \multicolumn{12}{c}{P-value} \\
        \cmidrule(l{3pt}r{3pt}){2-13}
        \multicolumn{1}{c}{ } & \multicolumn{2}{c}{omni} & \multicolumn{2}{c}{link} & \multicolumn{2}{c}{\texttt{bili}} & \multicolumn{2}{c}{\texttt{prot}} & \multicolumn{2}{c}{\texttt{albu}} & \multicolumn{2}{c}{\texttt{age}} \\
        \cmidrule(l{3pt}r{3pt}){2-3} \cmidrule(l{3pt}r{3pt}){4-5} \cmidrule(l{3pt}r{3pt}){6-7} \cmidrule(l{3pt}r{3pt}){8-9} \cmidrule(l{3pt}r{3pt}){10-11} \cmidrule(l{3pt}r{3pt}){12-13}
        model & mns & mis & mns & mis & mns & mis & mns & mis & mns & mis & mns & mis\\
        \midrule
        
        \hyperlink{model:pbc01}{Model 1} & {0.020} & {0.030} & {0.000} & {0.010} & {0.000} & {0.000} & {0.725} & {0.655} & {0.790} & {0.800} & {0.580} & {0.660} \\

        \cmidrule(l{3pt}r{3pt}){1-1} \cmidrule(l{3pt}r{3pt}){2-3} \cmidrule(l{3pt}r{3pt}){4-5} \cmidrule(l{3pt}r{3pt}){6-7} \cmidrule(l{3pt}r{3pt}){8-9} \cmidrule(l{3pt}r{3pt}){10-11} \cmidrule(l{3pt}r{3pt}){12-13} 
        
        \hyperlink{model:pbc02}{Model 2} & {0.640} & {0.595} & {0.165} & {0.180} & {0.440} & {0.400} & {0.555} & {0.485} & {0.800} & {0.805} & {0.880} & {0.820} \\
        
        \bottomrule
    \end{tabular}
\end{table}
    \begin{figure}[htp]
    \centering
    \begin{subfigure}[b]{0.45\textwidth}
        \centering
        \captionsetup{justification=raggedright,singlelinecheck = false}
	    \includegraphics[scale=0.39]{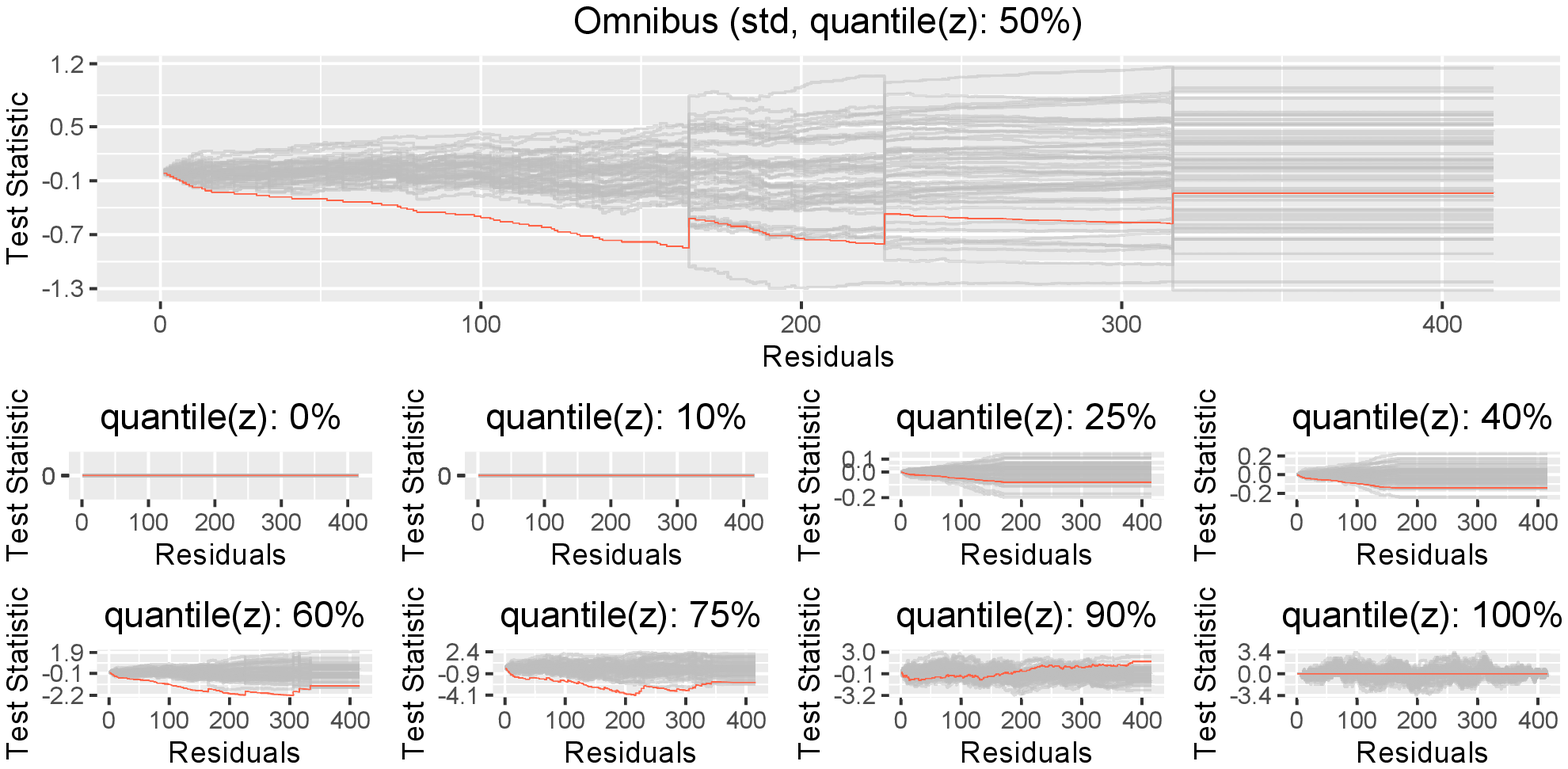} \\
	    \caption{afttestplot: pbc standardized omni function test} \label{fig:pbc01_omni_mis_std} \vspace{0em}
    \end{subfigure}
    \hfill
    \begin{subfigure}[b]{0.45\textwidth}
        \centering
        \captionsetup{justification=raggedright,singlelinecheck = false}
	    \includegraphics[scale=0.39]{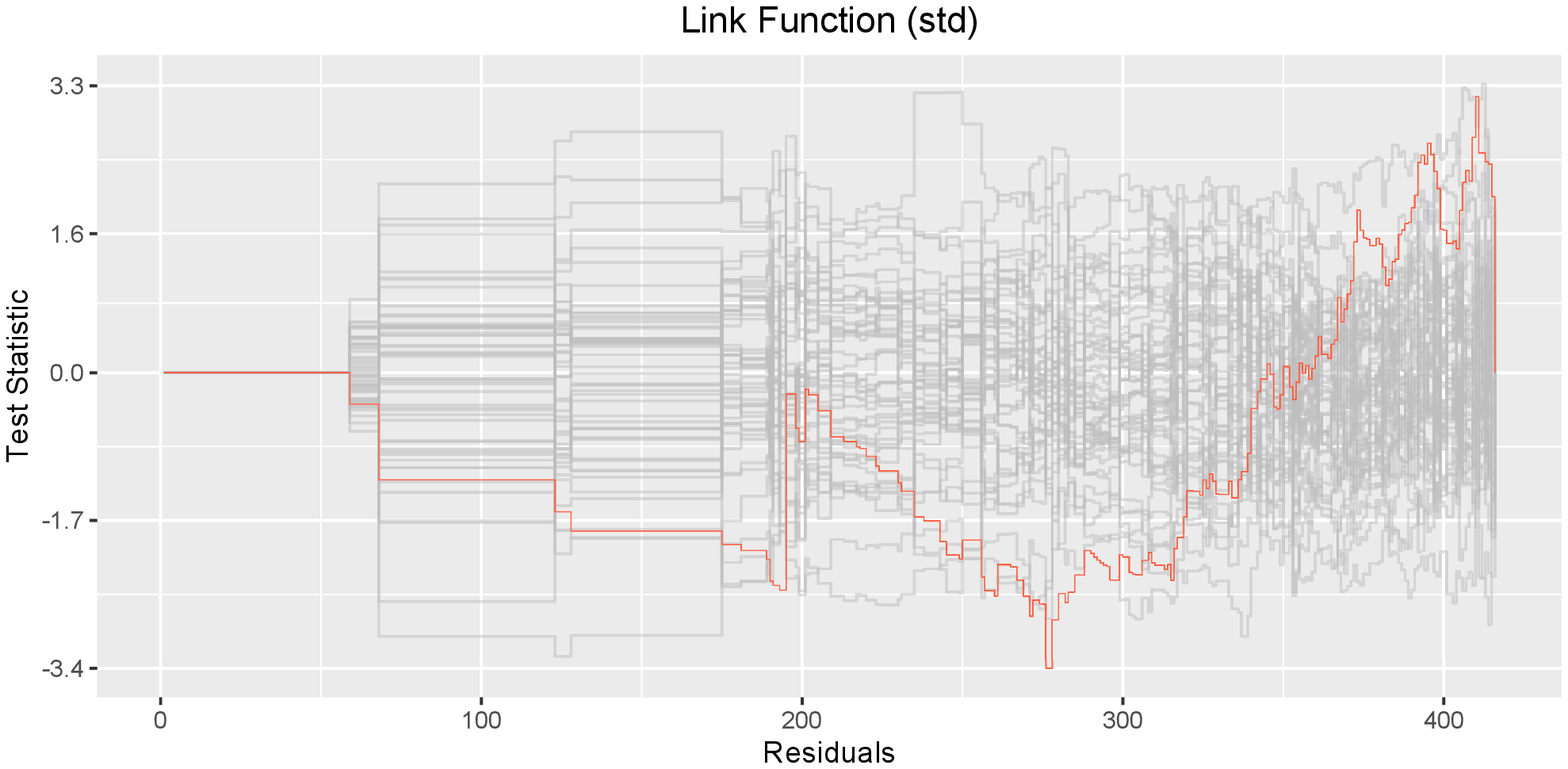} \\
	    \caption{afttestplot: standardized link function test} \label{fig:pbc01_link_mis_std} \vspace{0em}
    \end{subfigure}
    \\
    \begin{subfigure}[b]{0.45\textwidth}
        \centering
        \captionsetup{justification=raggedright,singlelinecheck = false}
	    \includegraphics[scale=0.39]{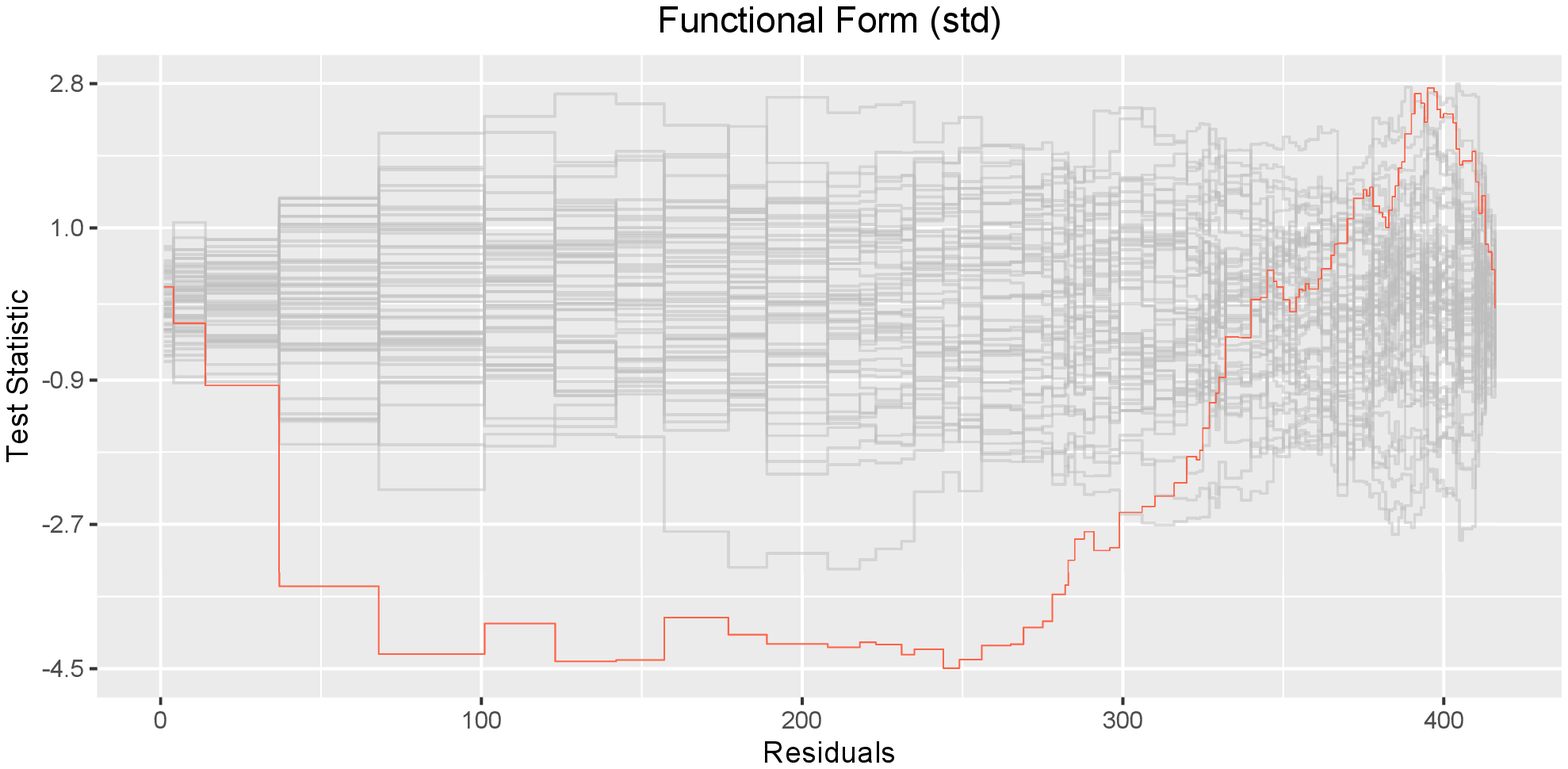} \\
	    \caption{afttestplot: standardized form function test (\texttt{bili})} \label{fig:pbc01_form1_mis_std} \vspace{0em}
    \end{subfigure}
    \hfill
    \begin{subfigure}[b]{0.45\textwidth}
        \centering
        \captionsetup{justification=raggedright,singlelinecheck = false}
	    \includegraphics[scale=0.39]{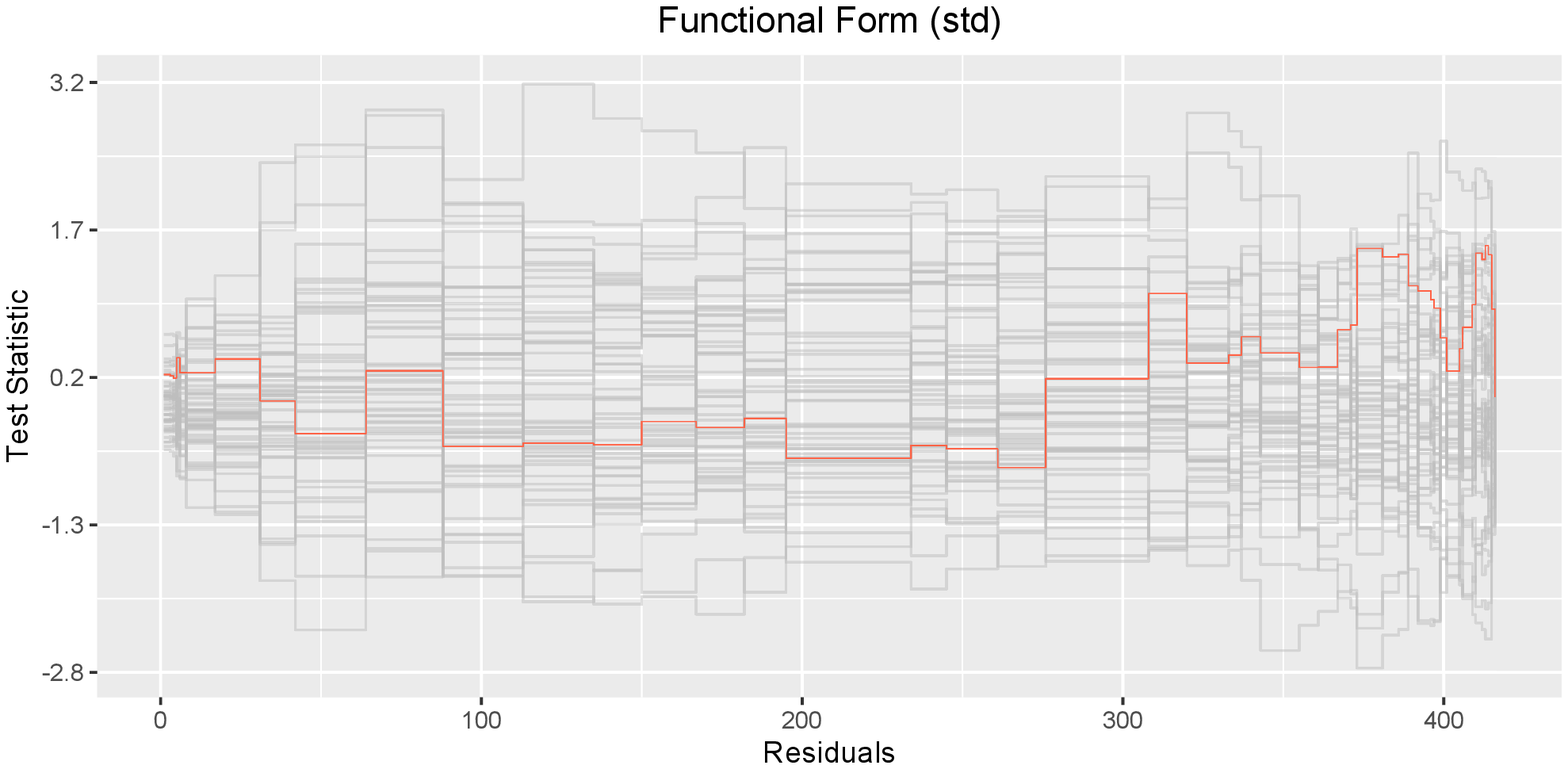} \\
	    \caption{afttestplot: standardized form function test (\texttt{prot})} \label{fig:pbc01_form2_mis_std} \vspace{0em}
    \end{subfigure}
    \\
    \begin{subfigure}[b]{0.45\textwidth}
        \centering
        \captionsetup{justification=raggedright,singlelinecheck = false}
	    \includegraphics[scale=0.39]{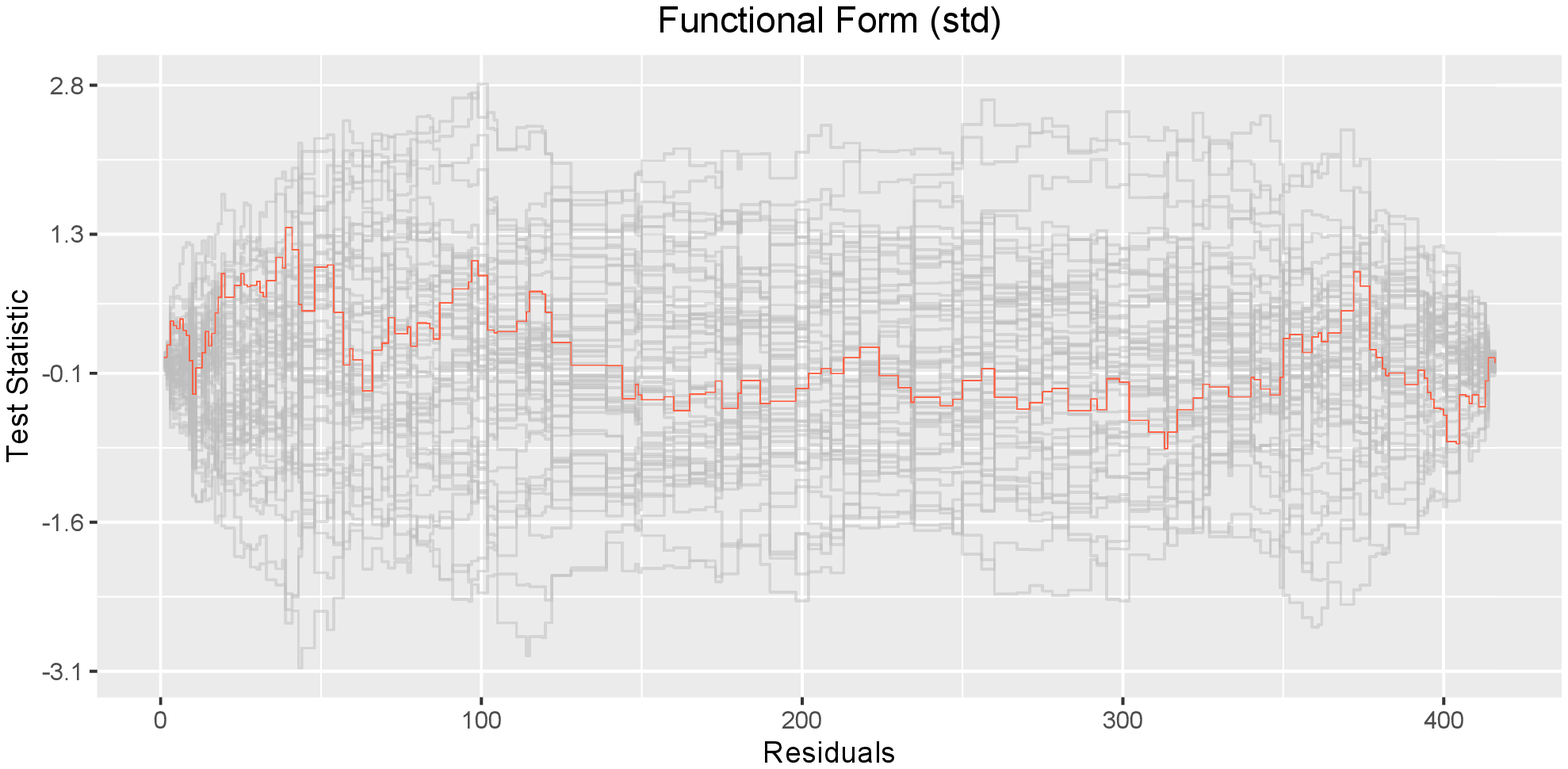} \\
	    \caption{afttestplot: standardized form function test (\texttt{albu})} \label{fig:pbc01_form3_mis_std} \vspace{0em}
    \end{subfigure}
    \hfill
    \begin{subfigure}[b]{0.45\textwidth}
        \centering
        \captionsetup{justification=raggedright,singlelinecheck = false}
	    \includegraphics[scale=0.39]{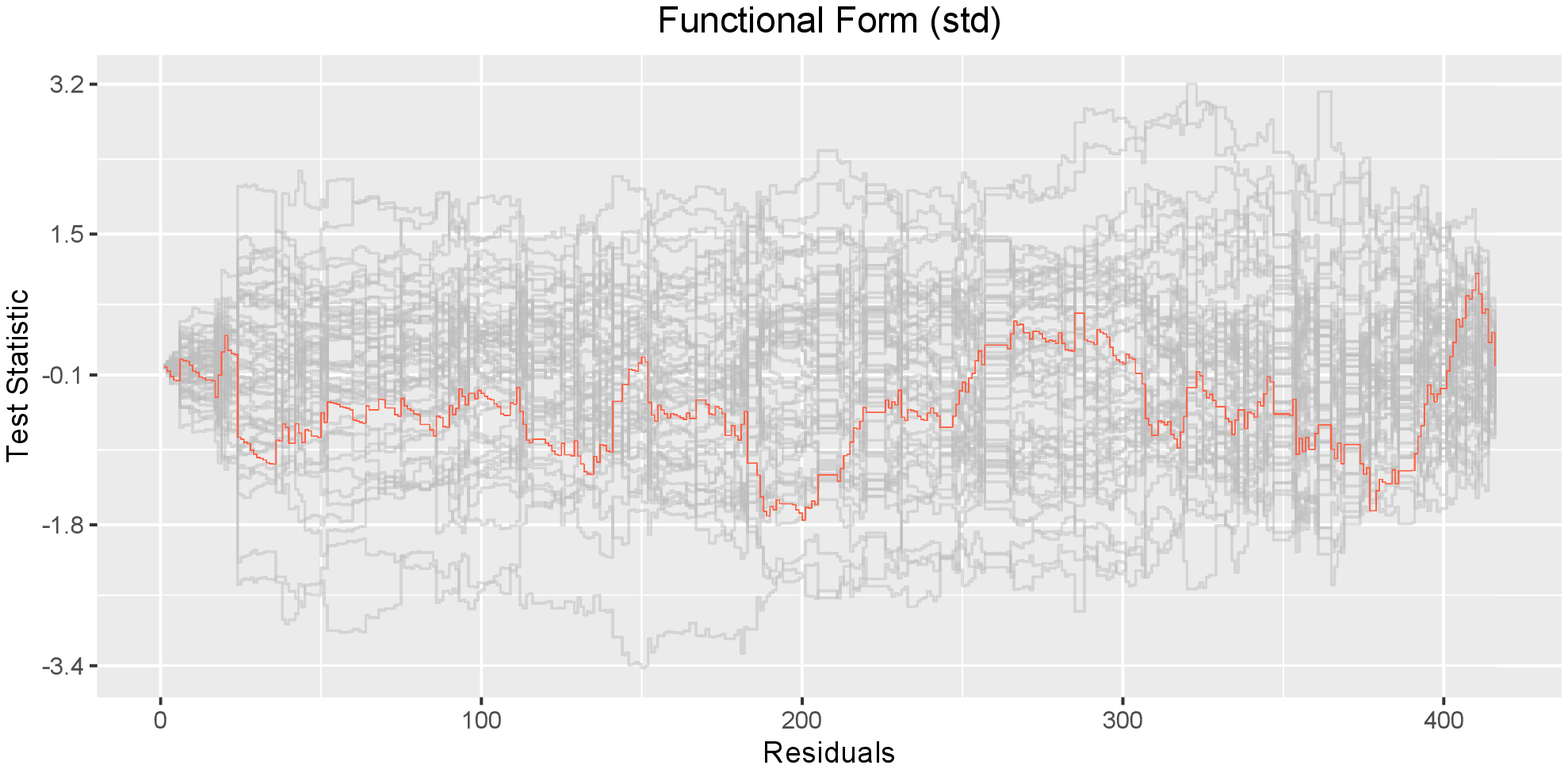} \\
	    \caption{afttestplot: standardized form function test (\texttt{age})} \label{fig:pbc01_form4_mis_std} \vspace{0em}
    \end{subfigure}
    \caption{\label{fig:pbc01:mis:std} Plots of sample paths based on the proposed standardized test procedures based on induced smoothing estimator for Model 1. $x$-axis indicates the rank of the log transformed residuals. The sample path for the test statistic (red) overlays by 50 approximated sample paths under the null (grey) are displayed. Panel (a) is the plot for the omnibus test under the different quantiles of covariates. Panel (b) is the link function test and panels (c)-(f) are functional form tests for \texttt{bili}, \texttt{prot}, \texttt{albu} and \texttt{age}, respectively.}
\end{figure}

    \begin{figure}[htp]
    \centering
    \begin{subfigure}[b]{0.45\textwidth}
        \centering
        \captionsetup{justification=raggedright,singlelinecheck = false}
	    \includegraphics[scale=0.39]{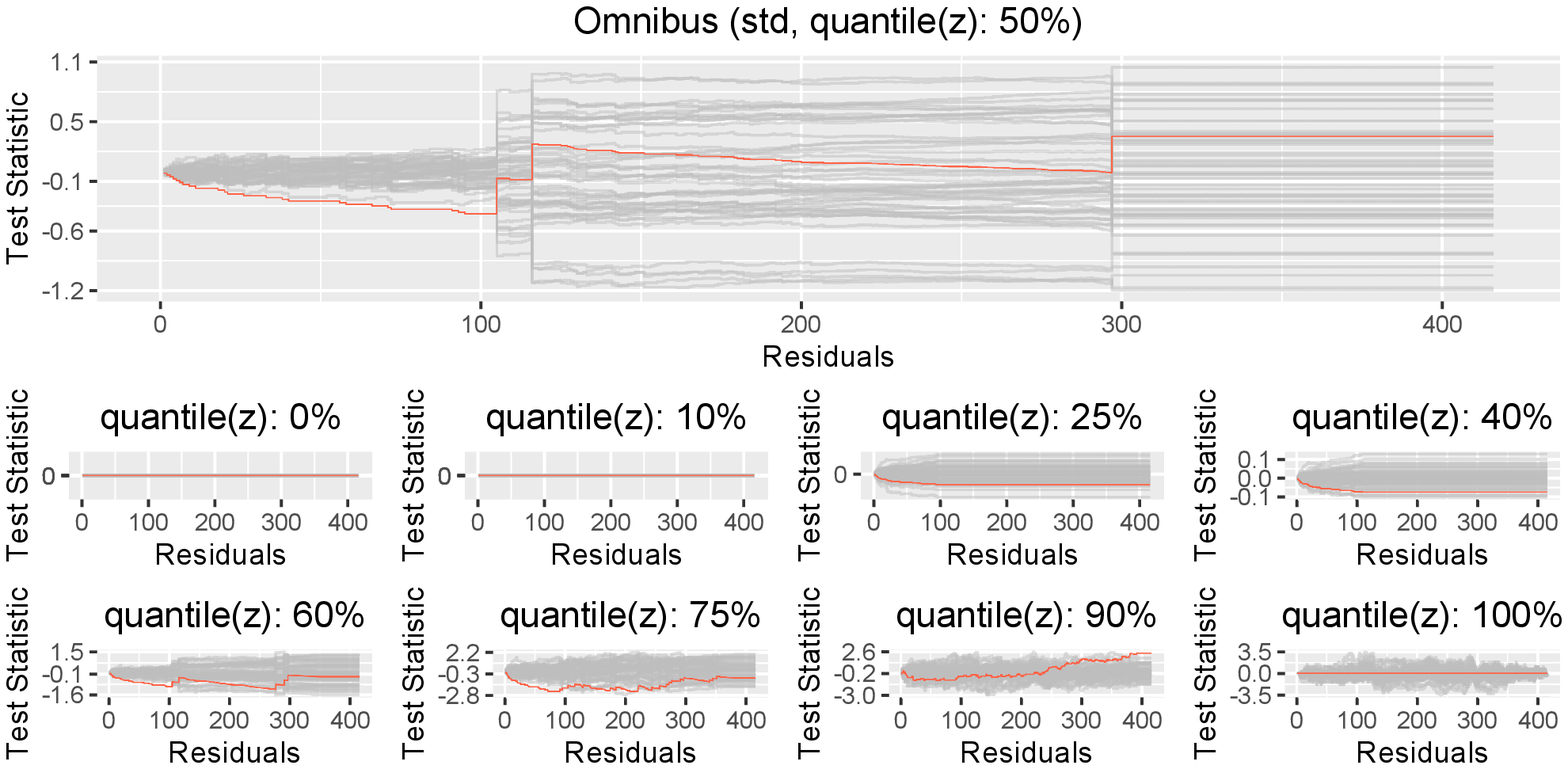} \\
	    \caption{afttestplot: pbc standardized omni function test} \label{fig:pbc02_omni_mis_std} \vspace{0em}
    \end{subfigure}
    \hfill
    \begin{subfigure}[b]{0.45\textwidth}
        \centering
        \captionsetup{justification=raggedright,singlelinecheck = false}
	    \includegraphics[scale=0.39]{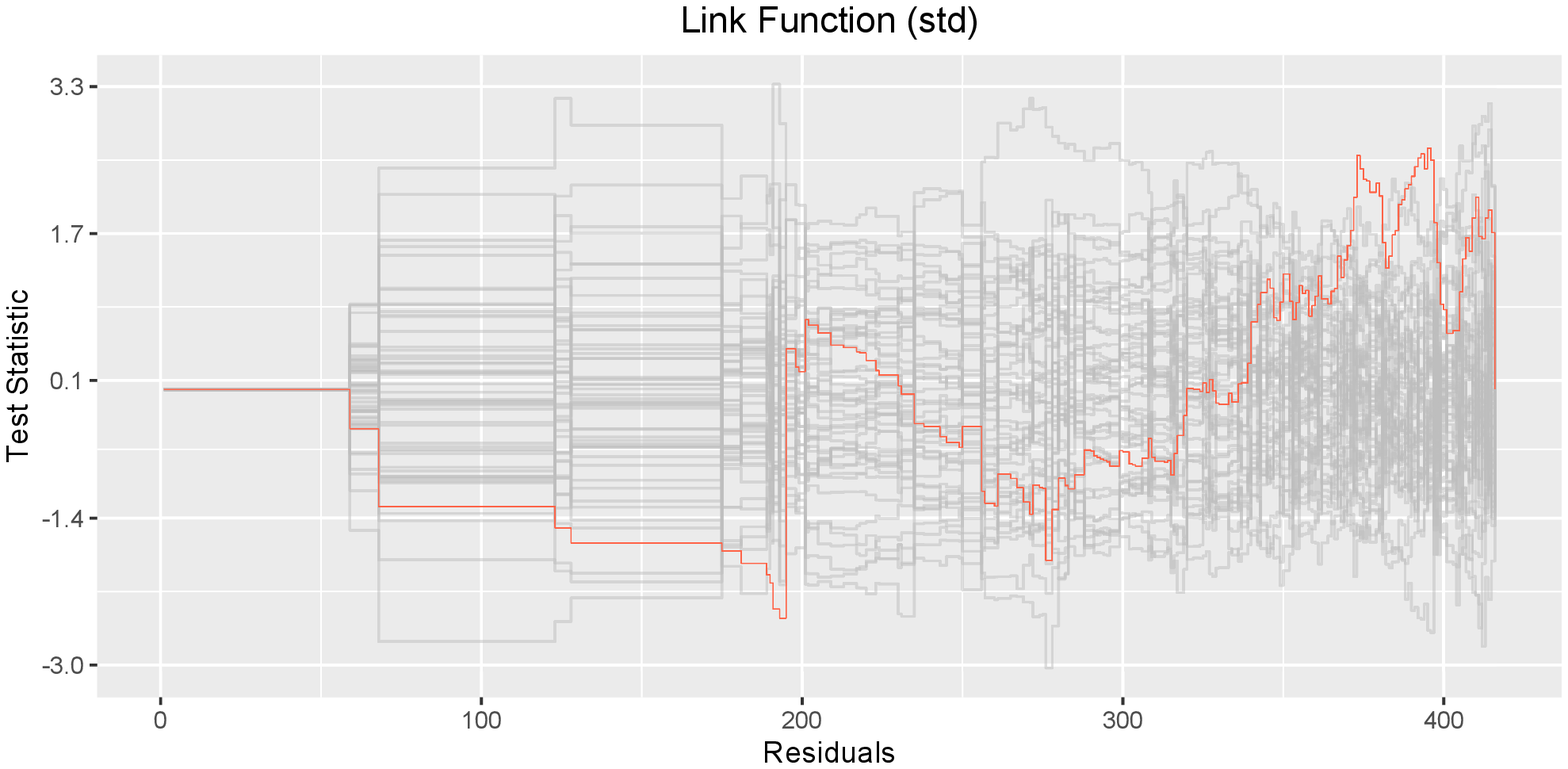} \\
	    \caption{afttestplot: standardized link function test} \label{fig:pbc02_link_mis_std} \vspace{0em}
    \end{subfigure}
    \\
    \begin{subfigure}[b]{0.45\textwidth}
        \centering
        \captionsetup{justification=raggedright,singlelinecheck = false}
	    \includegraphics[scale=0.39]{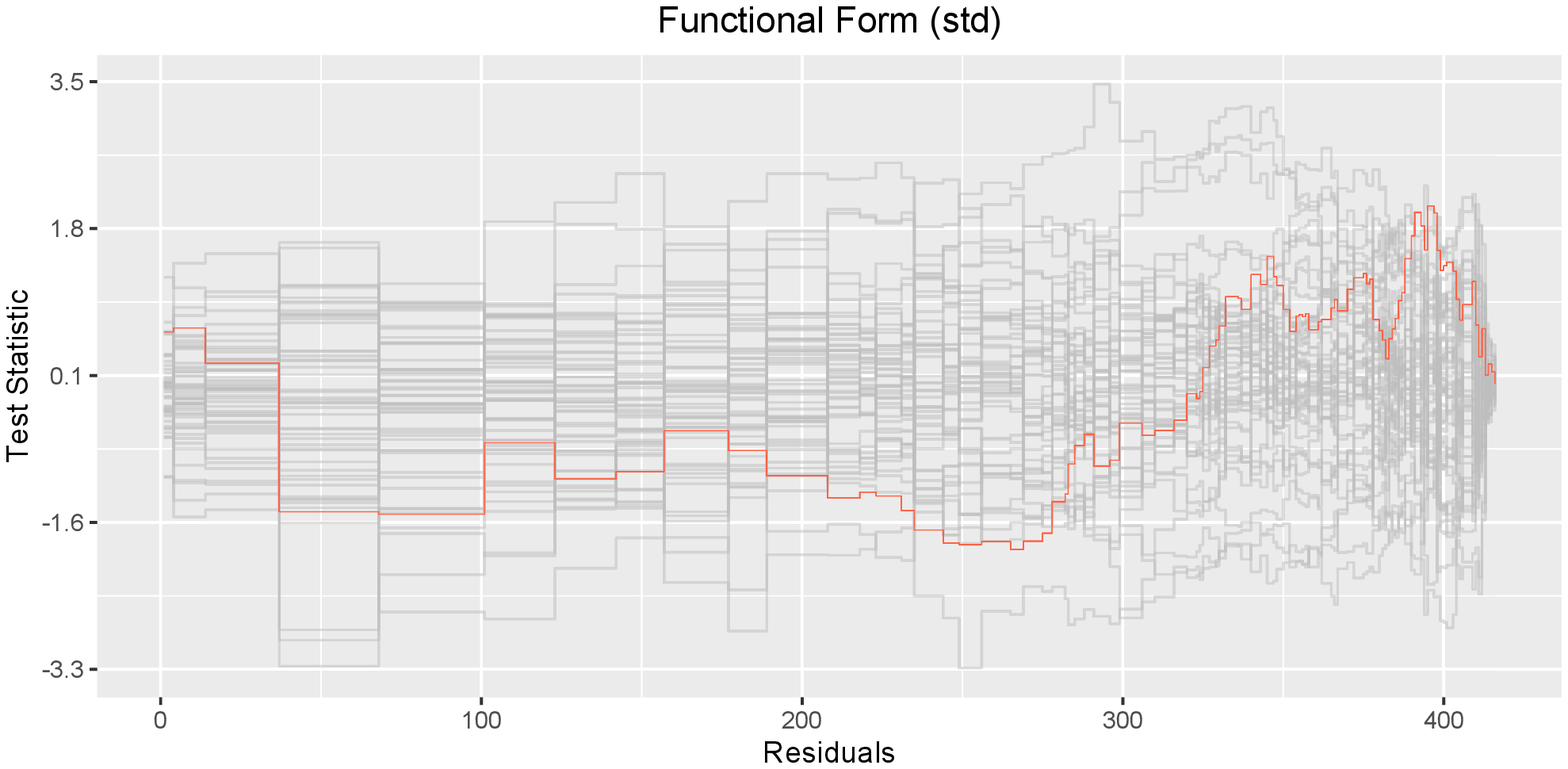} \\
	    \caption{afttestplot: standardized form function test (\texttt{logbili})} \label{fig:pbc02_form1_mis_std} \vspace{0em}
    \end{subfigure}
    \hfill
    \begin{subfigure}[b]{0.45\textwidth}
        \centering
        \captionsetup{justification=raggedright,singlelinecheck = false}
	    \includegraphics[scale=0.39]{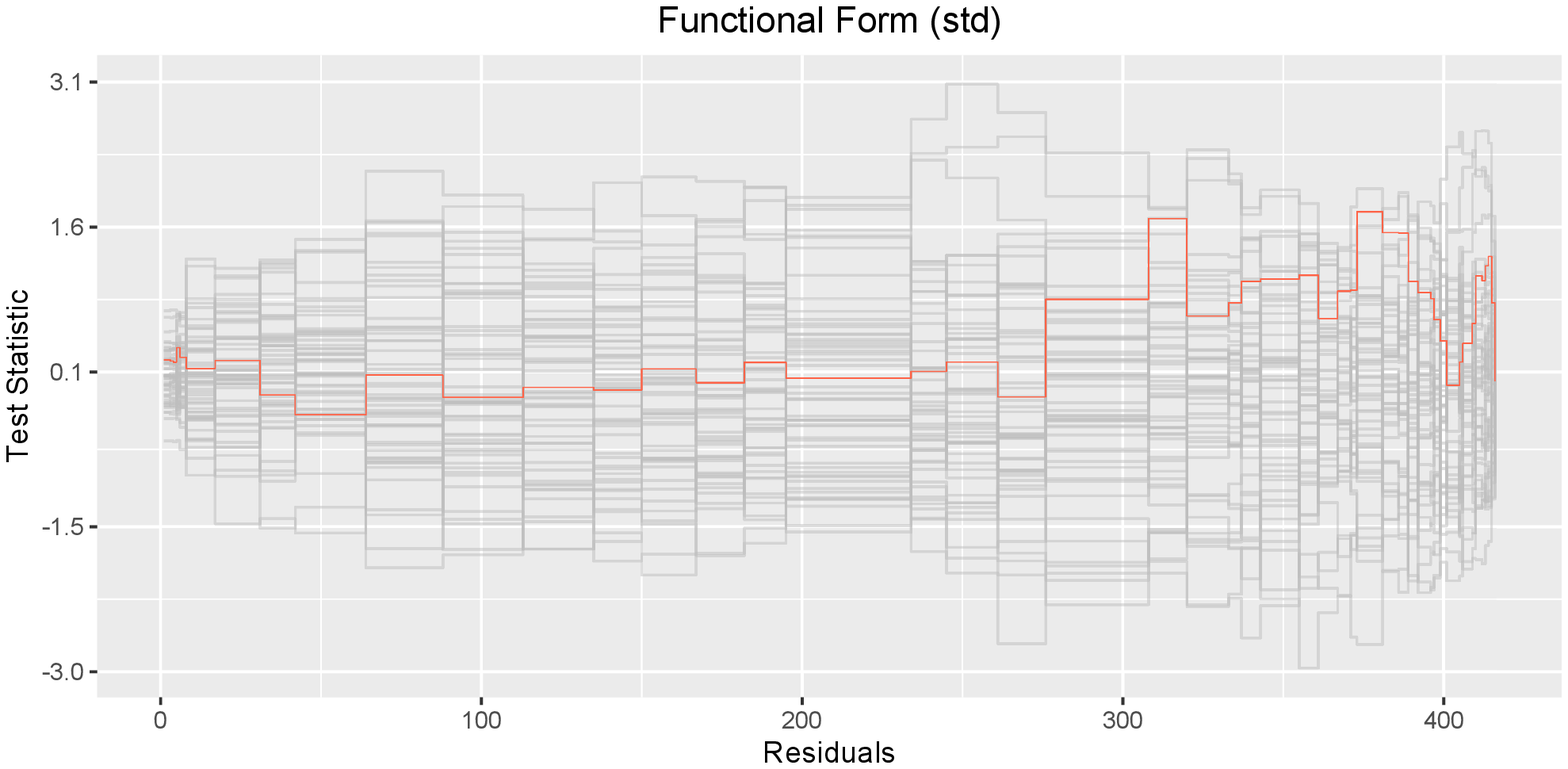} \\
	    \caption{afttestplot: standardized form function test (\texttt{prot})} \label{fig:pbc02_form2_mis_std} \vspace{0em}
    \end{subfigure}
    \\
    \begin{subfigure}[b]{0.45\textwidth}
        \centering
        \captionsetup{justification=raggedright,singlelinecheck = false}
	    \includegraphics[scale=0.39]{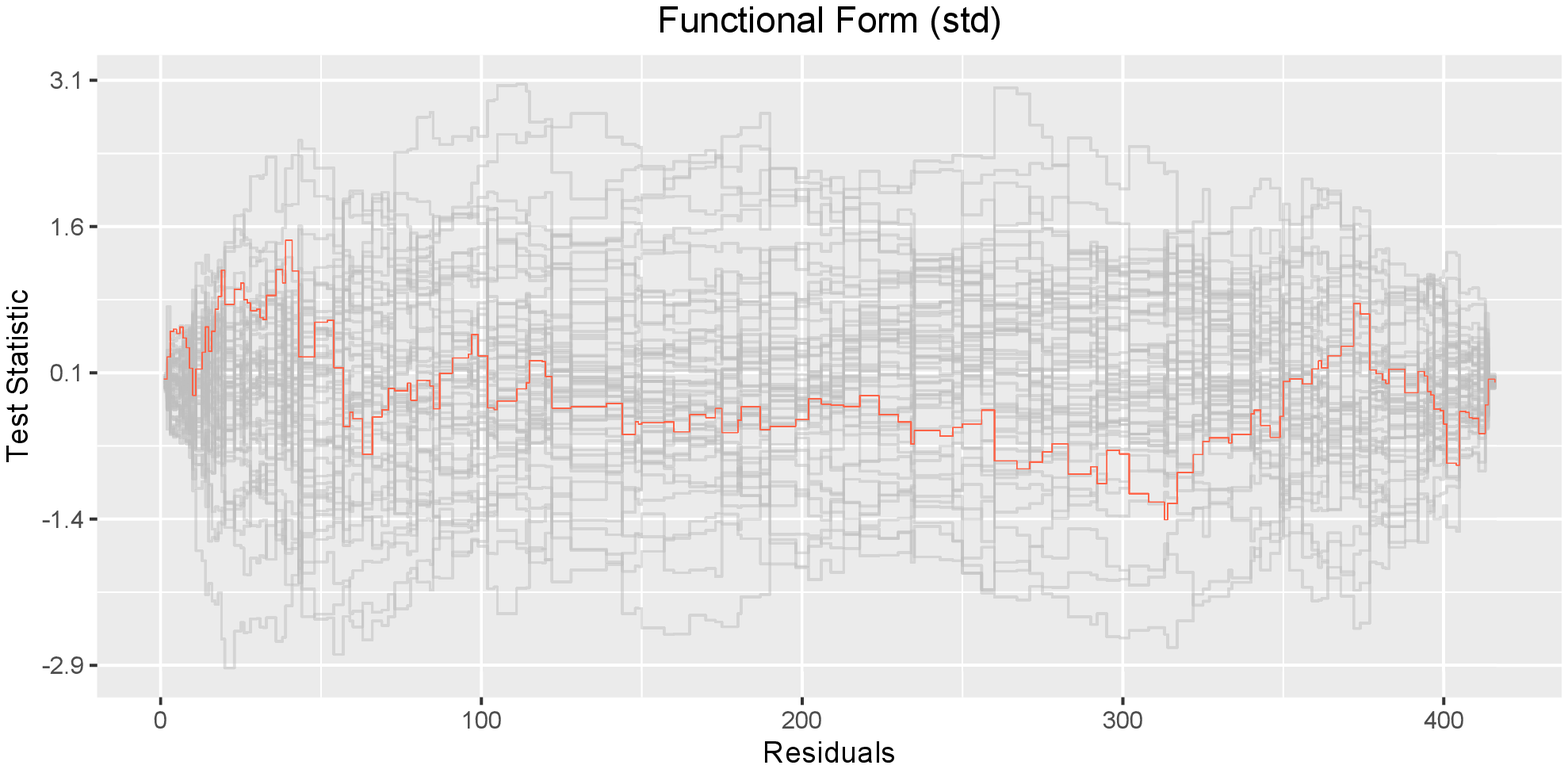} \\
	    \caption{afttestplot: standardized form function test (\texttt{albu})} \label{fig:pbc02_form3_mis_std} \vspace{0em}
    \end{subfigure}
    \hfill
    \begin{subfigure}[b]{0.45\textwidth}
        \centering
        \captionsetup{justification=raggedright,singlelinecheck = false}
	    \includegraphics[scale=0.39]{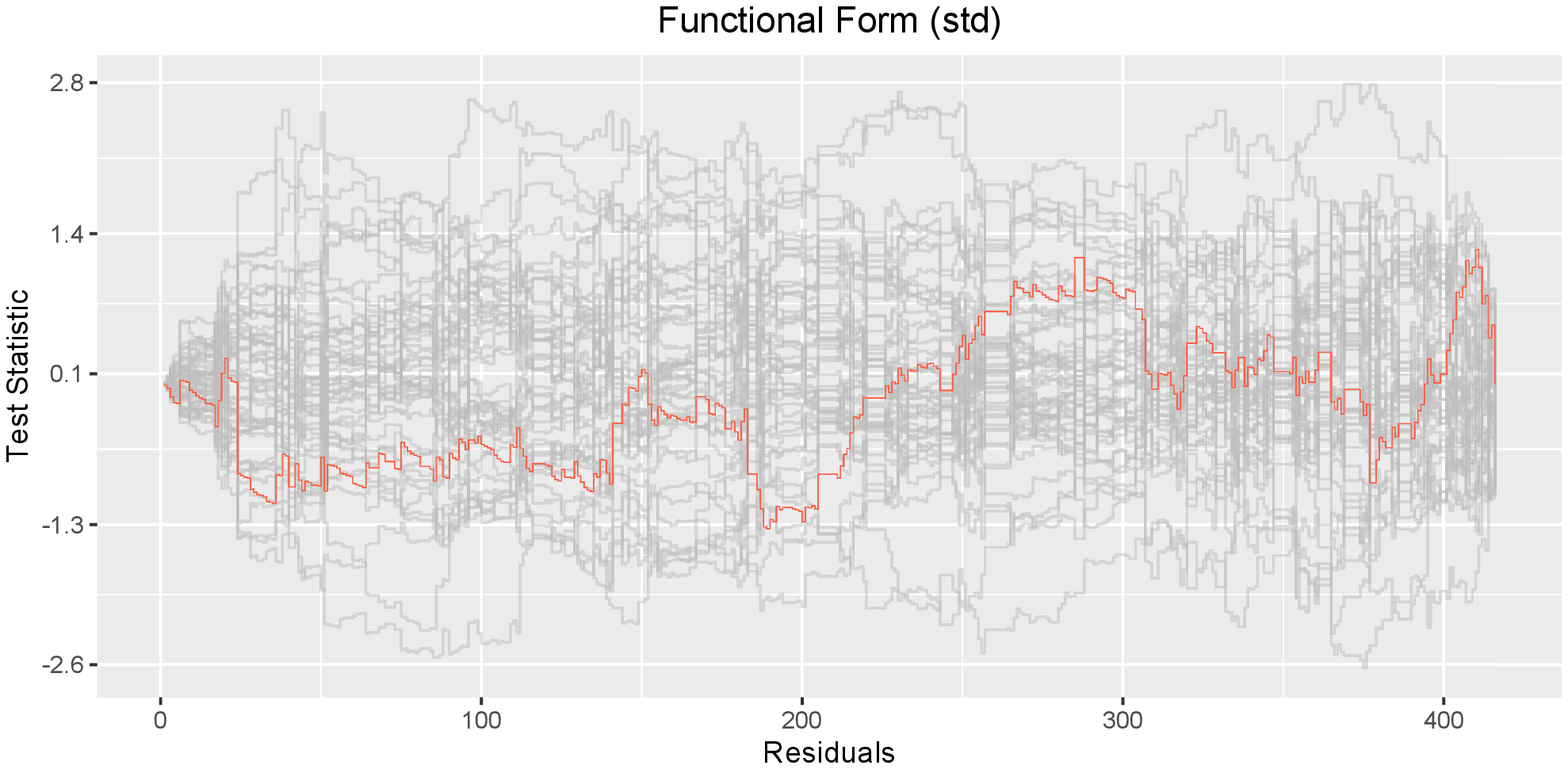} \\
	    \caption{afttestplot: standardized form function test (\texttt{age})} \label{fig:pbc02_form4_mis_std} \vspace{0em}
    \end{subfigure}
    \caption{\label{fig:pbc02:mis:std} Plots of sample paths based on the proposed standardized test procedures based on induced smoothing estimator for Model 2. $x$-axis indicates the rank of the log transformed residuals. The sample path for the test statistic (red) overlays by 50 approximated sample paths under the null (grey) are displayed. Panel (a) is the plot for the omnibus test under the different quantiles of covariates. Panel (b) is the link function test and panels (c)-(f) are functional form tests for \texttt{logbili}, \texttt{prot}, \texttt{albu} and \texttt{age}, respectively.}
\end{figure}
    
\section{Conclusion} \label{sec:conclusion}
    In this article, we propose AFT model-checking methods, including an omnibus test, a link function test, and a functional form test. These test statistics are specific forms of a weighted summation of martingale residuals. Although \citet{novak2013goodness, novak2015regression} proposed an omnibus test procedure based on a non-smooth rank-based estimator with Gehan-type weight, little research has been conducted on the omnibus test and other model-checking aspects, such as checking the link function and functional form of each covariate. This article elaborates on detailed model-checking techniques that encompass several important aspects of model-checking based on the induced smoothed estimator for model parameters, which is computationally more efficient than its non-smooth counterpart while maintaining the same asymptotic properties. The asymptotic properties of the proposed test statistics are rigorously established. Results from extensive simulation experiments show that the proposed test procedures maintain a Type \Romannum{1} error rate and are reasonably powerful in detecting departures from the assumed AFT model with practical sample sizes and settings considered. We also illustrate our proposed methods by analyzing the PBC data. Our concluding model backs up the final AFT model considered by several researchers with the same set of covariates. To make the proposed model-checking procedures more convenient to use, we have implemented all the aforementioned procedures in the \textbf{R} package \textbf{afttest} \citet{bae2022afttest} 
    
    The proposed procedures in this article could be extended in several directions. Currently, the proposed methods only consider time-invariant covariates, but time-varying covariates are frequently encountered in many biomedical studies. Semiparametric AFT models that accommodate time-varying covariates have been investigated \citep{lin1995semiparametric}. \citet{novak2013goodness,novak2015regression} also considered the incorporation of time-varying covariates to some extent under the non-smoothed estimation framework. However, the induced smoothing method has not yet been developed to accommodate time-varying covariates for a semiparametric AFT model. Therefore, extending the induced smoothing method to fit semiparametric AFT models with time-varying covariates and developing model-checking procedures for the induced smoothed estimator would be a natural next research direction.
    
    The proposed method in this article considers the setting with univariate failure time data assuming independence among subjects. However, extending our proposed methods to the setting with clustered failure time data is straightforward with the marginal model approach. Rank-based estimation methods and their induced smoothed versions are currently available \citep{johnson2009induced, chiou2015semiparametric}, but the corresponding model-checking procedures are still underdeveloped. Therefore, developing model-checking procedures for the induced smoothed estimator in the clustered failure time data setting is another possible future research direction.

\section*{Code Availability} \label{sec:computation}
    This paper presents the results obtained using version 4.3.0 of the statistical computing environment \textbf{R} and version 4.3.0 of the \textbf{afttest} package. The \textbf{afttest} package comprises two primary functions, namely \texttt{afttest} and \texttt{afttestplot}, which provide both nonsmooth (\texttt{mns}) and induced-smoothed (\texttt{mis}) based outcomes. All the packages used in this study are available from the Comprehensive \textbf{R} Archive Network (CRAN). The most recent source codes for the package and its analysis can be accessed via the following links: \url{https://github.com/WoojungBae/afttest} and \url{https://github.com/WoojungBae/afttest\_analysis}.

\section*{Acknowledgements}
    The first two authors have made equal contributions to this paper. 

\bibliographystyle{agsm}
\bibliography{ref}

\clearpage
\setcounter{section}{1}
\renewcommand*{\thesection}{\Alph{section}}
\section{Appendix} \label{sec:appendix}
    We prove Theorems 1 and 2 for the induced smoothed rank-based estimator, $\bbhi$. As mentioned in Section 3, these results are also applicable to the other estimators such as the non-smooth rank-based estimator, $\bbhn$, and least-squares estimator, $\bbhl$. For $\bbhn$, the results remain identical mainly due to the asymptotic equivalence of the limiting distributions of $\bbhi$ and $\bbhn$ \citep{johnson2009induced}. For $\bbhl$, the asymptotic covariance function for $\w\left( t, \bz \right)$ will have a different form but the main results in Theorems 1 and 2 remain the same. 
    
    We assume the following regularity conditions. Note that these conditions are also specified in \citet{lin1998accelerated} and \citet{novak2013goodness, novak2015regression}.
    \begin{itemize}
        \item[C1] $\left( N_{i}, Z_{i} \right)$ are bounded.
        
        \item[C2] $\left( N_{i}, C_{i}, Z_{i} \right)$ are i.i.d.
        
        \item[C3] $\psi\left( t, \bb_{0}\right), E \left( t, \bb_{0}\right), E_{\pi}\left( t, \bb_{0}\right)$, and $n^{-1} S_{\pi}\left( t, \bb_{0}\right) $ converge almost surely to $\xi, \mathscr{E}, \mathscr{E}_{\pi}$, and $\mathscr{S}_{\pi}$, respectively.
        
        \item[C4] $C_{i}\mbox{exp}(\boldsymbol{Z}^{\top}_{i} \beta_{0})$ have a uniformly bounded density and $\Lambda_{0}(\cdot)$ has a bounded second derivative.
        
    \end{itemize}
    
\subsection{Proof of Theorem~\ref{theorem1}} \label{subsec:prooftheorem1}
    We prove Theorem~\ref{theorem1} by first showing that $\w \left( t, \bz ; \bbhi \right)$ follows asymptotically a zero-mean Gaussian process. By the asymptotic expansion of $\bbhi$ \citep{johnson2009induced}, it can be shown that 
    \begin{align*}
        \w \left( t, \bz ; \bbhi \right)
        & = n^{-\frac{1}{2}} \sum_{i=1}^{n} \pi_{i} \left( \bz \right) \widehat{M}_{i} \left( t, \bbhi \right) \\
        & = n^{-\frac{1}{2}} \sum_{i=1}^{n} \pi_{i} \left( \bz \right) M_{i} \left( t, \bb_{0} \right) + n^{-\frac{1}{2}} \sum_{i=1}^{n} \pi_{i} \left( \bz \right) \left\{ \widehat{M}_{i} \left( t, \bbhi \right) - M_{i} \left( t, \bb_{0} \right) \right\} \\
        & = n^{-\frac{1}{2}} \sum_{i=1}^{n} \pi_{i} \left( \bz \right) M_{i} \left( t, \bb_{0} \right) + n^{-\frac{1}{2}} \sum_{i=1}^{n} \pi_{i} \left( \bz \right) \left\{ N_{i} \left( t, \bbhi \right) - N_{i} \left( t, \bb_{0} \right) \right\} \\
        & \hspace{2em} - n^{-\frac{1}{2}} \sum_{i} \pi_{i} \left( \bz \right) \left( \int_{0}^{t} Y_{i} \left( s , \bbhi \right) d \widehat{\Lambda}_{0} \left( s , \bbhi \right) - \int_{0}^{t} Y_{i} \left( s , \bb_{0} \right) d \Lambda_{0} \left( s \right) \right) \\
        & = n^{-\frac{1}{2}} \sum_{i=1}^{n} \pi_{i} \left( \bz \right) M_{i} \left( t, \bb_{0} \right) + n^{-\frac{1}{2}} \sum_{i=1}^{n} \pi_{i} \left( \bz \right) \left\{ N_{i} \left( t, \bbhi \right) - N_{i} \left( t, \bb_{0} \right) \right\} \\
        & \hspace{2em} 
        - n^{-\frac{1}{2}} \sum_{i} \pi_{i} \left( \bz \right) \int_{0}^{t} Y_{i} \left( s , \bb_{0} \right) d \left( \widehat{\Lambda}_{0} \left( s , \bbhi \right) - d \Lambda_{0} \left( s \right) \right) \\
        & \hspace{2em} 
        - n^{-\frac{1}{2}} \sum_{i} \pi_{i} \left( \bz \right) \int_{0}^{t} \left( Y_{i} \left( s , \bbhi \right) - Y_{i} \left( s , \bb_{0} \right) d \Lambda_{0} \left( s \right) \right) 
    \end{align*}

    Applying the results in \citet[][Lemma 6.1; p.53]{novak2013goodness}, and \citet[][Lemma 10; p.38]{novak2015regression}, we obtain
    \begin{align*}
        \w \left( t, \bz ; \bbhi \right)
        & = n^{-\frac{1}{2}} \sum_{i=1}^{n} \pi_{i} \left( \bz \right) M_{i} \left( t, \bb_{0} \right) + n^{-\frac{1}{2}} \sum_{i=1}^{n} \pi_{i} \left( \bz \right) \left\{ N_{i} \left( t, \bbhi \right) - N_{i} \left( t, \bb_{0} \right) \right\} \\
        & \hspace{2em} 
        - n^{-\frac{1}{2}} \sum_{i} \pi_{i} \left( \bz \right) \int_{0}^{t} Y_{i} \left( s , \bb_{0} \right) d \left( \widehat{\Lambda}_{0} \left( s , \bbhi \right) - d \Lambda_{0} \left( s \right) \right) \\
        & \hspace{2em} 
        - n^{-\frac{1}{2}} \sum_{i} \pi_{i} \left( \bz \right) \int_{0}^{t} \left( Y_{i} \left( s , \bbhi \right) - Y_{i} \left( s , \bb_{0} \right) d \Lambda_{0} \left( s \right) \right) \\
        & = n^{-\frac{1}{2}} \sum_{i=1}^{n} \pi_{i} \left( \bz \right) M_{i} \left( t, \bb_{0} \right) - n^{-\frac{1}{2}} f_{\pi}^{\top} \left( t,\bz \right) \left( \bbhi - \bb_{0} \right) \\
        & \hspace{2em} 
        - n^{-\frac{1}{2}} \sum_{i} \pi_{i} \left( \bz \right) \int_{0}^{t} Y_{i} \left( s , \bb_{0} \right) d \left( \widehat{\Lambda}_{0} \left( s , \bbhi \right) - d \Lambda_{0} \left( s \right) \right) \\
        & \hspace{2em} 
        - n^{-\frac{1}{2}} \sum_{i} \pi_{i} \left( \bz \right) \int_{0}^{t} \left( Y_{i} \left( s , \bbhi \right) - Y_{i} \left( s , \bb_{0} \right) d \Lambda_{0} \left( s \right) \right) 
    \end{align*}
    Using the arguments in \citet[][Theorems 2 and 4; p.616-617]{lin1998accelerated}, \citet[][Lemma 6.1; p.53]{novak2013goodness}, and \citet[][Lemma 10; p.38]{novak2015regression}, we have
    \begin{align*}
        n^{\frac{1}{2}} \left( \widehat{\Lambda}_{0} \left( s , \bbhi \right) - \Lambda_{0} \left( s \right) \right) 
        & = n^{\frac{1}{2}} \left( \widehat{\Lambda}_{0} \left( s , \bb_{0} \right) - \Lambda_{0} \left( s \right) \right) + \boldsymbol{\kappa}^{\top} \left( t \right) n^{\frac{1}{2}} \left( \bbhi - \bb_{0} \right) + o_{p} \left( 1 \right) \\
        & = n^{\frac{1}{2}} \left( \sum_{i=1}^{n} \int_{0}^{t} \frac{J \left( s \right)}{ S_{0} \left( s, \bb_{0} \right) } d M_{i} \left( s,\bb_{0} \right) \right) + \boldsymbol{\kappa}^{\top} \left( t \right) n^{-\frac{1}{2}} \boldsymbol{\Omega}^{-1} \scoreI \left( \bb_{0} \right) 
        + o_{p} \left( 1 \right)
    \end{align*}
    where $\boldsymbol{\kappa} \left( t \right) = - \int_{0}^{t} \mathscr{E} \left( s \right) d \left( \lambda_{0} \left( s \right) s \right)$ and $\boldsymbol{\Omega} = \int_{0}^{\infty} \xi \left( t \right) \mE{ Y_{1} \left( t, \bb_{0} \right) \left( Z_{1} - \mathscr{E} \left( t \right) \right)^{\otimes 2} } d \left( \lambda_{0} \left( t \right) t \right)$. Then,
    \begin{align*}
        \w \left( t, \bz ; \bbhi \right)
        & = n^{-\frac{1}{2}} \sum_{i=1}^{n} \pi_{i} \left( \bz \right) M_{i} \left( t, \bb_{0} \right) 
        - n^{-\frac{1}{2}} \left( f_{\pi} \left( t,\bz \right) + \int_{0}^{t} g_{\pi} \left( s, \bz \right) d \Lambda_{0} \left( s \right) \right)^{\top} \left( \bbhi - \bb_{0} \right) \\
        & \hspace{2em} 
        - n^{-\frac{1}{2}} \sum_{i} \int_{0}^{t} \frac{S_{\pi} \left( s, z, \bb_{0} \right)}{S_{0} \left( s,\bb_{0} \right)} d M_{i} \left( t, \bb_{0} \right)
        - n^{-\frac{1}{2}} \sum_{i} \int_{0}^{t} S_{\pi} \left( s, z, \bb_{0} \right) d \boldsymbol{\kappa}^{\top} \left( t \right) n^{-\frac{1}{2}} \boldsymbol{\Omega}^{-1} \scoreI \left( \bb_{0} \right) 
        + o_{p} \left( 1 \right) \\
        & = n^{-\frac{1}{2}} \sum_{i} \int_{0}^{t} \left( \pi_{i} \left( \bz \right) - E_{\pi} \left( s, \bz, \bb \right) \right) d M_{i} \left( t, \bb_{0} \right) \\
        & \hspace{2em} 
        - n^{-\frac{1}{2}} \left( f_{\pi} \left( t,\bz \right) + \int_{0}^{t} g_{\pi} \left( s, \bz \right) d \Lambda_{0} \left( s \right) + \sum_{i} \int_{0}^{t} n^{-1} S_{\pi} \left( s, z, \bb_{0} \right) d \boldsymbol{\kappa}^{\top} \left( t \right) \right)^{\top} \boldsymbol{\Omega}^{-1} \scoreI \left( \bb_{0} \right) 
        + o_{p} \left( 1 \right)  
    \end{align*}

    Denote
    \begin{equation*}
        V_{t} = n^{-\frac{1}{2}} \sum_{i} M_{i} \left( t,\bb_{0} \right), 
        V_{Z} = n^{-\frac{1}{2}} \sum_{i} \bZ_{i} M_{i} \left( t, \bb_{0} \right), \mbox{ and } 
        V_{\pi} = n^{-\frac{1}{2}} \sum_{i} \pi_{i} \left( \bz \right) M_{i} \left( t,\bb_{0} \right).
    \end{equation*}
    For fixed $t$ and $z$, each of the processes is a sum of i.i.d. quantities having a zero mean. The multivariate central limit theorem establishes the finite-dimensional convergence of $\left( V_{t}, V_{Z}, V_{\pi} \right)$. $M_{i} \left( t, \bb_{0} \right)$, $\bZ_{i} M_{i} \left(t,\bb_{0}\right)$, and $\pi_{i}\left(\bz\right) M_{i} \left( t, \bb_{0} \right)$, $i =1, \ldots, n$ are manageable \cite[p38]{pollard1990empirical} since they can be expressed as sums and products of monotone functions. Then, it follows from the functional central limit theorem \citep{pollard1990empirical}) that $\left( V_{t}, V_{Z}, V_{\pi} \right)$ is tight and converges weakly to a zero-mean Gaussian process, denoted as $\left( W_{t}, W_{Z}, W_{\pi} \right)$. By the Skorohod-Dudley-Wichura theorem \cite[p47]{shorack2009empirical}, there exists an equivalent process $\left( V_{t}, V_{Z}, V_{\pi} \right)$ in an alternative probability space that the weak convergence is strengthened to almost sure convergence. 
    Combining these results with the almost sure convergence results of $\psi\left( t, \bb_{0}\right)$, $E \left( t, \bb_{0}\right)$, $E_{\pi}\left( t, \bb_{0}\right)$, and $n^{-1} S_{\pi}\left( t, \bb_{0}\right)$ to continuous functions $\xi$, $\mathscr{E}$, $\mathscr{E}_{\pi}$, and $\mathscr{S}_{\pi}$, respectively, we have the following almost sure convergence results:
    \begin{align*}
        & (1) \quad n^{-\frac{1}{2}} \int_{0}^{t} \sum_{i} \pi_{i} \left( \bz \right) dM_{i} \left( s,\bb_{0} \right) \quad \text{converges to} \quad \int_{0}^{t} d W_{\pi}; \\
        & (2) \quad n^{-\frac{1}{2}} \int_{0}^{t} \sum_{i} E_{\pi} \left( t, \bz, \bb_{0} \right) dM_{i} \left( s,\bb_{0} \right) \quad \text{converges to} \quad \int_{0}^{t} \mathscr{E}_{\pi} \left( s,\bb_{0} \right) d W_{Z} \left( s \right); \\
        & (3) \quad n^{-\frac{1}{2}} \left( f_{\pi} \left(t,\bZ \right) + \int_{0}^{t} g_{\pi} \left(s, \bZ \right) d \Lambda_{0} \left( s \right) + \int_{0}^{t} n^{-1} S_{\pi} \left( s, z, \bb_{0} \right) d \boldsymbol{\kappa}^{\top} \left( s \right) \right)^{\top} \boldsymbol{\Omega}^{-1} \scoreI \left( \bb_{0} \right) \quad \text{converges to} \quad \\
        & \hspace{6em} \left( f_{\pi 0} \left( t \right) + \int_{0}^{t} g_{\pi 0} \left( s \right) d \Lambda_{0} \left( s \right) + \int_{0}^{t} \mathscr{S}_{\pi} \left( s, z, \bb_{0} \right) d \boldsymbol{\kappa}^{\top} \left( s \right) \right)^{\top} \left( \int_{0}^{\infty} \xi \left( s \right) d W_{Z} - \int_{0}^{\infty} \xi \left( s \right) \mathscr{E} \left( s,\bb_{0} \right) d W_{t} \right).
    \end{align*}
    In summary, as $n \to \infty$, $\w \left( t, z; \bbhi \right)$ converges weakly to 
    \begin{align*}
        \int_{0}^{t} dW_{Z} 
        - \int_{0}^{t} \mathscr{E} \left( s, \bb_{0} \right) dW_{t} - \boldsymbol{\nu}^{\top} \left( t \right) \left(\int_{0}^{\infty} \xi \left( s \right) dW_{Z} 
        - \int_{0}^{\infty} \xi \left( s \right) \mathscr{E} \left( s, \bb_{0} \right) dW_{t}  \right)
    \end{align*}
    where $\boldsymbol{\nu} \left( t \right) = f_{\pi 0} \left( t,\bz \right) + \int_{0}^{t} g_{\pi 0} \left( s, \bz \right) d \Lambda_{0} \left( s \right) + \int_{0}^{t}  \mathscr{S}_{\pi} \left( s, z, \bb_{0} \right) d \boldsymbol{\kappa}^{\top} \left( t \right)$ is a Gaussian process with the zero-mean and the covariance function being
    \begin{align*} 
        \sigma \left( t_{1}, t_{2} \right) 
        = & E \Bigg[ \left\{ \int_{0}^{t_{1}} \left( \pi_{1}\left( \bz \right) - \mathscr{E}_{\pi} \left( s, \bb_{0} \right) \right) d M_{1} \left( s, \bb_{0} \right) - \boldsymbol{\nu}^{\top} \left( t_{1} \right) \boldsymbol{\Omega}^{-1} \int_{0}^{\infty}\xi \left( s \right) \left( Z_1 - \mathscr{E} \left( s, \bb_{0} \right) \right) d M_{1}\left( s, \bb_{0} \right) \right\} \\ 
        & \times \left\{ \int_{0}^{t_{2}} \left( \pi_{1}\left( \bz \right) - \mathscr{E}_{\pi} \left( s, \bb_{0} \right) \right) d M_{1} \left( s, \bb_{0} \right) - \boldsymbol{\nu}^{\top} \left( t_{2} \right) \boldsymbol{\Omega}^{-1} \int_{0}^{\infty} \xi \left( s \right) \left( Z_1 - \mathscr{E}\left( s, \bb_{0} \right) \right) d M_{1} \left( s, \bb_{0} \right) \right\} \Bigg].
    \end{align*}
    
    Second, we claim that the omnibus test based on $\sup_{t, z} \abs{\w \left( t, \bz \right)}$ is consistent against the general alternative that there do not exist a constant vector $\bb$ and a function $\lambda_{0} \left( \cdot \right)$ such that $\lambda(t | \bZ) = \lambda_{0} \left( t \cdot g \left( \bZ \right) \right) g \left( \bZ \right)$, generalized hazard function, for almost all $t > 0$ and $z$ generated by the random vector $Z$. 
    
    We first decompose $n^{-\frac{1}{2}}\w \left( t, \bz ; \bbhi \right)$ into two parts: 
        \begin{align} 
       	& n^{-\frac{1}{2}} \w \left( t, \bz ; \bbhi \right) \nonumber \\
       	& = n^{-1} \sum_{i=1}^{n} \pi_{i} \left( \bz \right) \widehat{M}_{i} \left( t, \bbhi \right) \nonumber \\
       	& = n^{-1} \sum_{i=1}^{n} \ell \left( \bZ_{i} \right) I \left( \bZ_{i} < z \right) N_{i} \left( t, \bbhi \right) - n^{-1} \sum_{i=1}^{n} \ell \left( \bZ_{i} \right) I \left( \bZ_{i} < z \right) \int_{0}^{t} Y_{i} \left(s, \bbhi \right) d\widehat{\Lambda}_{0} \left(s, \bbhi \right). \label{eq:app:2.0} 
        \end{align}
        Then, it follows from the Strong Law of Large Numbers (SLLN) and \citet[][Theorem 2; p.616]{lin1998accelerated}, \eqref{eq:app:2.0} is asymptotically equivalent to  
        \begin{equation}
            \nE \left[ \ell \left( \bZ_{i} \right) I \left( \bZ_{i} < z \right) N_{i} \left( t, \bbhi \right) \right] - \nE \left[ \ell \left( \bZ_{i} \right) I \left( \bZ_{i} < z \right) \int_{0}^{t} Y_{i} \left(s, \bbhi \right) d \widehat{\Lambda}_{0} \left(s, \bbhi \right) \right] \label{eq:app:2.1}
        \end{equation}
        Let $H \left( \bz \right)$ denote the distribution of $\bZ_{i}$. Then, the first term of \eqref{eq:app:2.1} can rewritten as
        \begin{align*}
            & \nE \left[ \ell \left( \bZ_{i} \right) I \left( \bZ_{i} < z \right) N_{i} \left( t, \bbhi \right) \right] \\
            & = \nE \left[ \nE \left[ \ell \left( \bz \right) I \left( \bz < z \right) N_{i} \left( t, \bbhi \right) | \bz \right] \right] \\
            & = \int_{0}^{z} \ell \left( \bz \right) \nE \left[ N_{i} \left( t, \bbhi \right) | \bz \right] d H \left( \bz \right) \\
            & = \int_{0}^{z} \int_{0}^{t} \ell \left( \bz \right) Y_{i} \left(s, \bbhi \right) \lambda_{0} \left( s \cdot g \left( \bz ; \bbhi \right) \right) g \left( \bz ; \bbhi \right) d H \left( \bz \right) ds 
        \end{align*}
        The last equality above comes from $\lambda(t | \bZ) = \lambda_{0} \left( t \cdot g \left( \bZ \right) \right) g \left( \bZ \right)$. Likewise, the second term of \eqref{eq:app:2.1} can be written as 
        \begin{align*}
        	& \nE \left[ \ell \left( \bZ_{i} \right) I \left( \bZ_{i} < z \right) \int_{0}^{t} Y_{i} \left(s, \bbhi \right) d \widehat{\Lambda}_{0} \left( s, \bbhi \right) \right] \\
        	& = \int_{0}^{z} \ell \left( \bz \right) \int_{0}^{t} Y_{i} \left( s,\bbhi \right) d \widehat{\Lambda} \left( s,\bbhi \right) d H\left( \bz \right) \\
        	& = \int_{0}^{z} \int_{0}^{t} \ell \left( \bz \right) Y_{i} \left( s,\bbhi \right) d \widehat{\Lambda} \left( s,\bbhi \right) d H\left( \bz \right) \\
        	& = \int_{0}^{z} \int_{0}^{t} \ell \left( \bz \right) Y_{i} \left( s,\bbhi \right) \widehat{\lambda}_{0} \left( s \cdot \exp \left( - \bz^{\top} \bbhi \right) \right) \exp \left( - \bz^{\top} \bbhi \right) d H \left( \bz \right) ds. 
        \end{align*}
        Combining these two results, \eqref{eq:app:2.1} reduces to
        \begin{align} \nonumber 
            & \int_{0}^{z} \int_{0}^{t} \ell \left( \bz \right) Y_{i} \left(s, \bbhi \right) \left\{ \lambda_{0} \left( s \cdot g \left( \bz ; \bbhi \right) \right) g \left( \bz ; \bbhi \right)- \widehat{\lambda}_{0} \left( s \cdot \exp \left( - \bz^{\top} \bbhi \right) \right) \exp \left( - \bz^{\top} \bbhi \right) \right\}  d H \left( \bz \right) d s \\
            & = \int_{0}^{z} \int_{0}^{t} \ell \left( \bz \right) Y_{i} \left( s,\bbhi \right) \exp \left( - \bz^{\top} \bbhi \right) \left\{ \frac{g \left( \bz ; \bbhi \right)}{\exp \left( - \bz^{\top} \bbhi \right)} - \frac{\widehat{\lambda}_{0} \left( s \cdot \exp \left( - \bz^{\top} \bbhi \right) \right)}{\lambda_{0} \left( s \cdot g \left( \bz ; \bbhi \right) \right)} \right\} \lambda_{0} \left( s \cdot g \left( \bz ; \bbhi \right) \right) d H \left( \bz \right) d s. \label{eq:app:2.2}
        \end{align}
        Under the alternative hypothesis, $\bbhi$ and $\widehat{\Lambda}_{0} \left( t, \bbhi \right)$ converge to $\bbstar$ and $\int \lambda_{0}^{\star} (s, \bbstar) ds$, respectively, and $g \left( Z ; \bbstar \right) \neq \exp \left( Z^{\top} \bbstar \right)$. Therefore, \eqref{eq:app:2.2} converges to
        \begin{align*}	
        	\int_{0}^{z} \int_{0}^{t} \ell \left( \bz \right) Y_{i} \left( s,\bbstar \right) \exp \left( - \bz^{\top} \bbstar \right) \left\{ \frac{g \left( \bz ; \bbstar \right)}{\exp \left( - \bz^{\top} \bbstar \right)} - \frac{\lambda_{0}^{\star} \left( s \cdot \exp \left( - \bz^{\top} \bbstar \right) \right)}{\lambda_{0} \left( s \cdot g \left( \bz ; \bbstar \right)  \right)} \right\} \lambda_{0} \left( s \cdot g \left( \bz ; \bbstar \right) \right) d H \left( \bz \right) d s
        \end{align*}
        	
        By \citet[][p. 136]{hardy1952inequalities} and in \citet[Appendix 3]{lin1993checking}, as $n \to \infty$,
        \begin{align*}
            \frac{\lambda_{0}^{\star} \left( s \cdot \exp \left( - \bz^{\top} \bbstar \right) \right)}{\lambda_{0} \left( s \cdot g \left( \bz ; \bbstar \right)  \right)}
            \conv{}
            \frac{ \int  \exp \left( -\bz^{\top} \bbstar \right) \left\{ \frac{g \left( \bz ; \bbstar \right)}{\exp \left( -\bz^{\top} \bbstar \right)} \right\} Y_{i} \left( t, \bbstar \right) dH \left( \bz \right) }{\int \exp \left( -\bz^{\top} \bbstar \right) Y_{i} \left( t, \bbstar \right) dH\left( \bz \right) }
        \end{align*}
        almost surely. For the maximizer of $ \frac{g \left( \bz ; \bbstar \right)}{\exp \left( -\bz^{\top} \bbstar \right)}$, say $\bz_{\dag}$, we have
        \begin{align*}
            \left\{ \frac{g \left( \bz_{\dag} ; \bbstar \right)}{\exp \left( - \bz_{\dag}^{\top} \bbstar \right)} \right\} -    	\frac{ \int \exp \left( -\bz^{\top} \bbstar \right) \left\{ \frac{g \left( \bz ; \bbstar \right)}{\exp \left( -\bz^{\top} \bbstar \right)} \right\} Y_{i} \left( t, \bbstar \right) dH \left( \bz \right) }{\int \exp \left( -\bz^{\top} \bbstar \right) Y_{i} \left( t, \bbstar \right) dH\left( \bz \right) } > 0.
        \end{align*}
        It shows that the test statistic is consistent against the general alternative as we stated above.  
            
\subsection{Proof of Theorem \ref{theorem2}} \label{subsec:prooftheorem2}
    To show that $\w \left( t, \bz ; \bbhi \right)$ is asymptotically identically distributed as the test statistic based on non-smoothed score process $\wh \left( t, \bz \right)$ where
	\begin{align*}
            \wh \left( t, \bz \right) 
            &=n^{-\frac{1}{2}} \widehat{\score}_{\pi}^{\phi} \left( t, \bz,  \bbhi \right) -n^{\frac{1}{2}} \left( \widehat{f}_{\pi} \left( t, \bz \right) + \int_{0}^{t} \widehat{g}_{\pi} \left( s, \bz\right) d \widehat{\Lambda}_{0} \left(s, \bbhi \right) \right)^{\top} \left( \bbhi - \bbhid \right) \\
            & \hspace{2em} - n^{-\frac{1}{2}} \int_{0}^{t} S_{\pi} \left( s, \bz, \bbhi \right) d \left( \widehat{\Lambda}_{0} \left(s, \bbhi \right) - \widehat{\Lambda}_{0} \left(s, \bbhid \right) \right).
	\end{align*}

    For simplicity, we assume that $\psi \left(s, \bb \right) = 1$. Then,
    \begin{align*}
        \wh \left( t, \bz \right) &=  \frac{1}{\sqrt{n}}  \sum_{i=1}^{n} \int_{0}^{t} \left( \pi_{i} \left( \bz \right) - E_{\pi} \left( s, \bz, \bbhi \right) \right) d \widehat{M}_{i} \left( s, \bbhi \right) \phi_{i}  \\
        & - n^{\frac{1}{2}} \left( \widehat{f}_{\pi} \left( t, \bz \right) + \int_{0}^{t} \widehat{g}_{\pi} \left( s, \bz\right) d \widehat{\Lambda}_{0} \left(s, \bbhi \right) \right)^{\top} \left( \bbhi - \bbhid \right)\\
        & - n^{\frac{1}{2}}\int_{0}^{t} n^{-1} S_{\pi} \left( s, \bz, \bbhi \right) d \boldsymbol{\kappa} \left( s \right) \left( \bbhi - \bbhid \right) + o_{p}(1) \\
        & = \frac{1}{\sqrt{n}} \sum_{i=1}^{n} \int_{0}^{t}  \left( \pi_{i} \left( \bz \right) - E_{\pi} \left( s, \bz, \bbhi \right) \right) d \widehat{M}_{i} \left( s, \bbhi \right) \phi_{i} \\ 
        & - n^{-\frac{1}{2}} \left( \widehat{f}_{\pi} \left( t, \bz \right) + \int_{0}^{t} \widehat{g}_{\pi} \left( s, \bz\right) d \widehat{\Lambda}_{0} \left(s, \bbhi \right) + \int_{0}^{t} n^{-1} S_{\pi} \left( s, \bz, \bbhi \right) d \boldsymbol{\kappa} \left( s \right) \right)^{\top} \boldsymbol{\Omega}^{-1} \scoreI \left( \bbhid \right) + o_{p}(1)
    \end{align*}

    Since $\scoreI^{\phi} \left( \bbhid \right) = 0$, we observe that $\wh \left( t, \bz \right)$ shares the same components as $\w \left( t, \bz \right)$. Note that that $\bb_{0}$, $M_{i} \left( t, \bb_{0} \right)$, $f_{\pi} \left( t, \bz \right)$, and $g_{\pi} \left( t, \bz \right)$ can be replaced by $\bbhi$, $\phi_{i} M_{i} \left( t, \bb_{0} \right)$, $\widehat{f}_{\pi} \left( t, \bz \right)$, and $\widehat{g}_{\pi} \left( t, \bz \right)$, respectively. In addition, the resampled martingale residuals $\phi_{i} M_{i} \left( t, \bb_{0} \right)$ have the same distribution as $\phi_{i} \widehat{M}_{i} \left( t, \bbhi \right)$, and the kernel estimates of $f_{0}$ and $g_{0}$ converge uniformly to their true densities. Consequently, $\wh \left( t, \bz \right)$ exhibits the same limiting finite-dimensional distributions as $\w \left( t, \bz \right)$, and its tightness follows the same arguments as for $\w \left( t, \bz \right)$.

\end{document}